\documentclass[11pt]{article}
 
\usepackage{amsmath}
\usepackage{amssymb}
\usepackage{comment}
\usepackage{cancel}
\usepackage{tikz}
\usepackage{caption}
\usepackage{url}
\usepackage{graphicx}
\usepackage{hyperref}
\usepackage{bm}
\usepackage{enumitem}
\usepackage{booktabs}
\usetikzlibrary{shapes.geometric, arrows.meta}
 
\usepackage[authoryear]{natbib}
 
\newtheorem{theorem}{Theorem}

\newtheorem{definition}{Definition}
\newtheorem{example}{Example}
\newtheorem{remark}{Remark}
\newtheorem{assumption}{Assumption}

\newtheorem{algorithm}{Algorithm}
 
\newcommand{\ind}{\perp \!\!\!\perp}
 
\newenvironment{proof}[1][Proof]{\noindent\textbf{#1.} }{\ \rule{0.5em}{0.5em}}
 
\begin{document}
 
\title{Exponentially weighted estimands and the exponential family: \\
Filtering, prediction and smoothing}

\author{Simon Donker van Heel\textsuperscript{b,c} and Neil Shephard\textsuperscript{a}}

\maketitle

\begin{center}
  \small
  \textsuperscript{a}\textit{Dept.\ of Economics and Dept.\ of Statistics, Harvard University, Cambridge, MA 02138, USA}\\
  \textsuperscript{b}\textit{Econometric Institute, Erasmus University Rotterdam, Rotterdam, The Netherlands} \\
  \textsuperscript{c}\textit{Tinbergen Institute, The Netherlands}
\end{center}

\maketitle

\begin{abstract}
We propose using a discounted version of a convex combination of the log-likelihood with the corresponding expected log-likelihood such that when they are maximized they yield a filter, predictor and smoother for time series.  This paper then focuses on working out the implications of this in the case of the canonical exponential family.  The results are simple exact filters, predictors and smoothers with linear recursions. A theory for these models is developed and the models are illustrated on simulated and real data.  
\end{abstract}
 
 
\baselineskip=20pt

\section{Introduction}

We provide a simple way to carry out filtering, predicting and smoothing.  It is based on a discounted version of a convex combination of a log-likelihood and the expected log-likelihood function. The discounting is carried out using an exponential function.

We focus on log-likelihoods based around the canonical exponential family.  It has a similar structure to generalized linear regression models, but now for time series.  Properties of the resulting procedures are established and illustrated by simulation and empirical work.    

Section \ref{sect:likelihood} develops the new filter, predictor and smoother.     Section \ref{sect:CEF} looks at the canonical exponential family case.   These are found by a convex numerical optimization.  

Section \ref{sect:specialcase} looks at a flexible special case of the canonical exponential family which results in analytic expressions for the filters, predictors and smoothers.  They do not need any form of iteration or approximation. This special case covers a vast number of familiar models, e.g. based on distributions such as, for example, the Gaussian, multinomial, Poisson, Pareto, exponential, von Mises and Dirichlet.

Section \ref{sect:quasi} develops a quasi-likelihood based estimation strategy for the hyperparameters of this structure.  This method is extended to allow for a two-step estimator, where an easy to use moment estimator reduces the dimension of the required numerical maximization.  

Section \ref{Sec: Empirical example} provides an empirical illustration based on a seven category Dirichlet model.  
Section \ref{sect:conc} details some conclusions.   

\section{Principle of weighting}\label{sect:likelihood}

Think of time series data $y_{1:T} = \{y_1,...,y_T\}$ and a time series model for the corresponding random variables $Y_{1:T} = \{Y_1,...,Y_T\}$.  Here $T$ is the length of the time series.   The properties of the time series model $Y_{1:T}$ are determined by the abstract probability triple $(\Omega,\mathcal{F},P)$.

\begin{definition} For the random variable $Y_t$ write a frame
$$
\log L_t(\theta; y_{t}), \quad \theta \in \Theta,y_t \in \mathcal{Y}_t,\quad t=1,...,T,
$$
for the specification of the filtering, smoothing and prediction estimands.  The frame is called stable if it is of the form $\log L(\theta;y_t)$, i.e. not varying over time.   
\end{definition}

At no point do we regard this frame as true.  It will index a class of filters, smoothers and predictors. The frame may depend on additional static parameters $\phi$, which we suppress notationally for clarity until Section \ref{sect:quasi}. The form of the frame can vary over time, e.g. it could be based around the random variable $Y_t \sim {\tt Multinomial}(n_t,p)$, where $p$ is determined by $\theta$ and $n_t$ is regarded as non-stochastic (e.g. we have formed the probability mass function for $Y_t$ by conditioning on $n_t$) or $Y_t \sim {\tt Pois}(e^{\theta^{\tt T} x_t})$ where $x_t$ is a non-stochastic vector (or conditioned on).  

At various points in the paper some assumptions are made about the form of the frame.\footnote{We follow the statistical convention of taking the log-likelihood as ignoring terms which only involve data not also hyperparameters.  One way of implementing this, for a model with a probability density function, for example, is to take $\log L_t(\theta;y_t) := \log f_t(y_t;\theta) - \log f_t(y_t;\theta_0)$ where $\theta_0 \in \Theta$ is some arbitrary fixed value as $\theta\in \Theta$ varies. As we need the expected log-likelihood to exist, removing these data terms lessens the requirements we need to prove the expected value exists.} Here we list some of them.   

\begin{assumption}\label{assum:start1} The $\mathbb{E}[\log L_t(\theta;Y_t)]$ exist for each $t$ and $\theta\in \Theta$.   
\end{assumption}

The expectation is under the probability measure for $Y_t$, not the probability model that would rationalize the frame.    

\begin{assumption}\label{assum:concave-coercive}
The parameter space is $\Theta = \mathbb{R}^d$. 
For each $t$ and all $y \in \mathcal{Y}_t$, the $\log L_t(\theta;y)$ is concave 
and upper semicontinuous in $\theta$, taking values in $\mathbb{R} \cup \{-\infty\}$. 
The $\mathbb{E}[\log L_t(\theta;Y_t)]$ is finite for all 
$\theta \in \mathbb{R}^d$, strictly concave in $\theta$, and satisfies the 
coercivity condition that  $\|\theta\|\to\infty$ implies that $
\mathbb{E}[\log L_t(\theta;Y_t)] \to -\infty.$
\end{assumption}

The paper will be based around an exponential version of a weighting  principle.\footnote{The extension to allow a prior density $\pi(\theta)$ for $\theta \in \Theta$, yielding an  exponentially weighted MAP (i.e. mode of a type of posterior), means just adding $\log \pi(\theta)$ to our criterion function and then maximizing that with respect to $\theta$ for each $t$.  This is particularly important if $\theta$ is high dimensional.} 

The principle of weighting delivers three sequences of estimands through time: a filter, a predictor and a smoother.   Notice that the estimands are random variables.  

\begin{definition}[Exponentially weighted estimands]\label{defn:Q} Let Assumptions \ref{assum:start1}--\ref{assum:concave-coercive} hold, $\lambda \in [0,1]$ be a discount parameter and $\alpha \in [0,1]$ an anchoring parameter.  The exponentially weighted filter at time $t$, using data up to time $t$ and the sequence of frames $\{\log L_t(\theta;y_t)\}_{t=1}^T$, is
\begin{equation*}
\widetilde{\theta }_{t}=\underset{\theta \in \Theta}{\arg }\max \ Q_{\lambda,t}(\theta),
\end{equation*}%
where 
$$
Q_{\lambda,t}(\theta) = \sum_{j=1}^{t} \lambda^{t-j} \left\{(1-\alpha) \mathbb{E}[\log L_j(\theta; Y_j)] + \alpha \log L_j(\theta; Y_{j})\right\}.
$$
The time $t$ exponentially weighted predictor based on data up to time $(t-s)$, for $s\ge 1$, is 
\begin{equation*}
\widetilde{\theta }_{t|t-s}=\underset{\theta \in \Theta}{\arg }\max \ Q_{\lambda,t|t-s}(\theta),
\end{equation*}
where 
$$
Q_{\lambda, t \mid t-s}(\theta)=(1-\alpha) \sum_{j=1}^t \lambda^{t-j} \mathbb{E}\left[\log L_j\left(\theta ; Y_j\right)\right]+\alpha \sum_{j=1}^{t-s} \lambda^{t-j} \log L_j\left(\theta ; Y_j\right).
$$
The exponentially weighted smoother at time $t$, using data up to time $T\ge t$, is 
$$
\widetilde{\theta}_{t|T} = \underset{\theta \in \Theta}{\arg }\max \ Q_{\lambda,t|T}(\theta),
$$
where 
$$
Q_{\lambda,t|T}(\theta) = \sum_{j=1}^{T} \lambda^{|t-j|} \left\{(1-\alpha) \mathbb{E}[\log L_j(\theta; Y_j)] + \alpha \log L_j(\theta; Y_{j})\right\}.
$$
\end{definition}

Note that Definition \ref{defn:Q}, Assumptions \ref{assum:start1}--\ref{assum:concave-coercive} and $\alpha \in [0,1)$ imply that $Q_{\lambda,t}(\theta)$, $Q_{\lambda,t|t-s}(\theta)$ and $Q_{\lambda,t|T}(\theta)$ admit  unique maximizers $\widetilde{\theta }_{t}$, $\widetilde{\theta }_{t|t-s}$ and $\widetilde{\theta }_{t|T}$, respectively.

As the likelihood function is invariant to reparameterization, like the maximum likelihood estimator (MLE), the exponentially weighted filter, predictor and smoother are invariant under one-to-one reparameterization, i.e. if $\gamma = \gamma(\theta)$, then $\widetilde{\gamma}_t = \gamma(\widetilde{\theta}_t).$ This will turn out to be a very powerful property in practice. 

Intuitively, $\lambda \in [0,1]$ controls the rate at which older information is  downweighted, while $\alpha \in [0,1]$ controls the weight placed on the anchor $\mathbb{E}[\log L_j(\theta; Y_j)]$ terms.  

\begin{remark}\label{remark 1} (a) When $\alpha=1$ the resulting exponentially weighted maximum likelihood estimator often appears under the label of local modeling or local regression, e.g. Chapter 8 of \cite{FanYao(05)} and \cite{FanHeckmanWand(95)}.  There the localization is typically carried out using a kernel weighting function; our use of the exponential function is a special case of this. The Gaussian likelihood for an unknown mean yields the celebrated {\tt EWMA} of \cite{Brown(56)}. 

(b) Weighting likelihood contributions has appeared in applied contexts, e.g.  \cite{dixon1997modelling} use exponentially weighted likelihoods over time to model football match outcomes, while \cite{hu2002weighted} and \cite{blasques2016weighted} 
weight over the cross-section.

(c) \cite{luxenberg2024exponentially} look at exponentially weighted moving losses, replacing the likelihood with a loss function.  Their work has some overlap with our approach, but focuses on penalization (which we do not), and does not consider the $\mathbb{E}[\log L_j(\theta,Y_j)]$ terms.  

(d) When $\alpha<1$ we do not know of a literature which uses this approach, including the works cited in parts (a)--(c).  Of course, broadly, the expected log-likelihood can be thought as a form of shrinkage.  Hence it relates to the vast literatures on empirical Bayes, ridge and Lasso regression, etc (e.g. \cite{EfronMorris(77)}, \cite{Efron(12)}, \cite{HoerlKennard(70a)}, \cite{Tibshirani(96)}).     

(e) There is a large general Bayesian literature which raises the likelihood to a power less than one to obtain some ``robustness'' properties for the posterior.  Some of the corresponding literature is discussed in \cite{HolmesWalker(17)}. 

(f) When $\alpha$ is close to one, the estimands are mostly determined by past observations; when $\alpha$ is close to zero, the estimands are predominantly determined by the expected log-likelihoods $\mathbb{E}[\log L_j(\theta,Y_j)]$. The discount parameter $\lambda \in [0,1]$ determines how quickly past observations are downweighted, with larger values giving more weight to distant observations. 

(g) In applications with naturally constrained parameters (e.g.\ $\sigma>0$), we work with a smooth one-to-one reparametrization such as a 
log or logit transform so that the resulting parameter vector $\theta$ lives in 
$\mathbb{R}^d$. Since the exponentially weighted estimands are 
reparametrization invariant, this involves no loss of generality and allows us 
to take $\Theta=\mathbb{R}^d$. Together with the upper semicontinuity of the log-frame and the coercivity and strict concavity of the expected 
log-frame, this guarantees that all $Q_{\lambda,t}$, $Q_{\lambda,t|t-s}$ and $Q_{\lambda,t|T}$ are 
coercive, concave and upper semicontinuous, and hence that their maximizers 
exist and are unique.

(h) Stochastic approximation methods can be used to approximate the solutions 
to the optimization problems in Definition~\ref{defn:Q} 
(e.g.\ \cite{RobbinsMunro(51)}, 
\cite{toulis2015scalable,toulis2017asymptotic}, 
\cite{toulis2016towards}, \cite{toulis2021proximal}).

\end{remark}

\section{Canonical exponential family case}\label{sect:CEF}

The rest of this paper will focus entirely on the case where the frame is based on a member of the canonical exponential family.  This family includes most of the important probability models used in applied statistics.  In our context, crucially, it allows either an analytic solution or an easy to compute solution for each $t$.   

\subsection{Canonical exponential family}

Start by recalling the definition of the canonical exponential family and three of its properties.  

\begin{definition}[Canonical exponential family]\label{ex:EF} Assume a random variable $Y\sim {\tt CEF}(\theta,h,\psi)$, meaning that its  probability density function (or probability mass function for a discrete random variable) is a member of the minimal canonical exponential family: 
$$f(y;\theta) = \exp\{\theta^{\tt T} h(y) - \psi(\theta)\}/b(y),\quad \theta\in \Theta=\{\theta:\psi(\theta)<\infty\}, \quad y \in \mathcal{Y},$$ where $\psi(\theta)$ is  infinitely continuously differentiable and strictly convex with respect to $\theta$ while $b(y)>0$ for all $y \in \mathcal{Y}$, (e.g. \cite{McCullaghNelder(89)} and \cite{Efron(22)}). 
\end{definition}

The cumulant function of the {\tt CEF} is $\log \mathbb{E}[e^{a^{\tt T}h(Y)};\theta] = \psi(\theta+a)-\psi(\theta)$, so   
\begin{align*}
\psi'(\theta)&=\frac{\partial \psi(\theta)}{\partial \theta} 
 = \mathbb{E}[h(Y);\theta] 
 = \mu,
\end{align*} 
the finite expected value of $h(Y)$ under the ${\tt CEF}(\theta,h,\psi)$ model and 
\begin{align*}
\psi''(\theta)&=\frac{\partial^2 \psi(\theta)}{\partial \theta \partial \theta^{\tt T}} ={\mathrm V}(h(Y);\theta),
\end{align*}
the corresponding finite variance. The minimal {\tt CEF} requires $\psi(\theta)$ to be strictly convex so ${\mathrm V}(h(Y);\theta)$ is positive definite under the model and Assumption~\ref{assum:concave-coercive} is satisfied. It is helpful to go from $\mu = \mathbb{E}[h(Y)]$ to $\theta$ through the inverse function
$
\psi'^{-1}(\mu) = \theta.
$
This inverse has a unique solution (as $\psi$ is strictly convex), but in some cases it has to be solved numerically; a classic task in statistics as this is isomorphic to computing the MLE of $\theta$ in the {\tt CEF}.  Famously, it can be carried out reliably and rapidly using a Newton-Raphson algorithm, which converges quadratically close to the solution.  We state this formally as Algorithm \ref{alg:NR} for reference later.  
\begin{algorithm}[Solving $\theta = \psi'^{-1}(\mu)$]\label{alg:NR}  Starting from some point $\theta^{(0)}$, iterate
$$
\theta^{(i)} = \theta^{(i-1)} - \left\{\psi''(\theta^{(i-1)})\right\}^{-1}\{\psi'(\theta^{(i-1)}) - \mu\},
$$
for $i=1,2,...$, until convergence, which is the solution $\theta$.  Recall here $\mu = \mathbb{E}[h(Y)]$.
\end{algorithm}

There is a substantial literature on time series models based around the exponential family.  Examples include \cite{ZegerQaqish(88)}, \cite{Li(94)}, \cite{BenjaminRigbyStasinopoulos(03)} and the review by \cite{Davis03072021}.  The dynamic conditional score (DCS) filter and the generalized autoregressive score (GAS) filter  of \cite{Harvey(13)} and \cite{CrealKoopmanLucas(13)} can be applied in this context to deliver recursive filters.  Our predictive version is closest to the generalized {\tt ARMA} work of \cite{BenjaminRigbyStasinopoulos(03)}, which is built around the exponential family and generalized linear models.  

\subsection{Filtering, prediction and smoothing}

In what follows throughout we take the frame $\log L_t(\theta;y) = \log f_t(y;\theta)$ as coming from ${\tt CEF}(\theta,h_t,\psi_t)$. Then, for each $t$, (ignoring the implied $b_t(y)$ term as it has no impact on $\theta$ and dropping it means we do not need to make an assumption that $\mathbb{E}[\log b_t(Y_t)]$ exists)  
$$
(1-\alpha) \mathbb{E}[\log L_t(\theta;Y_t)] + \alpha \log L_t(\theta;Y_t) = \theta^{\tt T} \{(1-\alpha)\mathbb{E}[h_t(Y_t)] + \alpha h_t(Y_t)\} - \psi_t(\theta).
$$

In the case where the data is assumed to be strictly stationary, the frame is stable and $\mathbb{E}[h(Y_t)]$ exists, then the right hand side of this expression simplifies to 
$$\theta^{\tt T} \{(1-\alpha)\mathbb{E}[h(Y_1)] + \alpha h(Y_t)\} - \psi(\theta).$$  

To compactly write the exponentially weighted filter it is helpful to denote the exponentially weighted moving sums
$$
x_{\lambda,t} = \sum_{j=1}^t \lambda^{t-j}\mathbb{E}[h_j(Y_j)], \quad
h_{\lambda,t} = \sum_{j=1}^t \lambda^{t-j}h_j(Y_j),\\
$$
as well as the double sided exponentially weighted moving sums
$$
x_{\lambda,t|T} = \sum_{j=1}^T \lambda^{|t-j|}\mathbb{E}[h_j(Y_j)], \quad 
h_{\lambda,t|T} = \sum_{j=1}^T \lambda^{|t-j|}h_j(Y_j).
$$
Each of these terms can be computed recursively, implying each entire series, e.g. $h_{\lambda,1|T},...,h_{\lambda,T|T}$, can be computed in $O(\sum_{t=1}^T \dim(h_t(y_t)))$ computations.     

\begin{theorem}\label{thrm1}  Assume a minimal {\tt CEF} frame ${\tt CEF}(\theta,h_t,\psi_t)$ for each $t$.  Then the exponentially weighted filter ($\widetilde{\theta}_{t}$), predictor ($\widetilde{\theta}_{t|t-s}$ with $s \ge 1$) and smoother ($\widetilde{\theta}_{t|T}$) solve: 
$$
x_{\lambda,t} + \alpha (h_{\lambda,t}-x_{\lambda,t}) = \sum_{j=1}^t \lambda^{t-j} \psi_j'(\widetilde{\theta}_{t}),
$$
$$
(1-\alpha) x_{\lambda, t}+\alpha \lambda^s h_{\lambda, t-s}=(1-\alpha) \sum_{j=1}^t \lambda^{t-j} \psi_j^{\prime}\left(\widetilde{\theta}_{t \mid t-s}\right)+\alpha \sum_{j=1}^{t-s} \lambda^{t-j} \psi_j^{\prime}\left(\widetilde{\theta}_{t \mid t-s}\right),
$$
$$
x_{\lambda,t|T} + \alpha (h_{\lambda,t|T}-x_{\lambda,t|T}) = \sum_{j=1}^T \lambda^{|t-j|} \psi_j'(\widetilde{\theta}_{t|T}),
$$
respectively.  Then write the filtered, predicted and smoothed mean and variance as 
$$
\widetilde{\mu}_{t} := \psi_t'(\widetilde{\theta}_{t}),\quad 
\widetilde{\mu}_{t|t-s} := \psi_t'(\widetilde{\theta}_{t|t-s}),\quad 
\widetilde{\mu}_{t|T} := \psi_t'(\widetilde{\theta}_{t|T}),\quad 
$$
and
$$
\widetilde{\Sigma}_t := \psi_t''(\widetilde{\theta}_t),
\quad 
\widetilde{\Sigma}_{t|t-s} := \psi_t''(\widetilde{\theta}_{t|t-s}),
\quad 
\widetilde{\Sigma}_{t|T} := \psi_t''(\widetilde{\theta}_{t|T}).
$$   

\end{theorem}
\begin{proof} Focus on the filtering case.  
Now 
$$
Q_{\lambda,t}(\theta) = \theta^{\tt T} \{(1-\alpha)x_{\lambda,t} +\alpha h_{\lambda,t}\} - \sum_{j=1}^t \lambda^{t-j} \psi_j(\theta).
$$
Differentiate $Q_{\lambda,t}(\theta)$ with respect to $\theta$ and solve.  This yields the stated result. 
The same argument applies to prediction and smoothing.  
The remaining results follow by invariance. \end{proof}

All the left hand side terms (e.g. $x_{\lambda,t}$ and $h_{\lambda,t}$) can be computed recursively.  The right hand side, not so much, as they depend upon the value of $\theta$ that the $\psi_j$ function is evaluated at.  Typically $\widetilde{\theta}_{t}$, $\widetilde{\theta}_{t|t-s}$ and $\widetilde{\theta}_{t|T}$ have to be found by numerical root solving.  
Due to the strict convexity of $\psi$ this is numerically straightforward, for each individual value of $t$, using Newton-Raphson, e.g. Algorithm \ref{alg:NRfilter} computes $\widetilde{\theta}_t$, the exponentially weighted filter.  
\begin{algorithm}[Computing $\widetilde{\theta}_t$]\label{alg:NRfilter} Starting from some point $\widetilde{\theta}_t^{(0)}$, iterate
$$
\widetilde{\theta}_t^{(i)} = \widetilde{\theta}_t^{(i-1)} - \left\{\sum_{j=1}^t \lambda^{t-j} \psi_j''(\widetilde{\theta}_t^{(i-1)})\right\}^{-1}\left\{\sum_{j=1}^t \lambda^{t-j} \psi_j'(\widetilde{\theta}_t^{(i-1)}) - \{x_{\lambda,t} + \alpha (h_{\lambda,t}-x_{\lambda,t})\}\right\},
$$
for $i=1,2,...$, until convergence, which is the solution $\widetilde{\theta}_t$.  
\end{algorithm}  
A downside is that the numerical procedure has to be run separately for each $t$, so computing, for example, the time series $\widetilde{\theta}_{1},...,\widetilde{\theta}_{T}$ costs $O(\{\sum_{t=1}^T \dim(h_t(y_t))^3\}^2)$ calculations.  

\subsection{A ${\tt CEF}(\theta,h_t,\psi_t)$ data generating process}

We will use the following structure in our simulations later.  

\begin{assumption}[Model based data generating process]\label{def:DGP4} (a) Use a frame based on ${\tt CEF}(\theta,h_t,\psi_t)$, computing the sequence $\widetilde{\theta }_{t|t-1}$ for $t=1,2,...,T$, recursively through Definition \ref{defn:Q}.  (b) Generate the data as  
$$
Y_t | Y_{1:t-1} \sim {\tt CEF}(\widetilde{\theta }_{t|t-1},h_t,\psi_t),\quad t=1,2,...,T.
$$
The stable frame version of this is where we set the frame to have $h_t(y)=h(y)$ and $\psi_t=\psi$ for $t=1,...,T$. 
\end{assumption}

The use of the frame to derive a filter, which is then used as an input into a data generating process echos the DCS/GAS models of \cite{Harvey(13)} and \cite{CrealKoopmanLucas(13)}.

Some examples of these types of simulations will be given shortly.  

\section{Special case with analytic solution: $\psi_t=n_t \psi$}\label{sect:specialcase}

The following special case has an analytic solution. We focus on it in the rest of this paper.  

\begin{example}\label{thm:MLE} Use a minimal ${\tt CEF}(\theta,h_t,n_t\psi)$ frame for each $t$ where $n_t$ is a non-stochastic scalar.  Then the exponentially weighted filter ($\widetilde{\theta}_{t}$), predictor ($\widetilde{\theta}_{t|t-s}$ with $s \ge 1$) and smoother ($\widetilde{\theta}_{t|T}$) are:  
$$
\widetilde{\theta}_t =  \psi'^{-1}\left(\bar{m}_{t}\right),\quad 
\widetilde{\theta}_{t|t-s} =  \psi'^{-1}\left(\bar{m}_{t|t-s}\right), \quad\widetilde{\theta}_{t|T} =  \psi'^{-1}\left(\bar{m}_{t|T}\right),
$$
respectively, where 
$$
\bar{m}_{t}= \frac{m_{\lambda,t}}{n_{\lambda,t}},\quad
\bar{m}_{t|t-s}= \frac{m_{\lambda,t|t-s}}{n_{\lambda,t|t-s}},\quad
\bar{m}_{t|T}= \frac{m_{\lambda,t|T}}{n_{\lambda,t|T}},\quad
$$
with 
\begin{align*}
m_{\lambda,t} &= (1 - \alpha) x_{\lambda,t} + \alpha h_{\lambda,t},
& n_{\lambda,t} &= \sum_{j=1}^t \lambda^{t-j} n_j, \\
m_{\lambda, t \mid t-s} & =(1-\alpha) x_{\lambda, t}+\alpha \lambda^s h_{\lambda, t-s} , &
n_{\lambda, t \mid t-s} & =(1-\alpha) n_{\lambda, t}+\alpha \lambda^s n_{\lambda, t-s}, \\ 
m_{\lambda,t|T} &= (1 - \alpha) x_{\lambda,t|T} + \alpha h_{\lambda,t|T}, &
n_{\lambda,t|T} &= \sum_{j=1}^T \lambda^{|t-j|} n_j.
\end{align*}
Then 
$$
\widetilde{\mu}_t = n_t \bar{m}_t,
\quad 
\widetilde{\mu}_{t|t-s} = n_t \bar{m}_{t|t-s},
\quad 
\widetilde{\mu}_{t|T} = n_t \bar{m}_{t|T},
$$
and 
$$
\widetilde{\Sigma}_t = n_t \psi''(\widetilde{\theta}_t),
\quad 
\widetilde{\Sigma}_{t|t-s} = n_t \psi''(\widetilde{\theta}_{t|t-s}),
\quad 
\widetilde{\Sigma}_{t|T} = n_t \psi''(\widetilde{\theta}_{t|T}).
$$    
\end{example}

This Example covers many interesting models.

\begin{remark}\label{remark:SS CEF} (a) In the special case of a stable frame {\tt CEF} and strictly stationary process then 
$x_{\lambda,t} = n_{\lambda,t} \mathbb{E}[h(Y_1)]$ and $ x_{\lambda,t|T} = n_{\lambda,t|T} \mathbb{E}[h(Y_1)]$.

(b) For filtering, the $m_{\lambda,t}$ 
is a convex combination of {\tt EWMA}s of $\mathbb{E}[h(Y_1)],...,\mathbb{E}[h(Y_t)]$ and of $h(Y_1),...,h(Y_t)$.  For smoothing, the $m_{\lambda,t|T}$
is a convex combination of the double sided {\tt EWMA}s of $\mathbb{E}[h(Y_1)],...,\mathbb{E}[h(Y_T)]$ and of $h(Y_1),...,h(Y_T)$ for time $t$.  

(c) If $\lambda \in (0,1)$, and $n_t=1$ for all $t$, then $n_{\lambda,t} \rightarrow 1/(1-\lambda)$ as $t\rightarrow \infty.$ If $\lambda=1$, then $n_{\lambda,t}=t$, while if $\lambda=0$, then $n_{\lambda,t}=1$. Likewise $n_{\lambda,t|T} \rightarrow (1+\lambda)/(1-\lambda)$ as $t\rightarrow \infty.$ If $\lambda=1$, then $n_{\lambda,t|T}=T$, while if $\lambda=0$, then $n_{\lambda,t|T}=1$. 

(d) Famously the {\tt EWMA} can be computed recursively.  In our case it follows 
$$
n_{\lambda,t} = n_t + \lambda n_{\lambda,t-1}, \quad 
x_{\lambda,t} = \mathbb{E}[h(Y_t)] + \lambda x_{\lambda,t-1}, \quad 
h_{\lambda,t} = h(Y_t) + \lambda h_{\lambda,t-1},
$$
initialized at $n_{\lambda,0}, x_{\lambda,0}, h_{\lambda,0} :=0$. Likewise, the double sided {\tt EWMA} can be computed recursively
\begin{align*}
h_{\lambda,t|T} &= h_{\lambda,t} + \lambda(h_{\lambda,t+1|T}-\lambda h_{\lambda,t}),\quad h_{\lambda,T|T}=h_{\lambda,T} \\
n_{\lambda,t|T} &= n_{\lambda,t} + \lambda(n_{\lambda,t+1|T}-\lambda n_{\lambda,t}),\quad n_{\lambda,T|T} = n_{\lambda,T},
\end{align*} 
going backwards, using the output from the forward pass of the {\tt EWMA}.

(e) The $m_{\lambda,t}$ and $m_{\lambda,t|t-1}$ can also be written recursively:
\begin{align*}
m_{\lambda,t} &= \mathbb{E}[h_t(Y_t)] + \alpha \{h_t(Y_{t})-\mathbb{E}[h_{t}(Y_{t})]\} + 
\lambda m_{\lambda,t-1}, \\
m_{\lambda,t|t-1} &= (1-\alpha) \mathbb{E}[h_t(Y_t)] + \alpha \lambda h_{t-1}(Y_{t-1}) + 
\lambda m_{\lambda,t-1|t-2}.
\end{align*}
In the steady state stable frame case with $n_t=1$, the $\bar{m}_t=\widetilde{\mu}_t$ and $\bar{m}_{t|t-1}=\widetilde{\mu}_{t|t-1}$ become
\begin{align*}
\widetilde{\mu}_{t} &= (1-\alpha)(1-\lambda)\mathbb{E}[h(Y_1)] + \alpha (1-\lambda)h(Y_{t}) + 
\lambda \widetilde{\mu}_{t-1}, \\
\widetilde{\mu}_{t|t-1}&= \frac{(1-\alpha)(1-\lambda)}{1-\alpha(1-\lambda)} \mathbb{E}[h(Y_1)] + \frac{\alpha \lambda(1-\lambda)}{1-\alpha(1-\lambda)}h(Y_{t-1})+ 
\lambda \widetilde{\mu}_{t-1|t-2}. 
\end{align*}
The $\widetilde{\mu}_{t|t-1}$ relates to the generalized {\tt ARMA} model of \cite{BenjaminRigbyStasinopoulos(03)} who model a link function of the data as being linear in past conditional means and link functions of the data.

Define $U_t := h(Y_t) - \widetilde{\mu}_{t|t-1}$, then 
\begin{align*}
h(Y_t) &= \widetilde{\mu}_{t|t-1} + U_t \\
&= \frac{(1-\alpha)(1-\lambda)}{1-\alpha(1-\lambda)} \mathbb{E}[h(Y_1)] + \frac{\alpha\lambda(1-\lambda)}{1-\alpha(1-\lambda)}h(Y_{t-1})+ 
\lambda \widetilde{\mu}_{t-1|t-2} + U_t \\
&= \frac{(1-\alpha)(1-\lambda)}{1-\alpha(1-\lambda)} \mathbb{E}[h(Y_1)] + \frac{\alpha\lambda(1-\lambda)}{1-\alpha(1-\lambda)}h(Y_{t-1})+ 
\lambda \{h(Y_{t-1}) - U_{t-1}\} + U_t \\
&= \frac{(1-\alpha)(1-\lambda)}{1-\alpha(1-\lambda)} \mathbb{E}[h(Y_1)] + \frac{\lambda}{1-\alpha(1-\lambda)}h(Y_{t-1}) 
 + U_t -\lambda U_{t-1}. 
\end{align*}

If the data has the property that $\mathbb{E}[h(Y_t)|Y_{1:t-1}]=\widetilde{\mu}_{t|t-1}$ and $\mathbb{E}[|Y_t|]<\infty$, then $\{U_t\}_{t=1}^T$ is a martingale difference (MD) sequence with respect to the filtration generated by the data.  The steady state process is a vector {\tt ARMA}(1,1)-MD process, with the autoregressive root being 
$\lambda/(1-\alpha(1-\lambda)) \in [\lambda,1)$, assuming $\lambda,\alpha \in [0,1)^2,$ with the $-\lambda$ moving average root.  
For example, if $\lambda=0.93$ then the {\tt AR}(1) root is roughly 0.978 and 0.996 when $\alpha =0.7$ and $\alpha=0.95$, respectively.  Hence the process can have substantial memory, although individual autocorrelations can be modest due to near root cancellation.  If additionally ${\mathrm V}(h(Y_t))$ exists and is time invariant, then in steady state $\{U_t\}$ is weak white noise (WN) and $\{h(Y_t)\}$ is a vector {\tt ARMA}(1,1)-WN covariance stationary process.   

(f) For prediction, $m_{\lambda,t|t-s}$ places non-negative weight $(1-\alpha)\lambda^{t-j}$ on each $\mathbb{E}[h_j(Y_j)]$ for $j=1,\ldots,t$, and non-negative weight $\alpha\lambda^{t-j}$ on each $h_j(Y_j)$ for $j=1,\ldots,t-s$. These weights sum to $n_{\lambda,t|t-s}$, so $\bar{m}_{\lambda,t|t-s} = m_{\lambda,t|t-s}/n_{\lambda,t|t-s}$ has non-negative weights summing to one; likewise for $\bar{m}_{\lambda,t|t}$ and $\bar{m}_{\lambda,t|T}$. Hence Jensen's inequality applies, e.g.\ for any convex $\varphi$,
$$
\varphi(\bar{m}_{\lambda,t|t-s}(h_{1:t})) \le \bar{m}_{\lambda,t|t-s}(\varphi(h_1),\ldots,\varphi(h_t)),
$$
where $\bar{m}_{\lambda,t|t-s}(h_{1:t})$ denotes the dependence on $h_1(Y_1),\ldots,h_t(Y_t)$. The right-hand side applies the linear operator $\bar{m}_{\lambda,t|t-s}$ to the component-wise transformed vector $(\varphi(h_1),\ldots,\varphi(h_t))$ with $\varphi$ being any convex function. A simple version of this is where $h(y)=y$ and $\varphi(h(y))=y^2$.

\end{remark}

\subsection{Simulating ten CEF examples: design}\label{sec:simulating CEF examples}

We illustrate the exponentially weighted predictor on ten {\tt CEF} examples. Given $\widetilde{\mu}_{t|t-1}$, each yields a unique solution for $\widetilde{\theta}_{t|t-1} = (\psi')^{-1}(\widetilde{\mu}_{t|t-1})$. In cases 1--6 and 9, this can be solved analytically; in cases 7, 8, and 10, numerically via Algorithm~\ref{alg:NR}. Tables~\ref{tab:CEF1}--~\ref{tab:CEF2} summarize the canonical parameterizations.  

\begin{table}[h!]
\caption{Canonical exponential family parameterizations: uni-state cases.}
\label{tab:CEF1}
\small
\begin{tabular*}{\textwidth}{@{\extracolsep{\fill}}lllll}
\toprule
Distribution & $h(y)$ & $\psi(\theta)$ & $\psi'(\theta)$ & $(\psi')^{-1}(\mu)$ \\
\midrule
1. Bernoulli, $y \in \{0,1\}$ & $y$ & $\log(1+e^\theta)$ & $\frac{e^\theta}{1+e^\theta}$ & $\log\frac{\mu}{1-\mu}$ \\[6pt]
2. Gaussian (known $\sigma$), $y \in \mathbb{R}$ & $y$ & $\frac{\theta^2}{2\sigma^2}$ & $\frac{\theta}{\sigma^2}$ & $\sigma^2\mu$ \\[6pt]
3. Poisson, $y \in \{0,1,2,\ldots\}$ & $y$ & $e^\theta$ & $e^\theta$ & $\log(\mu)$ \\[6pt]
4. Exponential, $y \in \mathbb{R}_{>0}$ & $y$ & $-\log(-\theta)$ & $-\frac{1}{\theta}$ & $-\frac{1}{\mu}$ \\[6pt]
5. Gaussian (zero mean), $y \in \mathbb{R}$ & $y^2$ & $-\frac{1}{2}\log(-2\theta)$ & $-\frac{1}{2\theta}$ & $-\frac{1}{2\mu}$ \\[6pt]
6. Pareto, $y \geq m > 0$ & $\log(y)$ & $-\log(-\theta)+\theta\log m$ & $-\frac{1}{\theta}+\log m$ & $-\frac{1}{\mu - \log m}$ \\
\bottomrule
\end{tabular*}\\[2pt]
{\footnotesize \emph{Note}: Columns show the components of the ${\tt CEF}(\theta,h,\psi)$ density from Definition~\ref{ex:EF}.}
\end{table}
 
\begin{table}[h!]
\caption{Canonical exponential family parameterizations: multi-state cases.}
\label{tab:CEF2}
\small
\renewcommand{\arraystretch}{1.8}
\begin{tabular*}{\textwidth}{@{\extracolsep{\fill}}lllll}
\toprule
\noalign{\vspace{-9pt}}
Distribution & $h(y)$ & $\psi(\theta)$ & $\psi'(\theta)$ & $(\psi')^{-1}(\mu)$ \\[-4pt]
\midrule
\begin{tabular}{@{}l@{}}7. Beta, \\[-12pt] $y \in [0,1]$\end{tabular} & 
$\left(\begin{array}{c} \log(y) \\ \log(1-y) \end{array}\right)$ & 
$\log B(\theta_1,\theta_2)$ & 
$\left(\begin{array}{c} \text{Dig}(\theta_1) - \text{Dig}(\theta_1+\theta_2) \\ \text{Dig}(\theta_2) - \text{Dig}(\theta_1+\theta_2) \end{array}\right)$ & 
Alg.~\ref{alg:NR} \\
\begin{tabular}{@{}l@{}}8. Dirichlet, \\[-12pt] $y \in \Delta_{d-1}$\end{tabular} & 
$\left(\begin{array}{c} \log(y_1) \\ \vdots \\ \log(y_d) \end{array}\right)$ & 
$\log B(\theta_1,\ldots,\theta_d)$ & 
$\left(\begin{array}{c} \text{Dig}(\theta_1) - \text{Dig}(\sum_j\theta_j) \\ \vdots \\ \text{Dig}(\theta_d) - \text{Dig}(\sum_j\theta_j) \end{array}\right)$ & 
Alg.~\ref{alg:NR} \\
\begin{tabular}{@{}l@{}}9. Gaussian, \\[-12pt] $y \in \mathbb{R}$\end{tabular} & 
$\left(\begin{array}{c} y \\ y^2 \end{array}\right)$ & 
$-\frac{\theta_1^2}{4\theta_2} - \frac{1}{2}\log(-2\theta_2)$ & 
$\left(\begin{array}{c} -\frac{\theta_1}{2\theta_2} \\[4pt] \frac{\theta_1^2}{4\theta_2^2} - \frac{1}{2\theta_2} \end{array}\right)$ & \hspace{-7mm}
$\left(\begin{array}{c} \frac{\mu_1}{\mu_2-\mu_1^2} \\[4pt] -\frac{1}{2(\mu_2-\mu_1^2)} \end{array}\right)$ \\
\begin{tabular}{@{}l@{}}10. von Mises, \\[-12pt] $y \in [0,2\pi]$\end{tabular} & 
$\left(\begin{array}{c} \sin(y) \\ \cos(y) \end{array}\right)$ & 
$\log I_0(\sqrt{\theta_1^2+\theta_2^2})$ & 
$\frac{I_1\left(\sqrt{\theta_1^2+\theta_2^2}\right)}{I_0\left(\sqrt{\theta_1^2+\theta_2^2}\right)}\left(\begin{array}{c} \frac{\theta_1}{\sqrt{\theta_1^2+\theta_2^2}} \\[4pt] \frac{\theta_2}{\sqrt{\theta_1^2+\theta_2^2}} \end{array}\right)$ & 
Alg.~\ref{alg:NR} \\
\bottomrule
\end{tabular*}\\[2pt]
{\footnotesize \emph{Note}: Columns are as in Table~\ref{tab:CEF1}. $B(\cdot)$ denotes the beta function, $\text{Dig}(\cdot)$ the digamma function, $I_0$ the modified Bessel function of the 1st kind and $I_1$ the modified Bessel function of the 1st kind of order 1, e.g. \cite{AbramowitzStegun(70)}. For cases 7, 8 and 10, $(\psi')^{-1}(\mu)$ must be computed numerically, e.g.\ via Algorithm~\ref{alg:NR}; convergence typically requires fewer than 5 iterations.}
\end{table}

For case 2, the resulting predictor $\widetilde{\mu}_{t|t-1}$ is the stationary version of the {\tt ARMA}(1,1) process; we compare the filter with the Kalman filter in Section~\ref{sect:kalman}. For case 5, note that $\mu = \mathbb{E}[Y^2] = \mathrm{V}(Y)$, and the resulting predictor $\widetilde{\mu}_{t|t-1}$ is the stationary version of the {\tt GARCH}(1,1) model of \cite{Engle(82)} and \cite{Bollerslev(86)}. For case 9, $\widetilde{\sigma}^2_{t|t-1} = \widetilde{\mu}_{2,t|t-1} - \widetilde{\mu}^2_{1,t|t-1} \geq 0$ by invariance and Jensen's inequality. For the cases 7, 8, and 10, given $\mu$, the corresponding $\theta = (\psi')^{-1}(\mu)$, which is unique, can be found numerically via Algorithm~\ref{alg:NR}.

Following Assumption \ref{def:DGP4}, we simulate each process using $Y_t | Y_{1:t-1} \sim {\tt CEF}(\widetilde{\theta}_{t|t-1},h,\psi)$ with 
$$
\widetilde{\theta}_{t|t-1} = (\psi')^{-1}(\bar{m}_{t|t-1}) = (\psi')^{-1}\left( \frac{m_{\lambda,t|t-1}}{n_{\lambda,t|t-1}} \right),
$$
where $m_{\lambda,t|t-1} = (1-\alpha)\mathbb{E}[h(Y_1)] + \alpha\lambda h(Y_{t-1}) + \lambda m_{\lambda,t-1|t-2}$. We impose stability by centering at $\mathbb{E}[h(Y_1)] = \mathbb{E}[Y_1] = 0.5$ (Bernoulli), $\mathbb{E}[h(Y_1)] = \mathbb{E}[Y_1] = 0$ (Gaussian with known $\sigma^2=1$), $\mathbb{E}[h(Y_1)] = \mathbb{E}[Y_1] = 1$ (Poisson, Exponential), $\mathbb{E}[h(Y_1)] = \mathbb{E}[Y_1^2] = 1$ (Gaussian with zero mean), and $\mathbb{E}[h(Y_1)] = \mathbb{E}[\log(Y_1)] = 1/3$ (Pareto with scale $m=1$, and shape $\theta_0 = -3$). For the bivariate cases, we take $\mathbb{E}[h(Y_1)] = \psi'(\theta_0)$ with initial values $\theta_0 = (2, 5)'$ for Beta (corresponding to ${\tt Beta}(2,5)$), $\theta_0 = (0, -1/2)'$ for Gaussian (corresponding to mean $0$, variance $1$), and $\theta_0 = (0, -2)'$ for von Mises (corresponding to mean direction $\pi$, concentration $2$).

\subsection{Simulating ten {\tt CEF} examples: results}

Figures~\ref{fig:CEFunistateExact} and~\ref{fig:CEFunistate_additional} show simulated time series $Y_{t}$ and conditional predictors for $t=5,...,T=2000$ with $\lambda=0.93$, $\alpha \in \{0.70, 0.95\}$, and $\widetilde{\theta}_{1|0} = (\psi')^{-1}(\mathbb{E}[h(Y_1)])$. Figure~\ref{fig:CEFunistateExact} displays the conditional mean $\widetilde{\mathbb{E}[Y_t|Y_{1:t-1}]} = \psi'(\widetilde{\theta}_{t|t-1})$ for Bernoulli, Gaussian (known $\sigma^2=1$), and Poisson. Integer-valued observations are jittered with $\text{Unif}(-0.1, 0.1)$ noise for visualization. 
Figure~\ref{fig:CEFunistate_additional} shows the exponential distribution with conditional mean $\widetilde{\mathbb{E}[Y_t|Y_{1:t-1}]} = \psi'(\widetilde{\theta}_{t|t-1})$, the Gaussian (zero mean) with conditional standard deviation $\widetilde{\sigma}_{t|t-1} = \sqrt{\psi'(\widetilde{\theta}_{t|t-1})}$, and the Pareto with conditional mean $\widetilde{\mathbb{E}[Y_t|Y_{1:t-1}]} = \frac{\widetilde{\theta}_{t|t-1}}{\widetilde{\theta}_{t|t-1}+1}$ when $\widetilde{\theta}_{t|t-1} < -1$ and $\infty$ otherwise.

\begin{figure}[htbp]
    \centering
    \begin{tabular}{@{}c@{\hspace{-0.25cm}}c@{\hspace{-0.25cm}}c@{\hspace{-0.25cm}}c@{}}
        & \textbf{Bernoulli} & \textbf{Gaussian (known sd)} & \textbf{Poisson} \\
        \raisebox{2.0cm}[0pt][0pt]{{\small $\alpha=0.70$\hspace{0.3cm}}} &
        \includegraphics[width=0.32\linewidth]{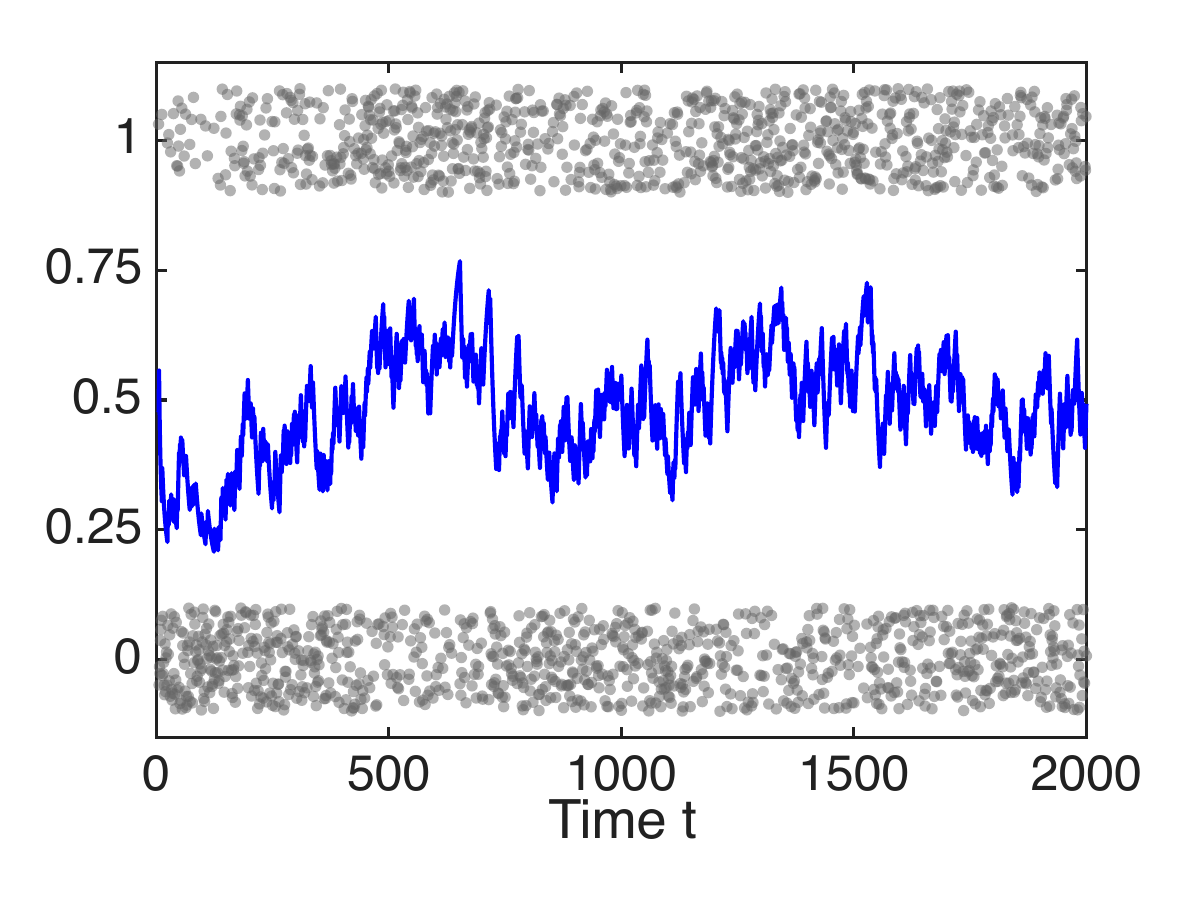} &
        \includegraphics[width=0.32\linewidth]{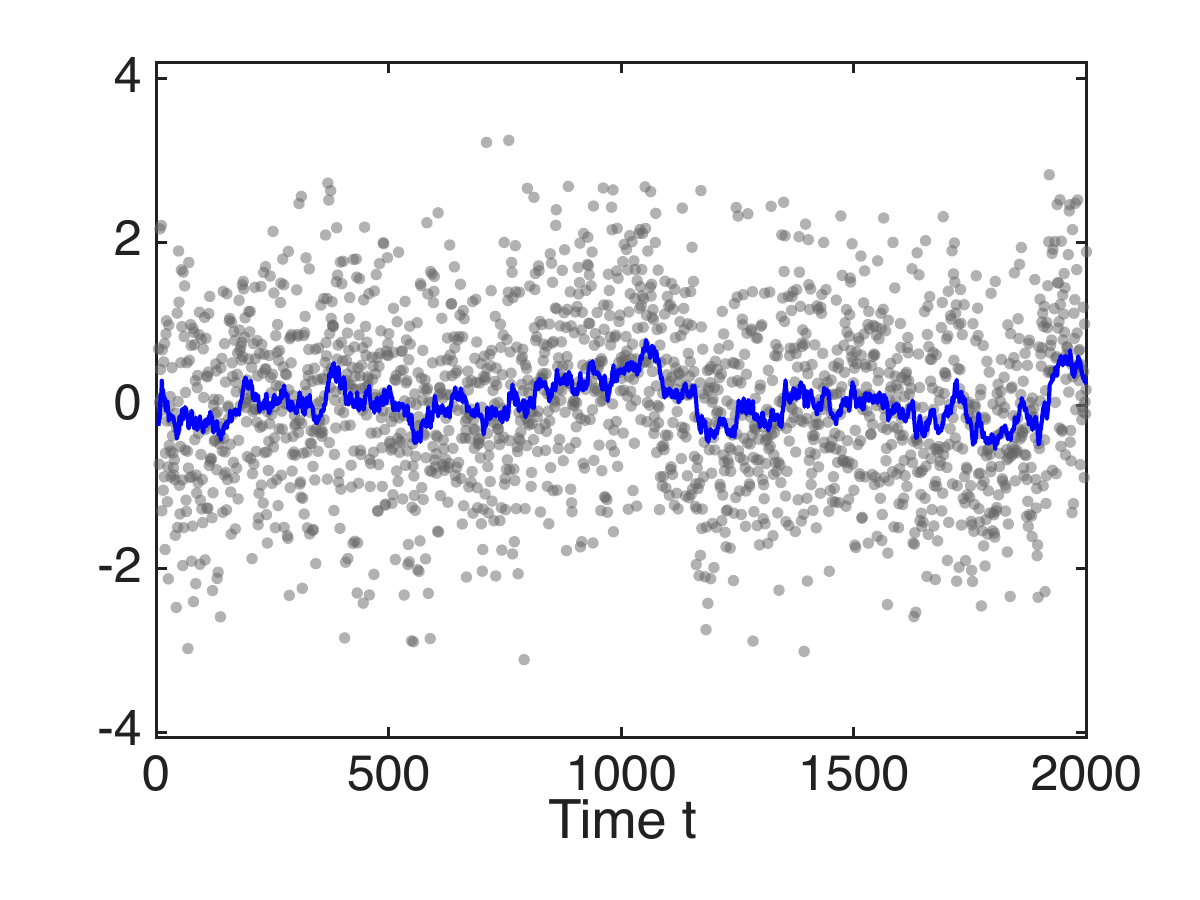} &
        \includegraphics[width=0.32\linewidth]{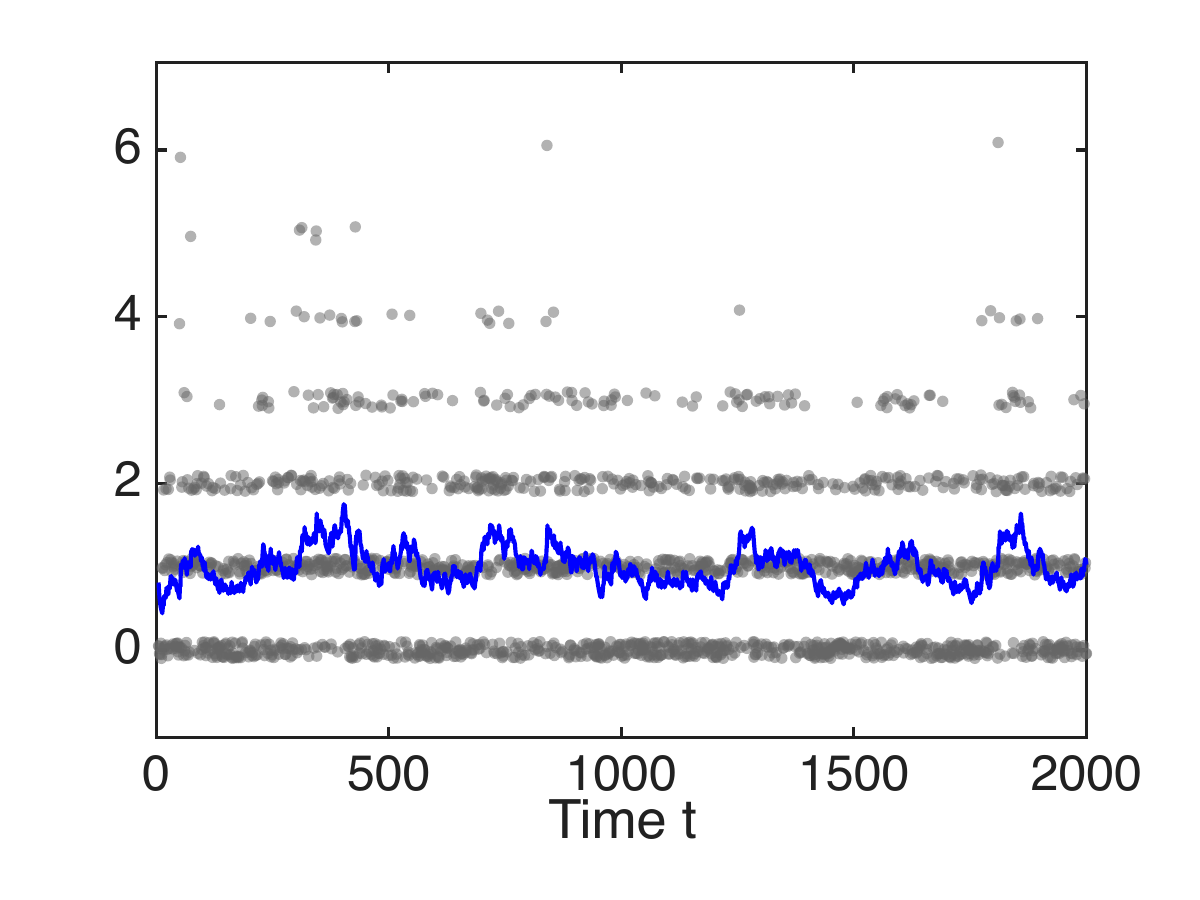} \\[-3mm]
        \raisebox{2.0cm}[0pt][0pt]{{\small $\alpha=0.95$\hspace{0.3cm}}} &
        \includegraphics[width=0.32\linewidth]{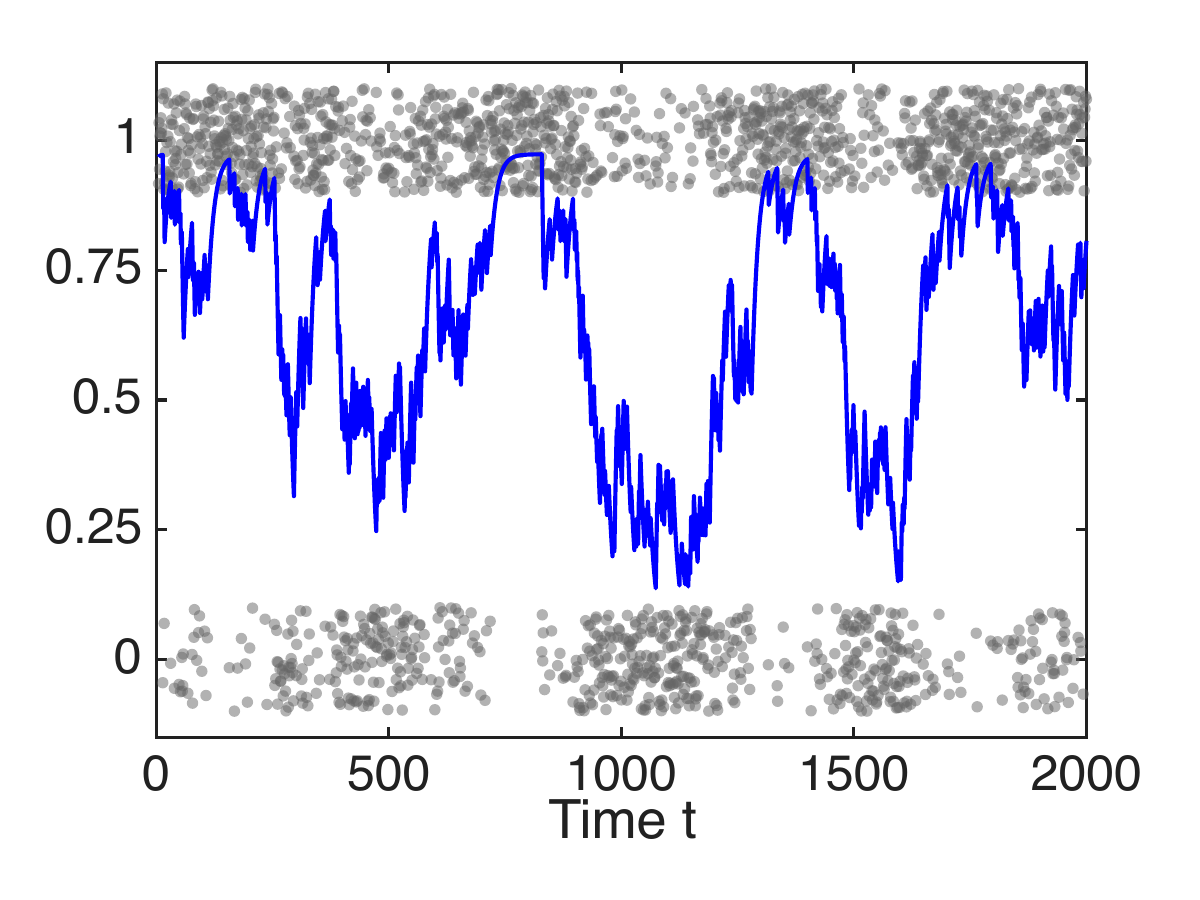} &
        \includegraphics[width=0.32\linewidth]{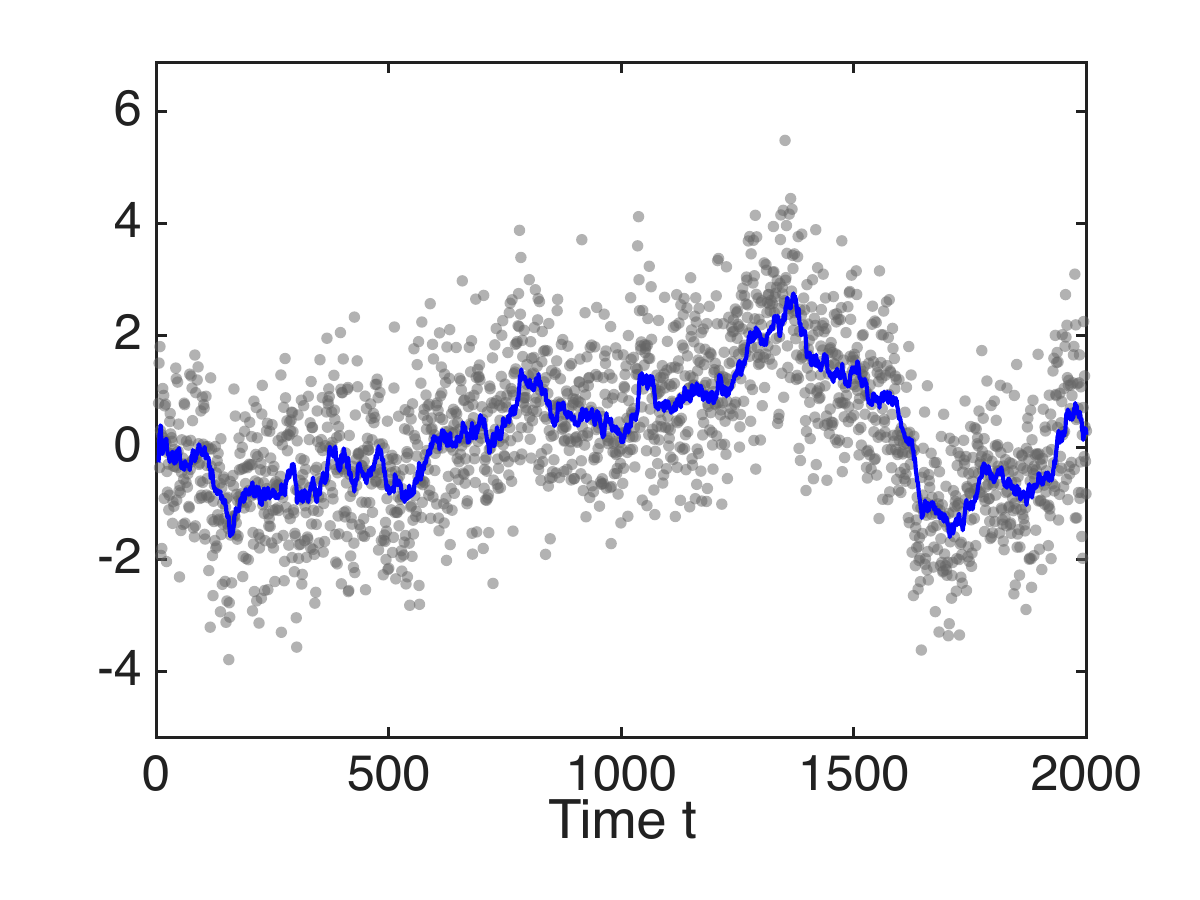} &
        \includegraphics[width=0.32\linewidth]{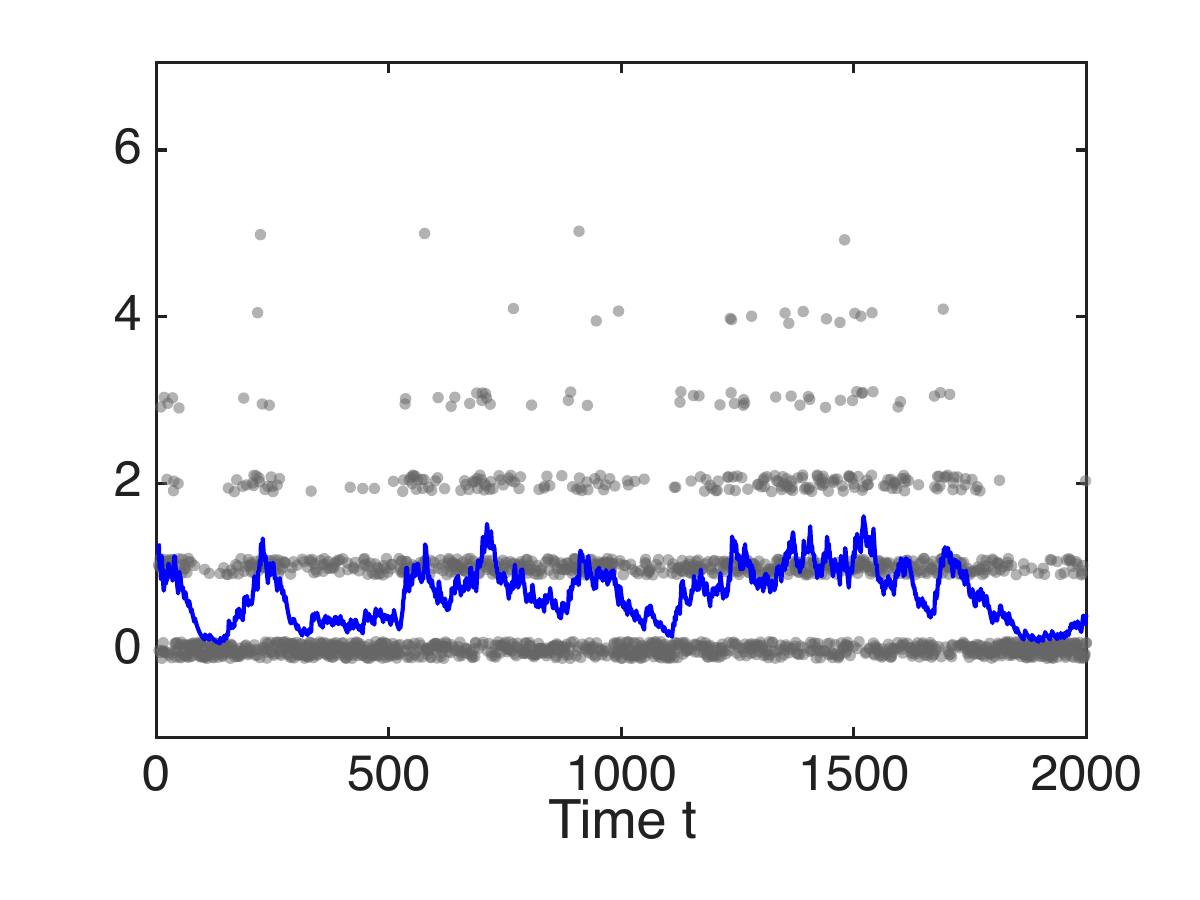}
    \end{tabular}
    \vspace{-5mm}
    \caption{Simulation from $Y_t | Y_{1:t-1} \sim {\tt CEF}(\widetilde{\theta}_{t|t-1},h,\psi)$ for time $t=5,...,T=2000$ with discount parameter $\lambda=0.93$. Top: anchoring parameter $\alpha=0.7$; bottom: $\alpha=0.95$. Gray circles show observations $Y_t$. Blue lines show the conditional mean $\widetilde{\mathbb{E}[Y_t|Y_{1:t-1}]} = \psi'(\widetilde{\theta}_{t|t-1})$.}
    \label{fig:CEFunistateExact}
\end{figure}

\begin{figure}[htbp]
    \centering
    \begin{tabular}{@{}c@{\hspace{-0.25cm}}c@{\hspace{-0.25cm}}c@{\hspace{-0.25cm}}c@{}}
        & \textbf{Exponential} & \textbf{Gaussian (zero mean)} & \textbf{Pareto} \\
        \raisebox{2.0cm}[0pt][0pt]{{\small $\alpha=0.70$\hspace{0.3cm}}} &
        \includegraphics[width=0.32\linewidth]{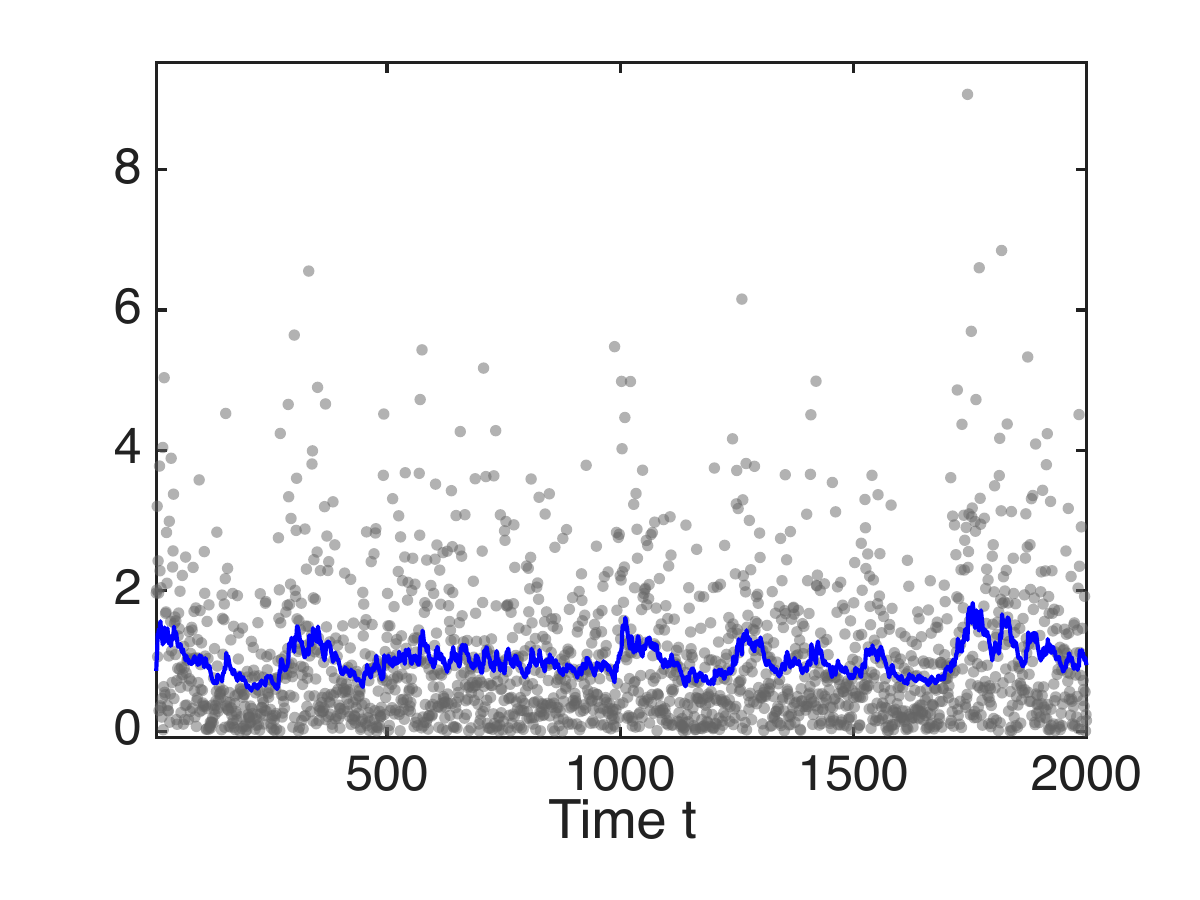} &
        \includegraphics[width=0.32\linewidth]{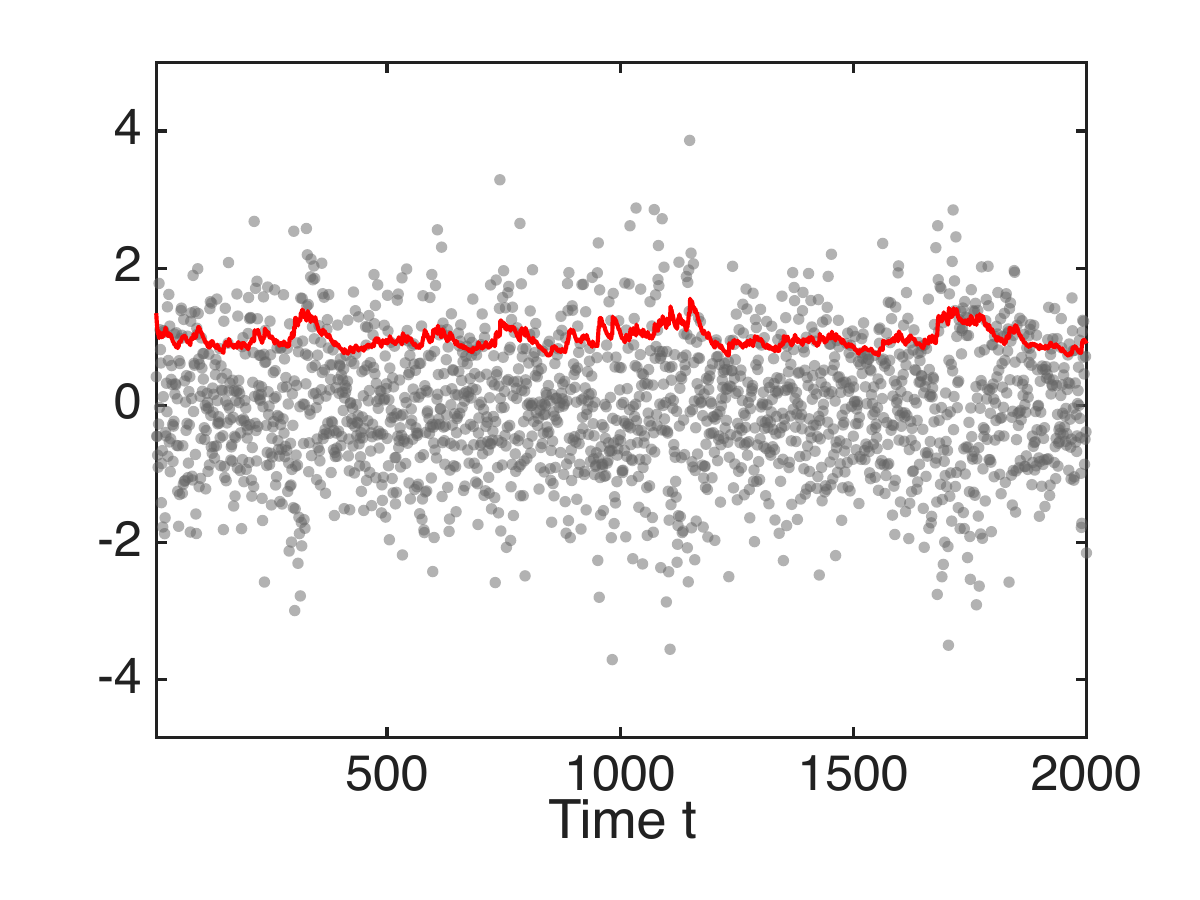} &
        \includegraphics[width=0.32\linewidth]{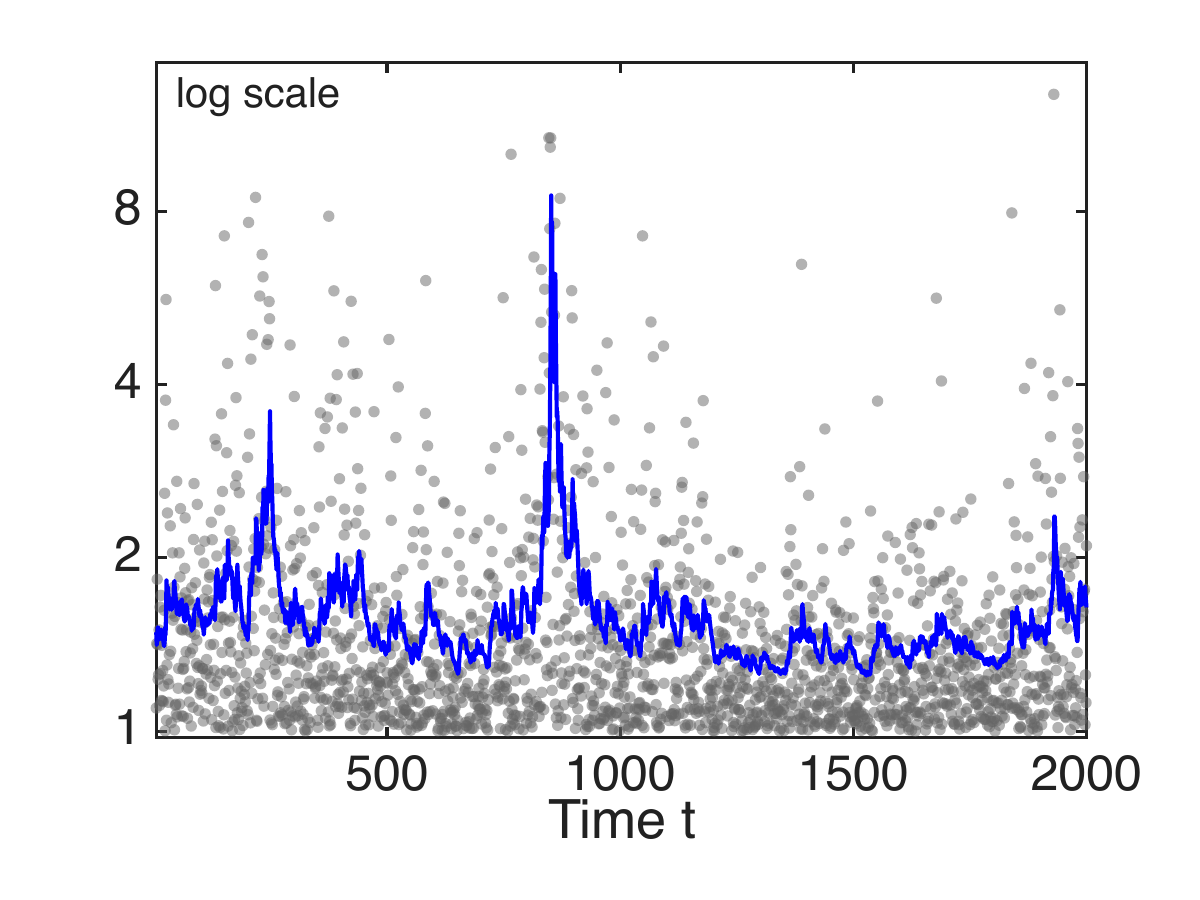} \\[-3mm]
        \raisebox{2.0cm}[0pt][0pt]{{\small $\alpha=0.95$\hspace{0.3cm}}} &
        \includegraphics[width=0.32\linewidth]{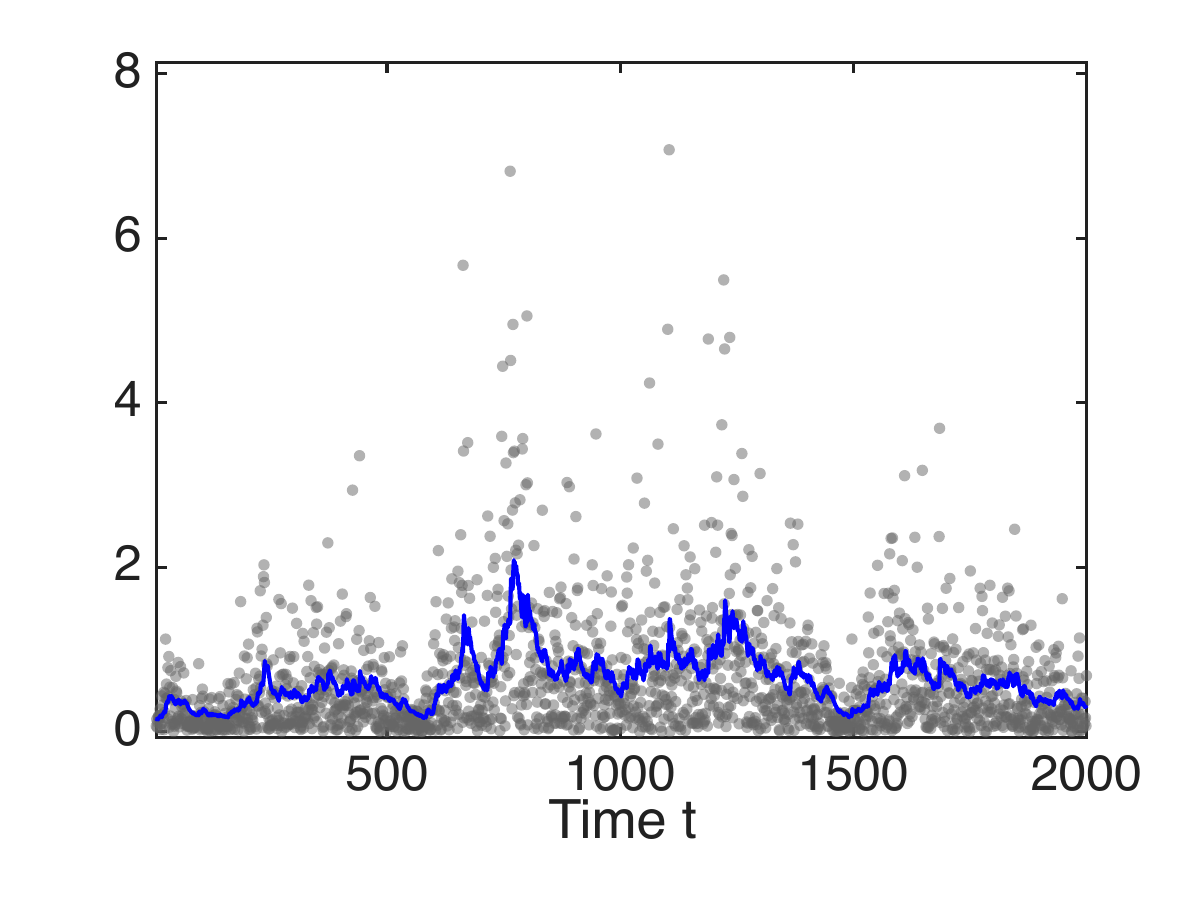} &
        \includegraphics[width=0.32\linewidth]{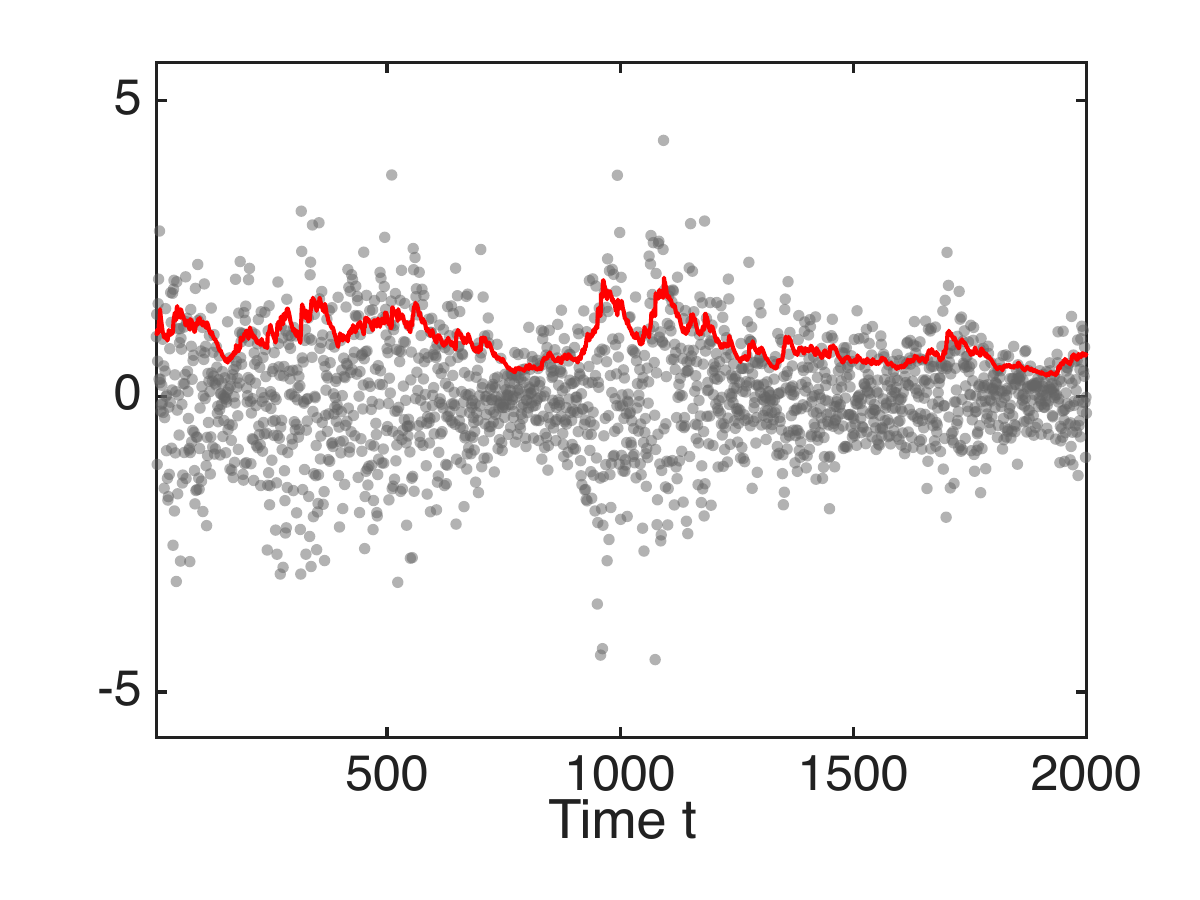} &
        \includegraphics[width=0.32\linewidth]{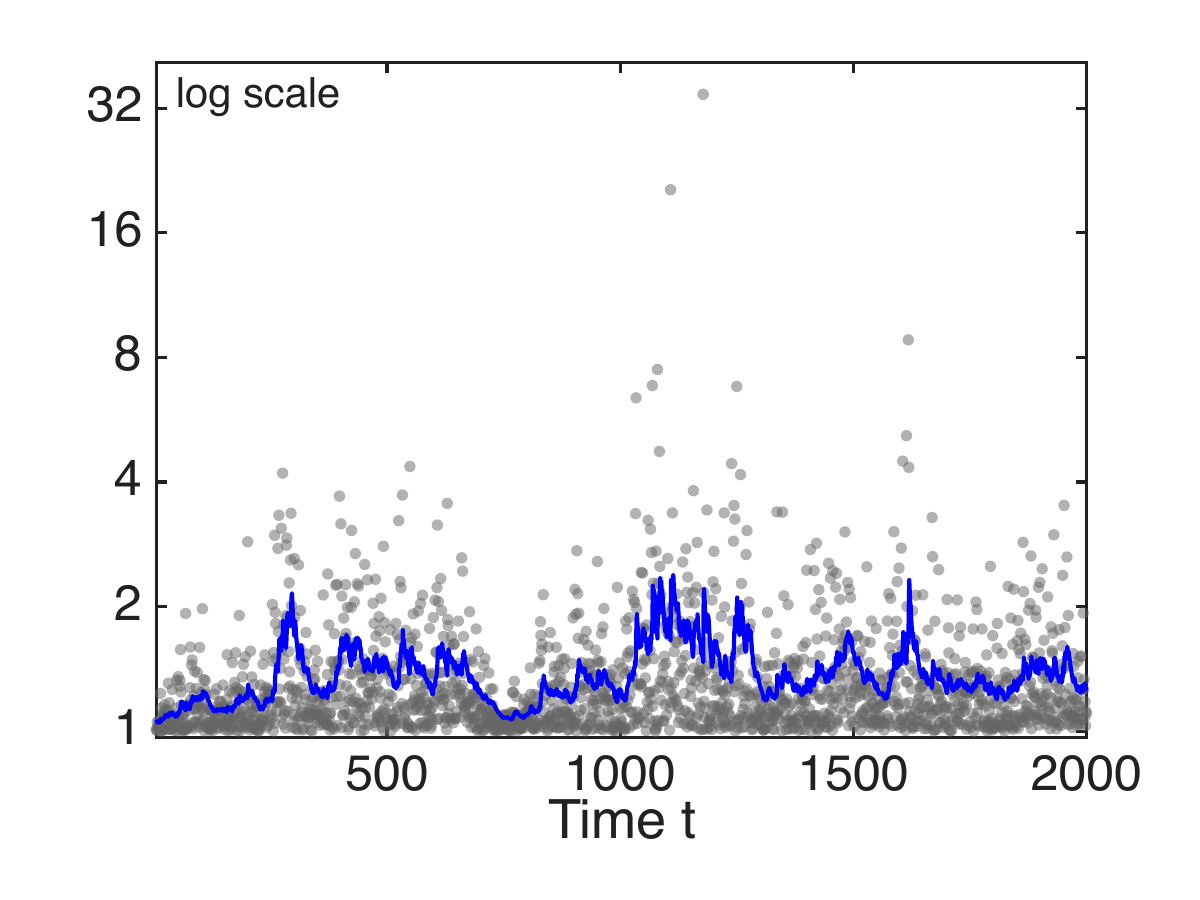}
    \end{tabular}
    \vspace{-5mm}
    \caption{Simulation from $Y_t | Y_{1:t-1} \sim {\tt CEF}(\widetilde{\theta}_{t|t-1},h,\psi)$ for time $t=5,...,T=2000$ with discount parameter $\lambda=0.93$. Top: anchoring parameter $\alpha=0.7$; bottom: $\alpha=0.95$. Gray circles show observations $Y_t$. Blue lines show the conditional mean $\widetilde{\mathbb{E}[Y_t|Y_{1:t-1}]} = \psi'(\widetilde{\theta}_{t|t-1})$ for the exponential, and $\widetilde{\mathbb{E}[Y_t|Y_{1:t-1}]} = \frac{\widetilde{\theta}_{t|t-1}}{\widetilde{\theta}_{t|t-1}+1}$ for the Pareto (y-axis on log scale) when $\widetilde{\theta}_{t|t-1} < -1$ and $\infty$ otherwise. The red line for the Gaussian (zero mean) shows the conditional standard deviation $\widetilde{\sigma}_{t|t-1} = \sqrt{\psi'(\widetilde{\theta}_{t|t-1})}$.}
    \label{fig:CEFunistate_additional}
\end{figure}

Figure~\ref{fig:two_state_simulations} shows simulation results for the three bivariate cases: Beta, Gaussian  with time-varying mean and variance, and von Mises. Each column displays the conditional expectation $\widetilde{\mathbb{E}[Y_t|Y_{1:t-1}]}$ with simulated observations. For Beta, the conditional mean is $\widetilde{\theta}_{1,t|t-1}/(\widetilde{\theta}_{1,t|t-1}+\widetilde{\theta}_{2,t|t-1})$. For Gaussian, it is $-\widetilde{\theta}_{1,t|t-1}/(2\widetilde{\theta}_{2,t|t-1})$ and the conditional standard deviation $\widetilde{\sigma}_{t|t-1} = \sqrt{-1/(2\widetilde{\theta}_{2,t|t-1})}$ is shown as a red line. For von Mises, we present the mean direction $\widetilde{\mu}_{t|t-1} = \text{atan2}(\widetilde{\theta}_{1,t|t-1}, \widetilde{\theta}_{2,t|t-1})$, modulo $2\pi$ to ensure $\widetilde{\mu}_{t|t-1} \in [0,2\pi]$, where $\text{atan2}(y,x)$ is the four-quadrant arctangent of $(x,y)$. For $\alpha = 0.95$, the predictors oscillate more around the centering values than for $\alpha = 0.70$, due to weaker mean reversion.

\begin{figure}[htbp]
    \centering
    \begin{tabular}{@{}c@{\hspace{-0.25cm}}c@{\hspace{-0.25cm}}c@{\hspace{-0.25cm}}c@{}}
        & \textbf{Beta} & \textbf{Gaussian} & \textbf{von Mises} \\
        \raisebox{2.0cm}[0pt][0pt]{{\small $\alpha=0.70$\hspace{0.3cm}}} &
        \includegraphics[width=0.32\linewidth]{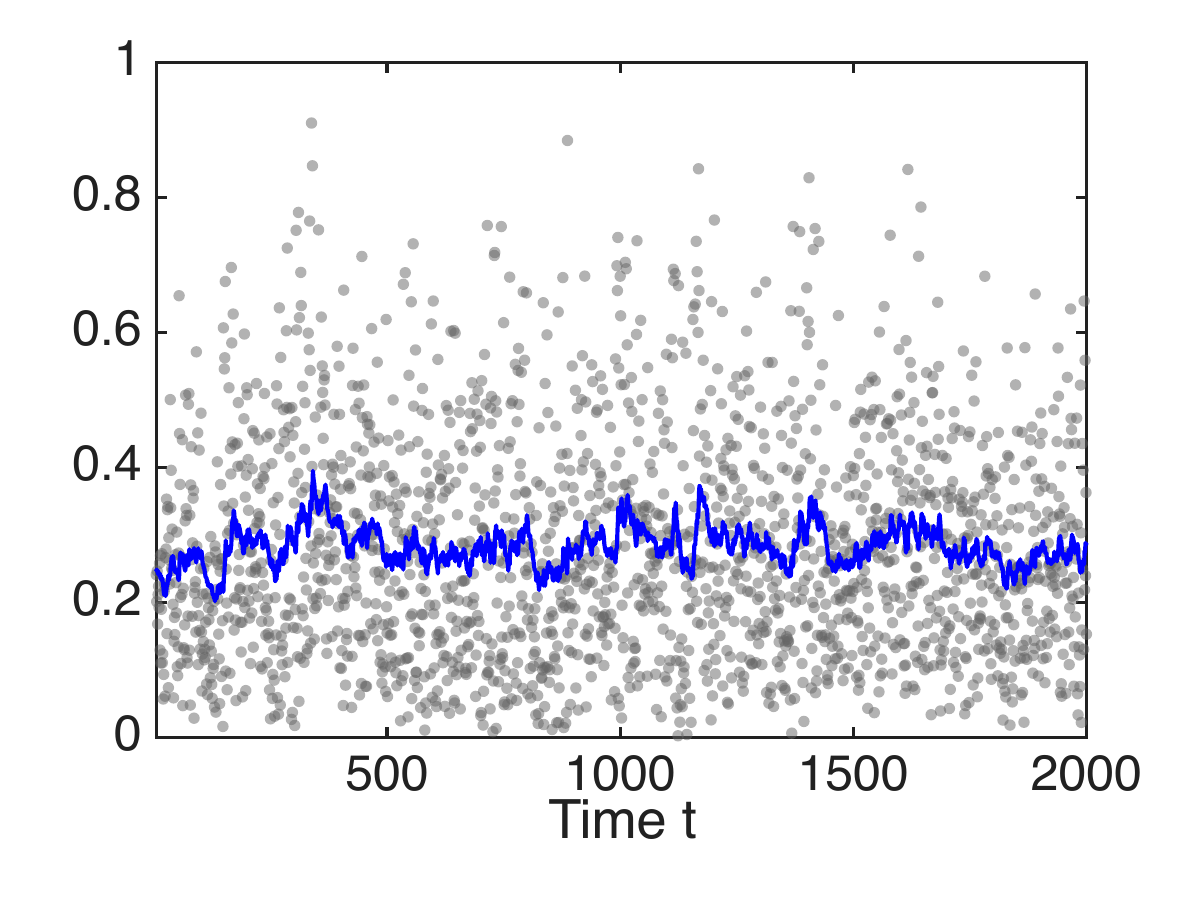} &
        \includegraphics[width=0.32\linewidth]{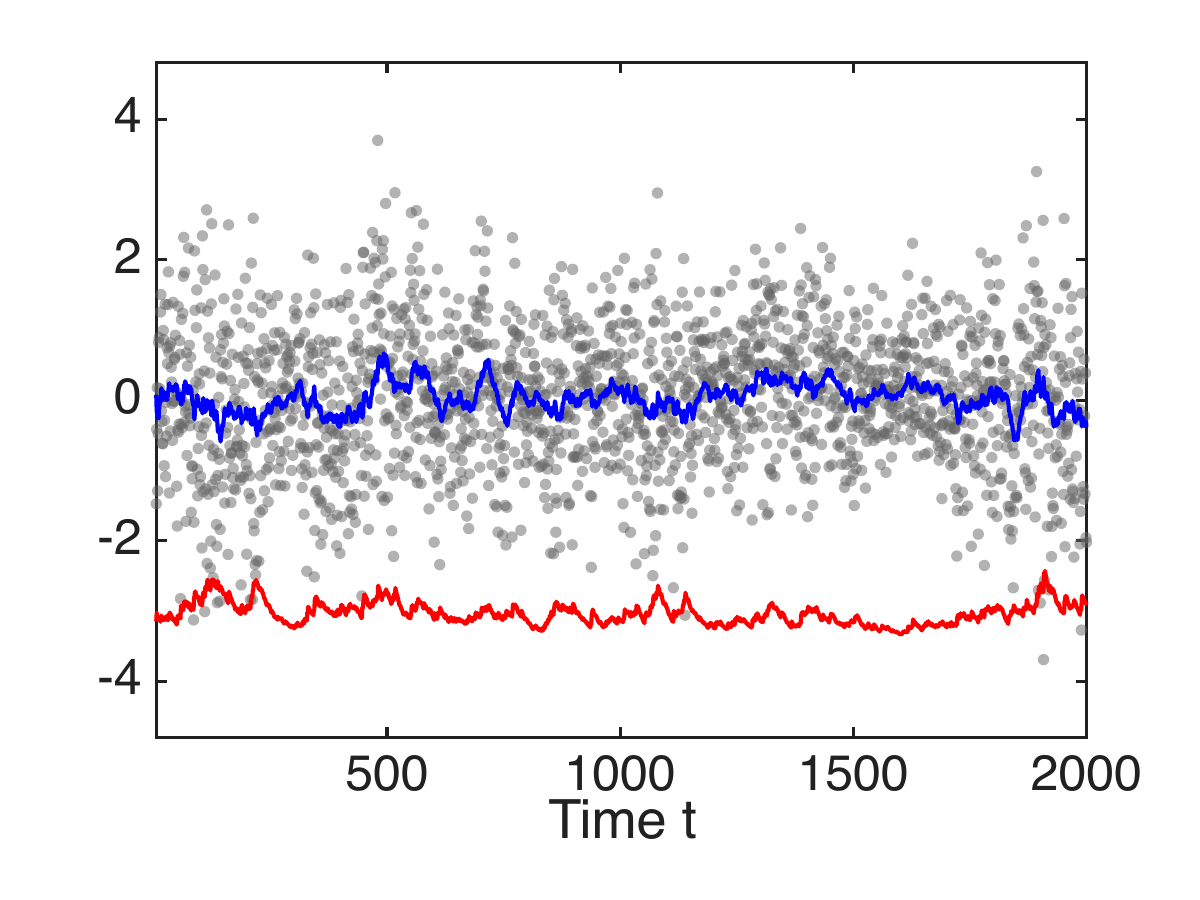} &
        \includegraphics[width=0.32\linewidth]{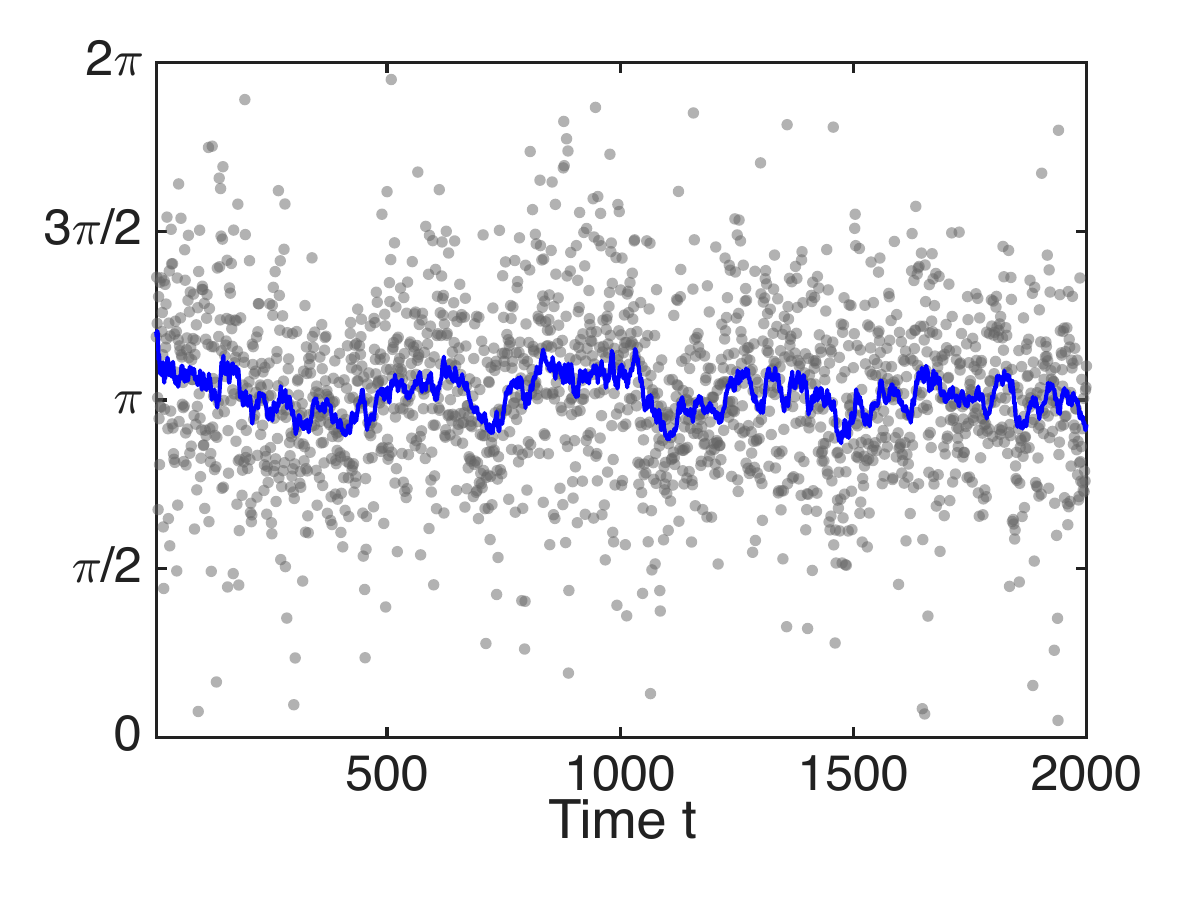} \\[-3mm]
        \raisebox{2.0cm}[0pt][0pt]{{\small $\alpha=0.95$\hspace{0.3cm}}} &
        \includegraphics[width=0.32\linewidth]{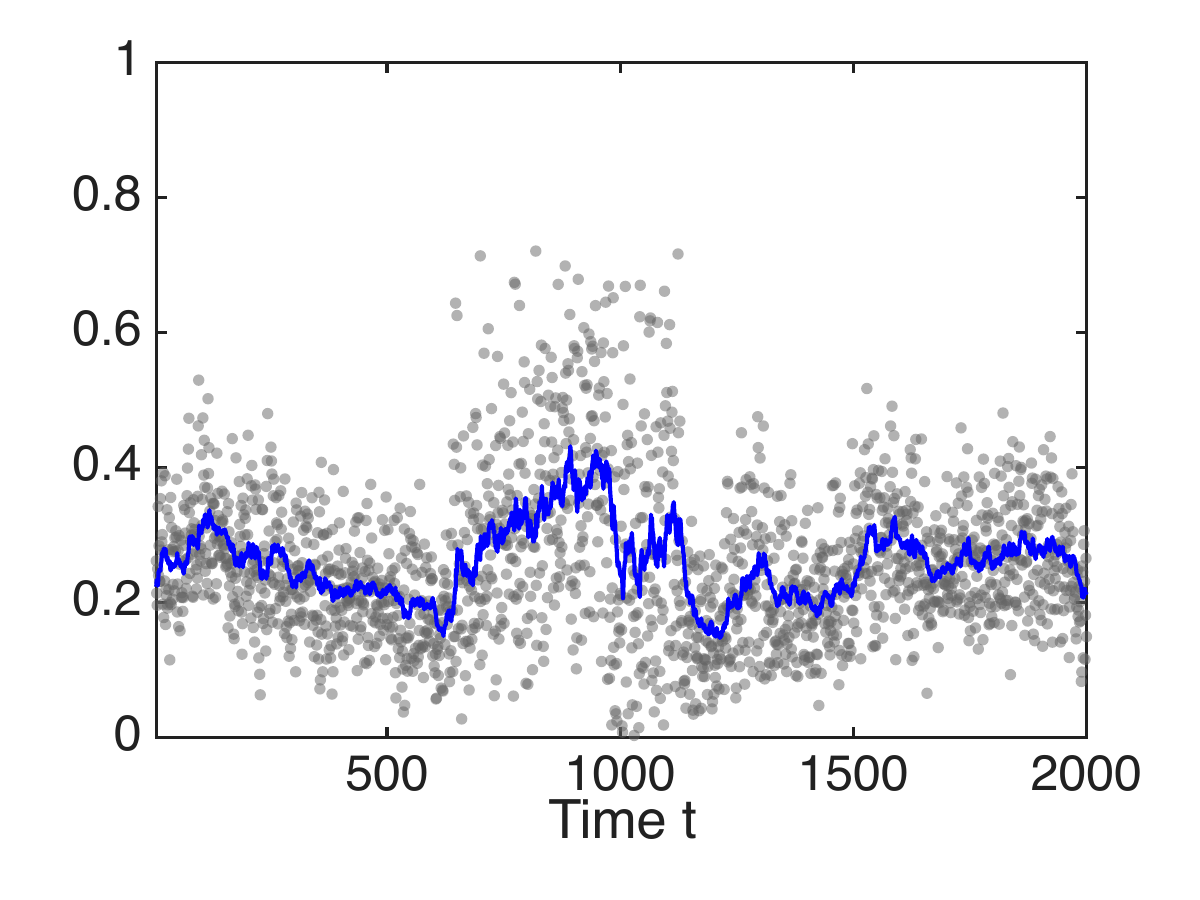} &
        \includegraphics[width=0.32\linewidth]{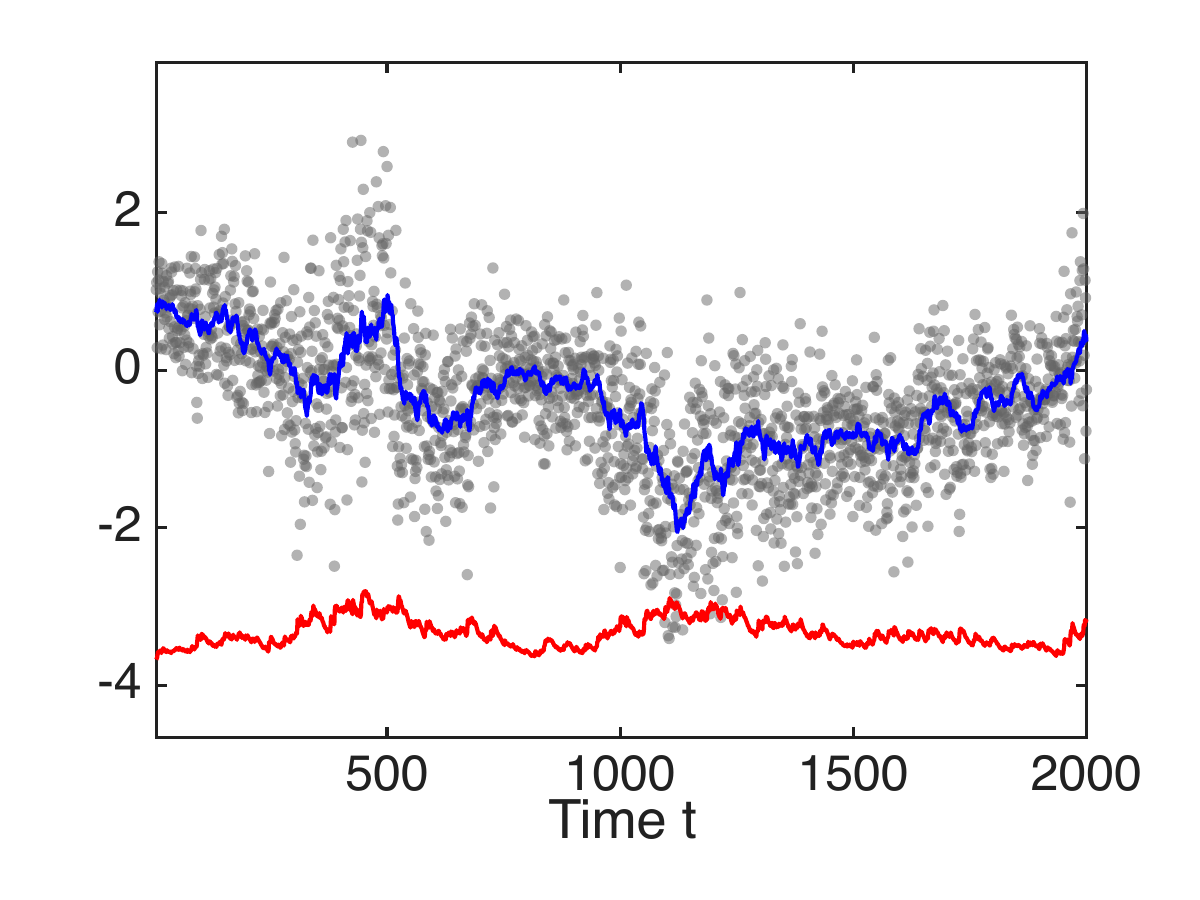} &
        \includegraphics[width=0.32\linewidth]{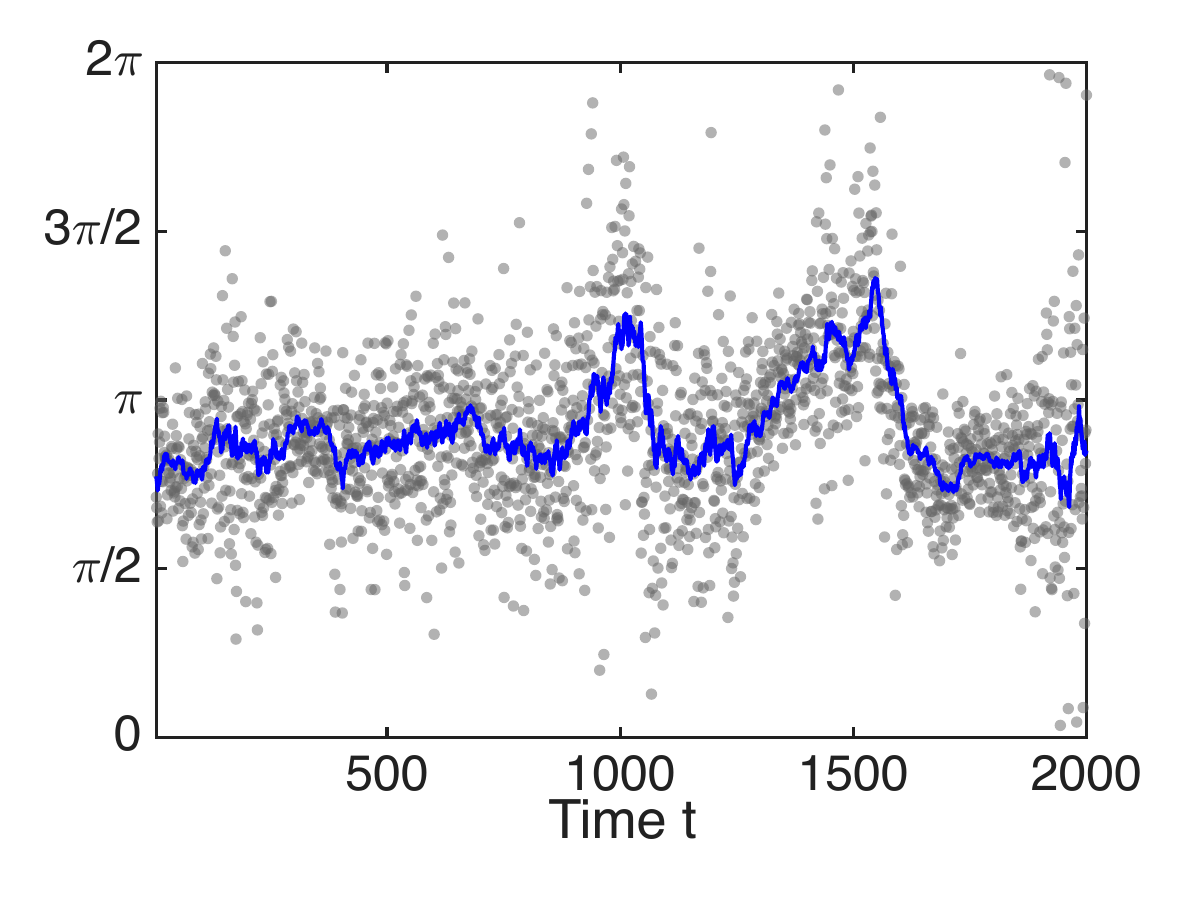}
    \end{tabular}
    \vspace{-5mm}
    \caption{Simulation from $Y_t | Y_{1:t-1} \sim {\tt CEF}(\widetilde{\theta}_{t|t-1},h,\psi)$ for time $t=5,...,T=2000$ with discount parameter $\lambda=0.93$. Top: anchoring parameter $\alpha=0.7$; bottom: $\alpha=0.95$. Gray circles show observations $Y_t$. Blue lines show the conditional mean $\widetilde{\mathbb{E}[Y_t|Y_{1:t-1}]}$ for Beta, Gaussian (time-varying mean and variance), and von Mises. For the Gaussian, the red line shows the conditional standard deviation $\widetilde{\sigma}_{t|t-1}$, shifted down by 4 units for visualization.}
    \label{fig:two_state_simulations}
\end{figure}

\subsection{Gaussian case: comparison to Kalman filter}\label{sect:kalman}

The Gaussian example with known standard deviation, case 2, links strongly to the Kalman filter for the univariate Gaussian local level model 
\begin{align*}
Y_t &= \mu_t + \varepsilon_t, \quad \varepsilon_t \overset{iid}{\sim } N(0, \sigma^2_\varepsilon), \\
\mu_{t+1} &= \mu_t + \eta_t, \quad \eta_t \overset{iid}{\sim } N(0, q\sigma^2_\varepsilon),\quad (\varepsilon_{1:T} \ind \eta_{2:T})|\mu_1,\quad \mu_1\sim N(a_{1|0},\sigma^2_\varepsilon P_{1}).
\end{align*}
The Kalman filter implies $a_{t|t-1} = {\mathrm E}[\mu_t|Y_{1:t-1}] = {\mathrm E}[\mu_{t-1}|Y_{1:t-1}]$ is given by (e.g. \cite{Harvey(89)} and \cite{DurbinKoopman(12)}) 
\begin{align*}
a_{t+1|t} &= a_{t|t-1} + K_t(Y_t - a_{t|t-1}),\quad K_t = \frac{P_t}{P_{t}+1},\quad P_{t+1} = K_t + q, 
\end{align*}
where $K_t$ is called the Kalman gain and $\sigma^2_\varepsilon P_t = {\mathrm V}(\mu_t|Y_{1:t-1})$. If $q>0$, in steady state, $K_t$ and $P_t$ converge to
$
K = \frac{P}{P + 1}$ and $ P=K+q. $
Then $K^2 + qK - q = 0$, implying 
$K = \frac{-q + \sqrt{q^2 + 4q}}{2}.
$ 
In steady state, the one-step predictor is
$$
a_{t+1|t} = (1-K)a_{t|t-1} + K Y_t= K \sum_{j=0}^{\infty} (1-K)^j Y_{t-j},
$$
which is an {\tt EWMA} with discount hyperparameter 
$$\lambda=1-K = \frac{2+q-\sqrt{(2+q)^2-4}}{2}.$$
The discount rate monotonically declines as $q$ increases, going from 1 to 0.  

Our predictor in steady state, setting $\mathbb{E}[Y_1] = 0$ and $\alpha=1$, implies 
$$
\widetilde{\mu}_{t+1|t} = \lambda \widetilde{\mu}_{t|t-1} + (1-\lambda)Y_{t}. 
$$

A more subtle comparison is out of steady state. For the Kalman filter, take $P_{1}=\infty$, then $K_1 = 1$. In our predictor, 
$$
\widetilde{\mu}_{t+1|t} = \left(1-\frac{1}{n_{\lambda,t}}\right) \widetilde{\mu}_{t|t-1} + \frac{1}{n_{\lambda,t}}Y_{t}.
$$
We note that if $\lambda>0$, then  
$
n_{\lambda,t} = \frac{1-\lambda^t}{1-\lambda}.
$
Hence it is interesting to plot $K_t \times n_{\lambda,t}$ 
against $t$, when $P_{1|0}=\infty$, the so called diffuse initial conditions for a Kalman filter (e.g. Chapter 1 of \cite{DurbinKoopman(12)}). If this product is less than one, then our recursion gives more weight on $Y_t$ than the Kalman filter, and subsequently, less on the past.

Then $K_t n_{\lambda,t}=1$ if $t=1$. It is also 1 for every value of $t$ when $q=0$ for then $n_{\lambda,t}=t$.  The limit is also 1 as $t \rightarrow \infty$ for $q>0$ (steady state), but the product is not 1 for every $t$ and $q>0$.  

The easiest way to see differences is to think of $q$ as small, then, for example,  
$
\lambda \approx 1-q^{1/2}$, $ n_{\lambda,2} \approx 2 - q^{1/2}$, $ K_2 \approx \frac{1}{2} -q,$ and $K_2 n_{\lambda,2} \approx 1-q^{1/2}/2.$
Figure \ref{fig:DiscountErr} plots the product $K_t \times n_{\lambda,t}$ for $q\in \{0.001,0.1,0.3,1,2,10\}$ against $t$.  The worst case is when $q$ is tiny, as it takes quite a large $t$ for the effect to disappear (as we move to the steady state). 

\begin{figure}[h]
\centering
\includegraphics[width=.55 \textwidth]{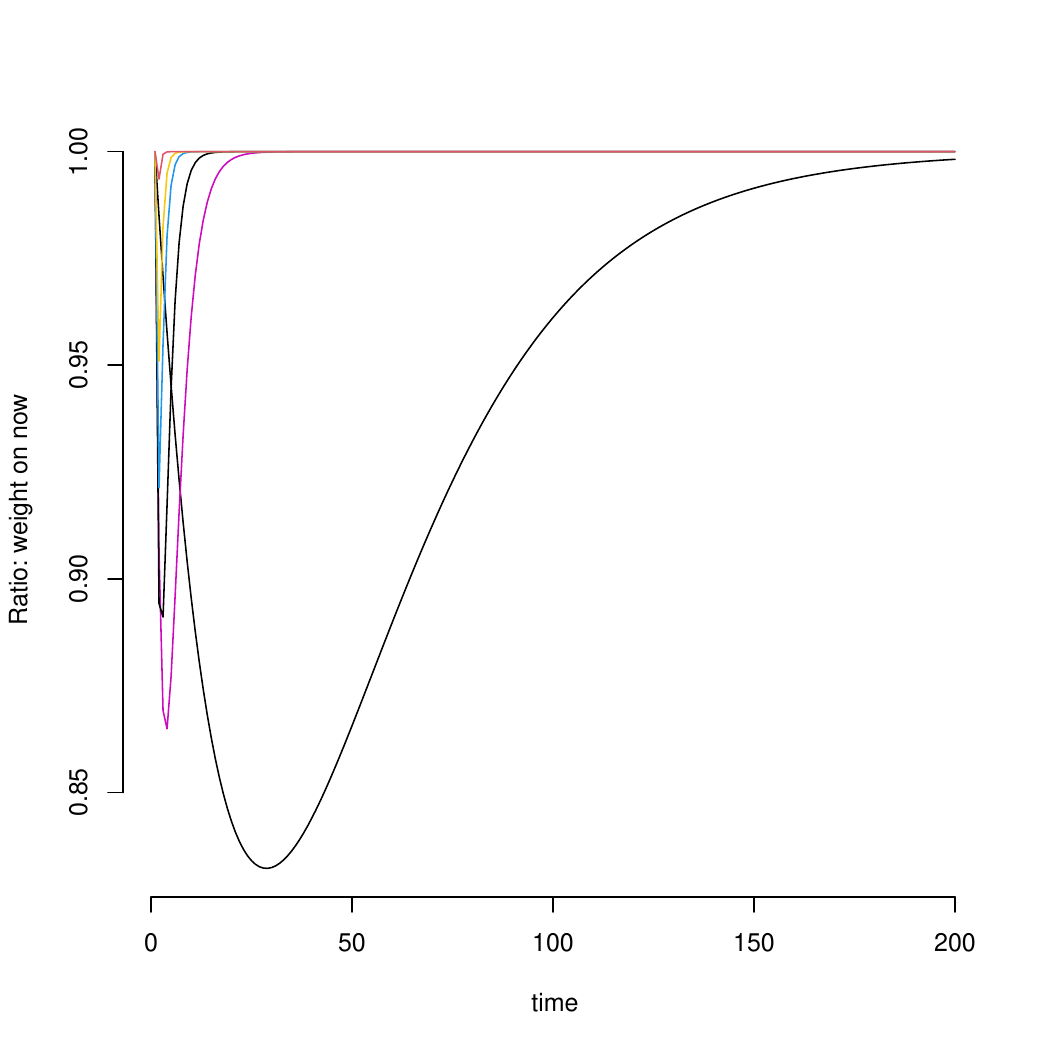} 
\vspace{-4mm}
    
\caption{The product $K_t n_{\lambda,t}$ against $t$ for $q \in \{0.001, 0.1, 0.3, 1, 2, 10\}$. 
This is the ratio of the Kalman filter's weight on $Y_t$ (under diffuse initialization) 
to the EWMA weight on $Y_t$, showing the impact of initial conditions. Values below 1 mean the Kalman filter places less weight 
on the current observation than the EWMA. All curves converge to 1 in steady state.
Deviations from 1 are largest, and convergence to steady-state
is slowest, for small $q$.
}   
    \label{fig:DiscountErr}
\end{figure}

Appendix~\ref{app:scoredriven-full} compares our predictor to the widely applied class of score-driven filters \citep{Harvey(13), CrealKoopmanLucas(13), lange2024robust}.

\section{Hyperparameter estimation}\label{sect:quasi}

The analysis in this Section will be based on a working model 
$$
Y_t|Y_{1:t-1} \sim {\tt CEF}(\widetilde{\theta}_{t|t-1};h_t,\psi_t),\quad t=1,...,T,
$$ 
where $\widetilde{\theta}_{t|t-1}$ comes from the exponentially weighted estimands. Later we will explicitly note the dependence of $\widetilde{\theta}_{t|t-1}$ on some hyperparameters, denoted $\omega$.  

Under the working model, write  
$$
\widetilde{\mu}_{t|t-1} := \psi_t'(\widetilde{\theta }_{t|t-1}),\quad 
\widetilde{\Sigma}_{t|t-1} := \psi_t''(\widetilde{\theta }_{t|t-1}),
$$
the model's conditional mean and variance, respectively. 

The working model yields a working log-likelihood via the prediction decomposition, ignoring constants, 
\begin{align*}
l_t &= \sum_{j=1}^t \Delta l_j,\quad \Delta l_t = h_t(y_t)^{\tt T}\widetilde{\theta }_{t|t-1} - \psi_t(\widetilde{\theta }_{t|t-1}).
\end{align*}
To make sense of this working log-likelihood function we need some regularity assumptions.  

\begin{assumption}\label{ass:quasiReg} Assume for a sequence $\{\widetilde{\theta}_{t|t-1}\}_{t=1}^T$ there exist constants $c_{1:2},d_{1:2}$ such that for all~$t$:

\noindent (a) the  
$\mathbb{E}[|h(Y_t)^{\tt T}\widetilde{\theta }_{t|t-1}|]<c_1$ and $\mathbb{E}[|\psi_t(\widetilde{\theta }_{t|t-1})|]<d_1$.  

\noindent (b) the ${\mathrm V}[h(Y_t)^{\tt T}\widetilde{\theta }_{t|t-1}]<c_2$ and ${\mathrm V}[\psi_t(\widetilde{\theta }_{t|t-1})]<d_2$.
\end{assumption}

\subsection{A quasi-likelihood interpretation}

As the working log-likelihood $$l_t,\quad t=1,...,T,$$ comes from the ${\tt CEF}(\widetilde{\theta}_{t|t-1},h_t,\psi_t)$, under weak regularity conditions, it will also be a quasi-likelihood, so long as the predictive mean is modeled correctly.  This property does not depend on the particular features of  $\widetilde{\theta}_{t|t-1}$.  We demonstrate this for the canonical exponential family time series case.  It applies when $\widetilde{\theta}_{t|t-1}$ is formed from the weighted predictor.  

Now make some assumptions about the data.  

\begin{assumption}\label{ass:MDx} Under the data, assume that:

\noindent (a) $\mathbb{E}[|h(Y_t)|]<\infty$ and denote 
$
\mathbb{E}[h_t(Y_t)|Y_{1:t-1}] := \mu_{t|t-1},$ where $t=1,...,T.$

\noindent (b) ${\mathrm V}[h_t(Y_t)]$ exists, the 
$
\Sigma_{t|t-1}:= {\mathrm V}[h_t(Y_t)|Y_{1:t-1}],$ for $t=1,...,T
$ 
and there exists a positive definite matrix $C$ such that $\Sigma_{t|t-1} - C$ is positive semidefinite for every $t$.  
\end{assumption}

Using strict convexity of $\psi_t$ and Assumption \ref{ass:MDx}(a) there exists a unique $\theta_{t|t-1}$ such that $\mu_{t|t-1}  =\psi'_t(\theta_{t|t-1})$.  

\begin{definition} Assume the sequence $\{\theta_{t|t-1}\}_{t=1}^T$ obeys Assumption \ref{ass:quasiReg}(a) and $\theta_{t|t-1}=(\psi')^{-1}(\mu_{t|t-1})$ from Assumption \ref{ass:MDx}(a).  Then define the oracle log-likelihood: 
\begin{align*}
l^*_t &= \sum_{j=1}^t \Delta l^*_j,\quad \Delta l^*_t = h_t(y_t)^{\tt T}\theta _{t|t-1} - \psi_t(\theta _{t|t-1}),\quad t=1,...,T.
\end{align*}
\end{definition}

Under Assumptions \ref{ass:quasiReg}(a) and \ref{ass:MDx}(a), define the two sequences 
\begin{align*}
M_t &= \sum_{j=1}^t \{l_j - \mathbb{E}[l_j^*|Y_{1:j-1}]\}, \quad t=1,2...,T,\\
C_t &= \sum_{j=1}^t \{\mathbb{E}[l_j|Y_{1:j-1}]-\mathbb{E}[l^*_j|Y_{1:j-1}]\}. 
\end{align*}
then 
$\{M_t\}_{t=1}^T$ is a supermartingale with respect to the data's natural filtration. The $\{C_t\}_{t=1}^T$ is a previsible drift with increments
$$
c_t = \psi'(\theta _{t|t-1})^{\tt T} \{\widetilde{\theta }_{t|t-1}-\theta_{t|t-1}\} - [\psi\{\widetilde{\theta }_{t|t-1}\}
-\psi\{\theta _{t|t-1}\}] \le 0,
$$
by convexity of $\psi$ from Definition \ref{ex:EF} and then using p.69 of \cite{BoydVandenberghe(04)}.  The equality is only obtained iff 
$
\widetilde{\theta }_{t|t-1} = \theta _{t|t-1},
$
that is only if the working model has the correct predictive mean.  Thus under Assumptions \ref{ass:quasiReg}(a) and \ref{ass:MDx}(a) the log-likelihood process $\{l_t\}_{t=1}^T$ can be viewed as a quasi-likelihood process. 

If the magnitude of $\widetilde{\theta }_{t|t-1}-\theta_{t|t-1}$ is small, then $$c_t \approx -\frac{1}{2}(\widetilde{\theta }_{t|t-1}-\theta_{t|t-1})^{\tt T} \psi_t''(\theta_{t|t-1}) (\widetilde{\theta }_{t|t-1}-\theta_{t|t-1}).$$ 

More broadly, under Assumptions \ref{ass:quasiReg}(a) and \ref{ass:MDx}(a), the $\{M_t-C_t\}_{t=1}^T$ is a martingale.  Further, under Assumptions \ref{ass:quasiReg}(b) and \ref{ass:MDx}(b) the angle bracket of the log-likelihood process  
\begin{align*}
\langle l\rangle_t & := \sum_{j=1}^t {\mathrm V}(l_j|Y_{1:j-1}) = \sum_{j=1}^t  \widetilde{\theta }_{j|j-1}^{\tt T}\Sigma_{j|j-1}\widetilde{\theta }_{j|j-1}
= \langle M-C\rangle_t,\quad t=1,...,T
\end{align*}
the angle-bracket process of the $\{M_t-C_t\}_{t=1}^T$ process.  For all cases where $\langle M-C\rangle_T\rightarrow \infty$ then $(M_T-C_T)/\langle M-C\rangle_T \rightarrow 0$ almost surely, by the martingale strong law of large numbers.    

\begin{remark} (a) This quasi-likelihood interpretation is unsurprising, falling in line with the history of generalized linear models, e.g. \cite{Wedderburn(74)} and \cite{McCullaghNelder(89)}.  More broadly, quasi-likelihood estimation theory goes back at least to \cite{Cox(61)}, \cite{Huber(67)}, \cite{Gallant(87)} and \cite{White(82)}.

(b) The scaled $\{C_t(\omega)\}_{t=1}^T$ process measures how the scaled working log-likelihood $l_T$ minus the scaled oracle $l^*_T$ drifts downwards as $T$ increases.

\end{remark}

\subsection{Maximum likelihood estimation}

Now turn to the maximum likelihood estimator of the hyperparameters based on the above exponential family quasi-likelihood, where the predictor comes from exponential weighting and uses a {\tt CEF} frame.  

\begin{assumption}\label{ass:psuedo} Based on a ${\tt CEF}(\theta,h_t,\psi_t)$ frame, write $\widetilde{\mu}_{t|t-1}(\omega)$ as the predictor using some finite dimensional hyperparameter vector $\omega \in \Omega$. The hyperparameter vector may include any static parameters $\phi \in \Phi$ from the model density. Assume there exists a $\omega^*$ such that 
$$\widetilde{\mu}_{t|t-1}(\omega^*)=\widetilde{\mu}_{t|t-1},\quad \text{for all} \quad t=1,...,T,
$$
in the data and that $\widetilde{\mu}_{t|t-1}(\omega^*)\ne\widetilde{\mu}_{t|t-1}$ for all $\omega \in (\Omega \setminus \omega^*)$.  Then refer to $\omega^*$ as the pseudo-true or oracle value.  
\end{assumption}

The leading version of this is given in Example \ref{ref:hyper}.

\begin{example}\label{ref:hyper} For a strictly stationary stochastic process $\{Y_t\}_{t\ge 1}$, assume that
$\mathbb{E}[h(Y_1)]$ exists and lies in $\mathcal{H} := h(\mathcal{Y}) \subseteq \mathbb{R}^k$. For the exponentially weighted estimands under the stable frame minimal ${\tt CEF}(\theta,h,\psi)$, in steady state, 
$$
\widetilde{\mu}_{t|t-1}(\omega) = \frac{(1-\alpha)(1-\lambda)}{1-\alpha(1-\lambda)}\mathbb{E}[h(Y_1)] + \frac{\alpha \lambda(1-\lambda)}{1-\alpha(1-\lambda)} h(Y_{t-1}) + \lambda \widetilde{\mu}_{t-1|t-2},
$$ following Remark~\ref{remark:SS CEF}(e).
Here the hyperparameters are
$$
\omega := (\mathbb{E}[h(Y_1)]^{\tt T},\alpha, \lambda)^{\tt T} \in \Omega^* = \mathcal{H}\times (0,1)\times(0,1).
$$  
The $\widetilde{\mu}_{t|t-1}(\omega)$ is infinitely differentiable with respect to $\omega$, linear in $h(y_1),...,h(y_{t-1})$ and $\mathbb{E}[h(Y_1)]$, but nonlinear in $\lambda$.   
More generally, if the model density contains additional static parameters $\phi$ (e.g., a static variance parameter $\sigma^2$ in the Gaussian case that is not modeled by the exponentially weighted estimand), these can be included in the hyperparameter vector
$$
\omega := (\mathbb{E}[h(Y_1)]^{\tt T},\alpha, \lambda, \phi^{\tt T})^{\tt T} \in \Omega^* = \mathcal{H}\times (0,1)\times(0,1) \times \Phi.
$$
Such static parameters enter the $l_t(\omega)$ but do not affect the structure of the predictor $\widetilde{\mu}_{t|t-1}$.

\end{example}

The resulting maximum likelihood estimate formed from the quasi-likelihood is: 
$$
\hat{\omega}_{\tt MLE} = \underset{\omega \in \Omega}{\arg }\max \ l_T(\omega),  
$$ 
noting explicitly how the quasi-likelihood is impacted by the choice of $\omega$ through $\{\widetilde{\theta}_{t|t-1}(\omega)\}_{t=1}^T$. Under Assumption \ref{ass:psuedo}, the corresponding $\omega^*$ is the pseudo-true value of $\omega$ for this quasi-likelihood. 

Throughout we will posit that $\widetilde{\theta}_{t|t-1}(\omega)$ is infinitely differentiable for all $\omega \in \Omega^*$ and write 
$$
\widetilde{\mu}'_{t|t-1}(\omega) =  \frac{\partial \widetilde{\mu}_{t|t-1}(\omega)^{\tt T}}{\partial \omega}.
$$

The corresponding score up to time $1 \le t\le T$ is 
\begin{align*}
S_t(\omega;y_{1:t}) &:= \frac{\partial l_t(\omega)}{\partial \omega}\\
&= \sum_{j=1}^t s_j(\omega), \quad 
s_t(\omega):= \frac{\partial \Delta l_t(\omega)}{\partial \omega}
= \frac{\partial \widetilde{\theta}_{t|t-1}(\omega)^{\tt T}}{\partial \omega}\frac{\partial \Delta l_t(\omega)}{\partial \widetilde{\theta}_{t|t-1}(\omega)}.
\end{align*}
Assume that for all $\omega \in \Omega$ that $\widetilde{\Sigma}_{t|t-1}(\omega)=\psi''(\widetilde{\theta}_{t|t-1}(\omega))$ is invertible.  Then 
$$
s_t(\omega) 
= \widetilde{\theta}'_{t|t-1}(\omega)\{h(y_t) -\widetilde{\mu}_{t|t-1}(\omega) \}, \text{where } \widetilde{\theta}'_{t|t-1}(\omega):= \frac{\partial \widetilde{\theta}_{t|t-1}(\omega)^{\tt T}}{\partial \omega} = \widetilde{\mu}'_{t|t-1}(\omega)\widetilde{\Sigma}^{-1}_{t|t-1}(\omega),
$$
recalling $\widetilde{\Sigma}_{t|t-1}(\omega):= \psi''(\widetilde{\theta}_{t|t-1}(\omega))$.  The Hessian is 
$$
H_t(\omega) = -\sum_{j=1}^t l_j''(\omega),\quad \text{where} \quad l_t''(\omega)=\frac{\partial^2 l_t(\omega)}{\partial \omega \partial \omega^{\tt T}}.
$$

\begin{remark} (a) Let $s_t := s_t(\omega^*;Y_{1:t})$, and $S_t = \sum_{j=1}^t s_j$.  Then $\{S_t\}_{t=1}^T$ is a martingale sequence with respect to the data's natural filtration so long as $\mathbb{E}[|s_t|]<\infty.$  Under the \cite{Brown(71)} martingale central limit theorem the 
$$
\langle S,S\rangle_T^{-1/2} S_T \xrightarrow{D} N(0,I),
$$
noting that 
$\langle S,S\rangle_T = 
\sum_{t=1}^T \widetilde{\theta}'_{t|t-1} \Sigma_{t|t-1} (\widetilde{\theta}'_{t|t-1})^{\tt T}$, where $\widetilde{\theta}'_{t|t-1}:=\widetilde{\theta}'_{t|t-1}(\omega^*)$ and $\Sigma_{t|t-1}:=\Sigma_{t|t-1}(\omega^*)$.

(b) By a multivariate mean value expansion $0 = S_T - \bar{H}_T (\hat{\omega}_{\tt MLE}-\omega^*)$,
where $\bar{H}_T = \int_0^1 H_T(\omega^* + u(\hat{\omega}_{\tt MLE} - \omega^*)) \, du$. Hence if $\langle S,S \rangle_T$ is invertible then 
\begin{align*}
\{\langle S,S\rangle_T^{-1/2} \bar{H}_T \}(\hat{\omega}_{\tt MLE} - \omega^*) &= \langle S,S\rangle_T^{-1/2} S_T 
\xrightarrow{D} N(0,I).
\end{align*}
If $\bar{H}_T$ is invertible, then 
$
W_T = \bar{H}^{-1}_T \langle S,S\rangle_T \bar{H}^{-1}_T,
$
is an infeasible approximation to the variance-covariance matrix of the estimator.

(c) In practice we use the estimator 
$
\hat{W}_T = H^{-1}_T(\hat{\omega}_{\tt MLE}) [\hat{S},\hat{S}]_T H^{-1}_T(\hat{\omega}_{\tt MLE}),$ where $ [\hat{S},\hat{S}]_T=\sum_{t=1}^T \hat{s}_t \hat{s}_t^{\tt T},$ with $\hat{s}_t = s_t(\hat{\omega}_{\tt MLE};Y_{1:t}),
$
as a feasible approximation to the variance-covariance matrix of the estimator.  
Hence this is a sandwich matrix of the tradition we see in quasi-likelihood estimation, going back to at least \cite{Cox(61)} and \cite{Huber(67)}.

\end{remark}

\subsection{Two step alternative to the MLE}

An alternative to the MLE is a two step procedure (e.g. \cite{NeweyMcFadden(94)}, \cite{EngleMezrich(96)}, \cite{FrancqHorvathZakoian(13)}), which has the following structure. 

\begin{definition}[2-step estimator]\label{def:2-step estimator} For a strictly stationary stochastic process $\{Y_t\}_{t\ge 1}$, assume $\mathbb{E}[h(Y_1)] \in \mathcal{H}$ exists.  For the exponentially weighted estimands under the stable frame minimal ${\tt CEF}(\theta,h,\psi)$, let $\phi$ denote any additional static parameters from the model density. Write  
$$
\hat{\omega}_{\tt 2Step} = \left(\begin{matrix}
\widehat{\mathbb{E}[h(Y_1)]}_{\tt 2Step} \\
\widehat{\alpha}_{\tt 2Step}\\
\widehat{\lambda}_{\tt 2Step}\\
\widehat{\phi}_{\tt 2Step}
\end{matrix}\right). 
$$ Then compute: 
\begin{enumerate}
\item The method of moments estimator:
$$
\widehat{\mathbb{E}[h(Y_1)]}_{\tt 2Step} = \frac{1}{T} \sum_{t=1}^T h(Y_t).
$$

\item The likelihood based estimator:  
$$
\{\widehat{\alpha}_{\tt 2Step},\widehat{\lambda}_{\tt 2Step}, \widehat{\phi}_{\tt 2Step}\} =
\underset{\{\alpha,\lambda,\phi\} \in (0,1)^2 \times \Phi}{\arg }\max \ \sum_{t=1}^T l_t(\widehat{\mathbb{E}[h(Y_1)]}_{\tt 2Step},\alpha,\lambda,\phi),
$$
\end{enumerate}
\end{definition}

Here the numerical optimization of the likelihood is only $2+\dim(\phi)$ dimensional. This estimation strategy can be very attractive when $\mathbb{E}[h(Y_1)]$ is high dimensional.

The two step procedure can be viewed as a method of moments estimator \citep{Pearson(1894)}, based around a random function 
$$
S_t(\omega;y_{{1:t}}) := \sum_{j=1}^t s_j(\omega;y_{1:j}),\quad s_t(\omega;y_{1:t}) := \left(\begin{matrix} h(Y_t) - \mathbb{E}[h(Y_1)] \\
\partial l_t(\omega)/\partial \alpha  \\
\partial l_t(\omega)/\partial \lambda \\
\partial l_t(\omega)/\partial \phi 
\end{matrix}  \right)
$$
for $t=1,...,T,$ 
then $\hat{\omega}_{\tt 2Step}$ is a method of moments estimator which solves $S_T(\hat{\omega}_{\tt 2Step};Y_{1:T})=0,$ while $\mathbb{E}[s_t(\omega^*;Y_{1:t})]=0$ for each $t=1,...,T$.  The corresponding Hessian is 
$$
H_T(\omega) := -\sum_{t=1}^T \frac{\partial s_t(\omega)}{\partial \omega^{\tt T}} = \sum_{t=1}^T \left( \begin{matrix} I \ \ 0 \ \ 0 \\
-\partial^2 l_t(\omega)/\partial \alpha \partial \omega^{\tt T}  \\
-\partial^2 l_t(\omega)/\partial \lambda \partial \omega^{\tt T}  \\
-\partial^2 l_t(\omega)/\partial \phi \partial \omega^{\tt T}  \\
\end{matrix}\right).
$$

\begin{remark} (a) Write $s_t = s_t(\omega^*;Y_{1:t})$ and $S_t = S_t(\omega^*;Y_{1:t})$. The last three elements of $\{s_t\}$ are martingale differences.  The remaining elements, $h(Y_t)-\mathbb{E}[h(Y_1)]$, have an unconditional zero mean but possibly substantial time series memory. 
Write out a corresponding CLT as  
$$
V_T^{-1/2} S_T(\omega^*) \xrightarrow{D} N(0,I),\quad \text{where} \quad V_T = {\mathrm V}(S_T(\omega^*)),
$$
which needs the memory in $\{s_t\}$ to be controlled.  The martingale difference elements have limited memory, the problem is the time series of $\{h(Y_t)-\mathbb{E}[h(Y_1)]\}$.  A basic way of generating a CLT for these type of objects is to assume the $h(Y_t)$ series is strictly stationary and exhibits $m$-dependence (e.g. \cite{janson2021central}). 

(b) By a multivariate mean value expansion $0 = S_T - \bar{H}_T (\hat{\omega}_{\tt 2Step} - \omega^*),$ where $\bar{H}_T = \int_0^1 H_T(\omega^* + u(\hat{\omega}_{\tt 2Step} - \omega^*)) \, du$. Thus 
\begin{align*}
V_T^{-1/2} \bar{H}_T (\hat{\omega}_{\tt 2Step} - \omega^*) &= V_T^{-1/2} S_T 
\xrightarrow{D} N(0,I).
\end{align*}
So the infeasible variance matrix for the two step estimator is 
$
W_T = \bar{H}^{-1}_T V_T \bar{H}^{-1}_T.
$

(c) In practice we use the estimator of the covariance matrix
$
\hat{W}_T = H^{-1}_T(\hat{\omega}_{\tt 2Step}) \hat{V}_T H^{-1}_T(\hat{\omega}_{\tt 2Step})^{\tt T},
$
where $\hat{V}_T$ approximates ${\mathrm V}(\hat{S}_T)$.  The latter can be estimated using $T$ times a long-run variance of the time series $\hat{s}_1,...,\hat{s}_T$, where $\hat{s}_t = s_t(\hat{\omega}_{\tt 2Step};Y_{1:t})$. 
\end{remark}

\section{Empirical example}\label{Sec: Empirical example}

\subsection{Dirichlet based frame: household financial situation}\label{sec:household}

We apply the exponentially weighted predictor and smoother to monthly data on household financial situation expectations from the University of Michigan Survey of Consumers, covering January 1978 to September 2025 ($T = 573$ observations). The survey asks respondents two questions: first, whether they are better off or worse off financially than a year ago, and second, whether they expect to be better off or worse off a year from now. The combination of responses yields seven mutually exclusive categories:
\vspace{-2mm}
\begin{enumerate}[label=(\roman*), itemsep=0pt, parsep=0pt]
    \item Continuous Increase (better off now and expect to be better off),
    \item Intermittent Increase (one period better, one period same),
    \item Remain Unchanged (same in both periods),
    \item Intermittent Decline (one period worse, one period same),
    \item Continuous Decline (worse off now and expect to be worse off),
    \item Mixed Change (improvement followed by decline or vice versa),
    \item Don't Know/No answer.
\end{enumerate}
\vspace{-2mm}

For each month, only the sample proportions (percentages) falling into each category are reported, not the underlying individual responses or sample sizes.\footnote{While monthly sample sizes are available, our framework would require within-month i.i.d. 
responses, which is unlikely given the clustered 
sampling design. We leave them out of our analysis.}

We model the observed proportions $y_t = (y_{1,t}, \ldots, y_{7,t})^{\tt T}$ as realizations from a time-varying seven category Dirichlet distribution 
$
Y_t \mid Y_{1:t-1} \sim \text{Dirichlet}(\widetilde{\theta}_{t|t-1}),$ with $ \widetilde{\theta}_{t|t-1} = (\widetilde{\theta}_{1,t|t-1}, \ldots, \widetilde{\theta}_{7,t|t-1})^{\tt T},
$
with concentration parameters $\widetilde{\theta}_{j, t|t-1} > 0$ for all $j \in \{1, \ldots, 7\}$ and all $t$. Each observation is treated as a single draw of a proportion vector. The predictor is based on a Dirichlet frame. 

The exponentially weighted predictor and smoother are computed via the recursions in Example~\ref{thm:MLE} using the Dirichlet distribution's CEF parameterization given in Section~\ref{sec:simulating CEF examples}. 
Following the two-step quasi-likelihood procedure from Definition~\ref{def:2-step estimator}, we  estimate the centering parameter $\mathbb{E}[h(Y_1)]$ by the sample mean $\widehat{\mathbb{E}[h(Y_1)]} = T^{-1} \sum_{t=1}^T (\log y_{1,t}, \ldots, \log y_{7,t})^{\tt T}.$ Next, we fix $\widehat{\mathbb{E}[h(Y_1)]}$ and estimate $(\alpha, \lambda)$ as maximizing 
$\sum_{t=1}^T \log f(y_t; \widetilde{\theta}_{t|t-1}(\alpha, \lambda)).$

For the household financial situation data with seven categories, we obtain the parameter estimates $(\widehat{\mathbb{E}[h(Y_1)]}^{\tt T}, \widehat{\alpha}, \widehat{\lambda}) \approx (-1.76, -1.41, -1.78, -1.77, -2.73, -2.23, -3.53, 0.95, 0.64)$. The high anchoring estimate $\widehat{\alpha} = 0.95$ indicates little long-run effect, suggesting that shifts in household financial expectations, once they occur, tend to persist a long time. The discount parameter $\widehat{\lambda} =0.64$ implies a half-life of approximately $1.56$ months for the exponential weights. Following Remark~\ref{remark:SS CEF}(e), the steady state process is a vector {\tt ARMA}(1,1)-MD process, with autoregressive root $\widehat{\lambda}/\{1-\widehat{\alpha} (1-\widehat{\lambda})\}\approx 0.97$, and moving average root $-\widehat{\lambda} = -0.64$. Hence, the process exhibits substantial memory, with substantial root cancellation.

Figure~\ref{fig:household_financial_results} presents the estimation results. Each panel shows the evolution of one expectation category over the 47-year period. The gray circles represent observed monthly proportions $y_{k,t}$. The light colored lines show the exponentially weighted predictor, the dark colored lines the exponentially weighted smoother. Both track major shifts in household financial expectations. There are the sharp declines in both ``Continuous Increase'' and ``Intermittent Increase'' during the 2008 financial crisis, and, more strikingly, since around 2018. This deterioration in optimism about household finances is mirrored by corresponding increases in negative expectations: ``Intermittent Decline'' and ``Continuous Decline'' roughly track their inverse patterns. The categories ``Remain Unchanged'' and ``Don't Know/No Answer'' stay relatively stable at 15-20\% and 3\%, respectively, throughout the sample. Appendix~\ref{app:montecarlo} uses the empirical results to provide simulation based evidence for the performance of the hyperparameter estimators.

\begin{figure}[htbp] 
\centering 
\begin{tabular}{@{}c@{\hspace{0.2cm}}c@{\hspace{0.2cm}}c@{\hspace{0.2cm}}c@{}}
\textbf{Cont. Increase} & \textbf{Int. Increase} & \textbf{Unchanged} & \textbf{Int. Decline} \\
\includegraphics[width=0.24\linewidth]{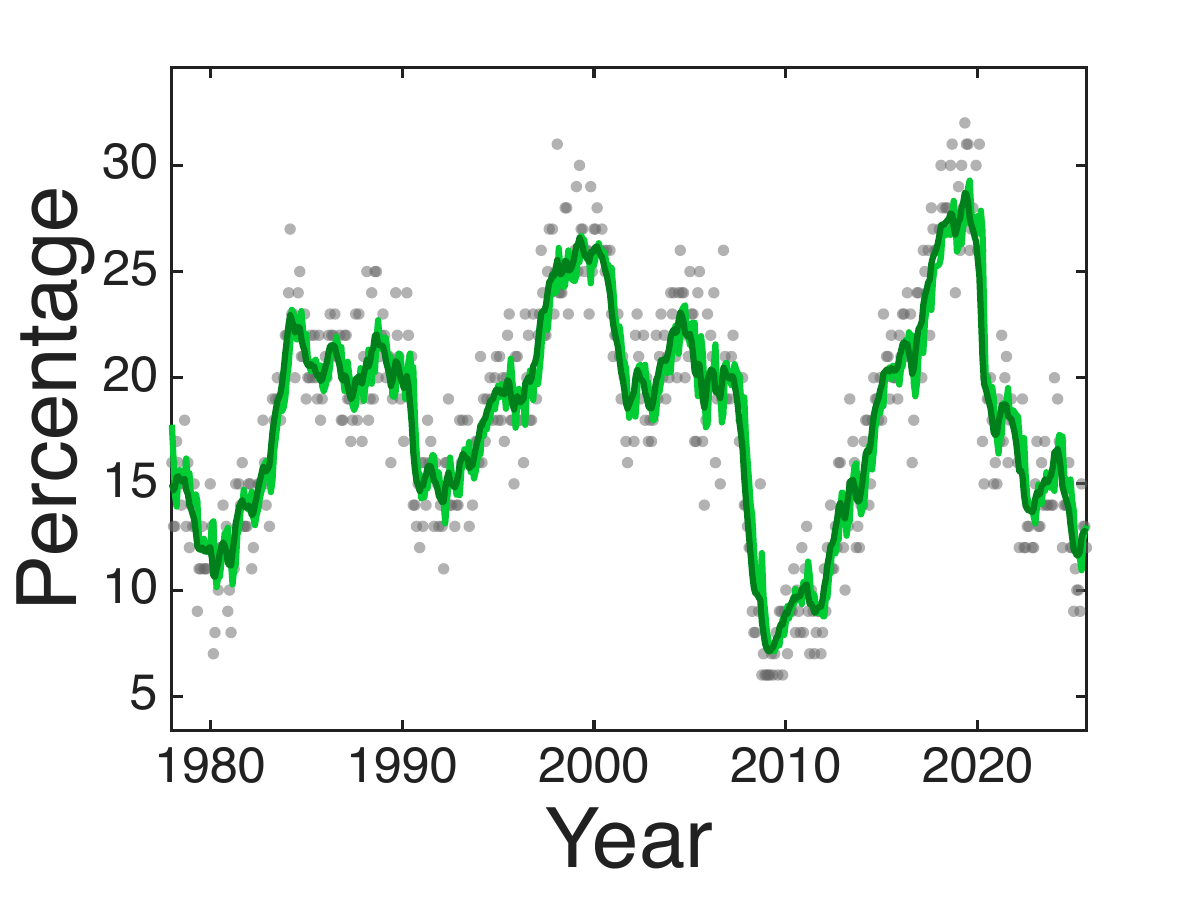} & 
\includegraphics[width=0.24\linewidth]{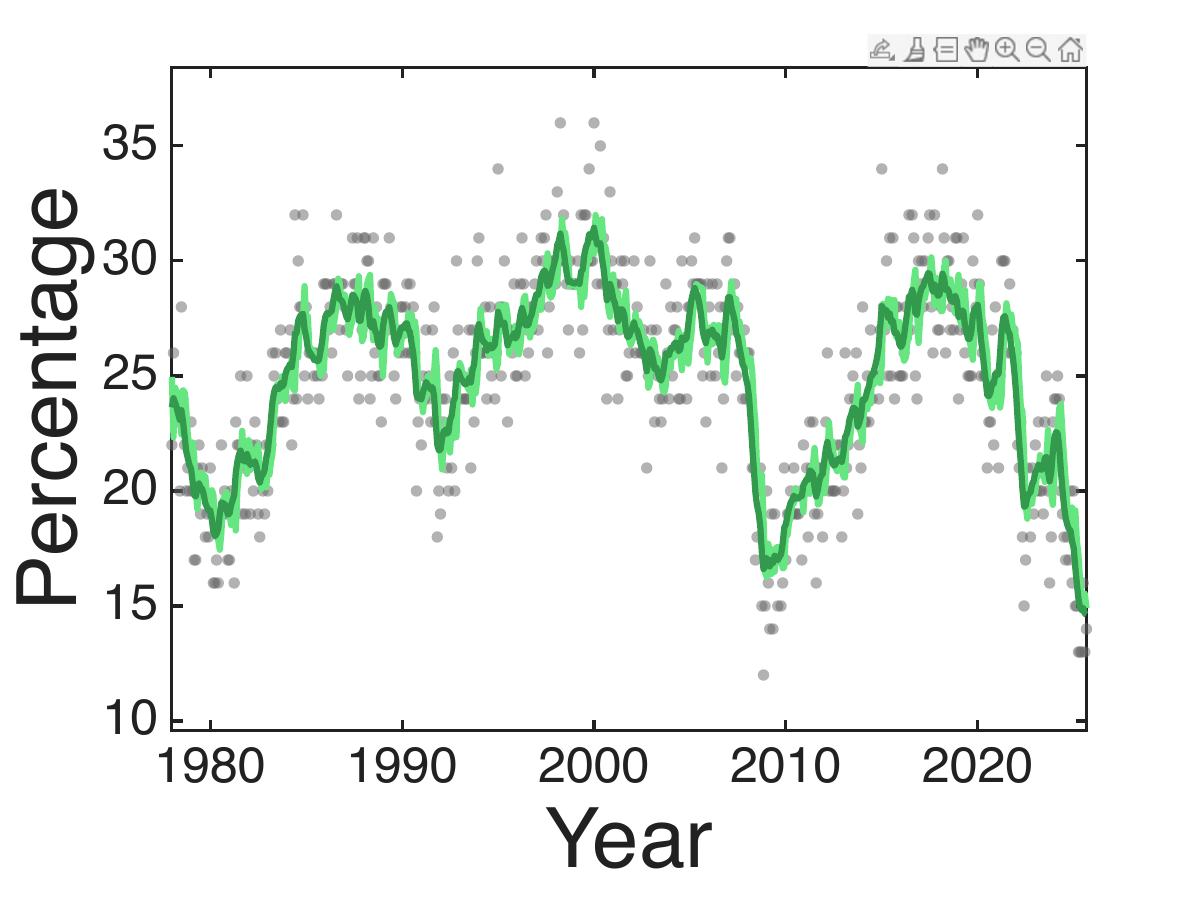} &
\includegraphics[width=0.24\linewidth]{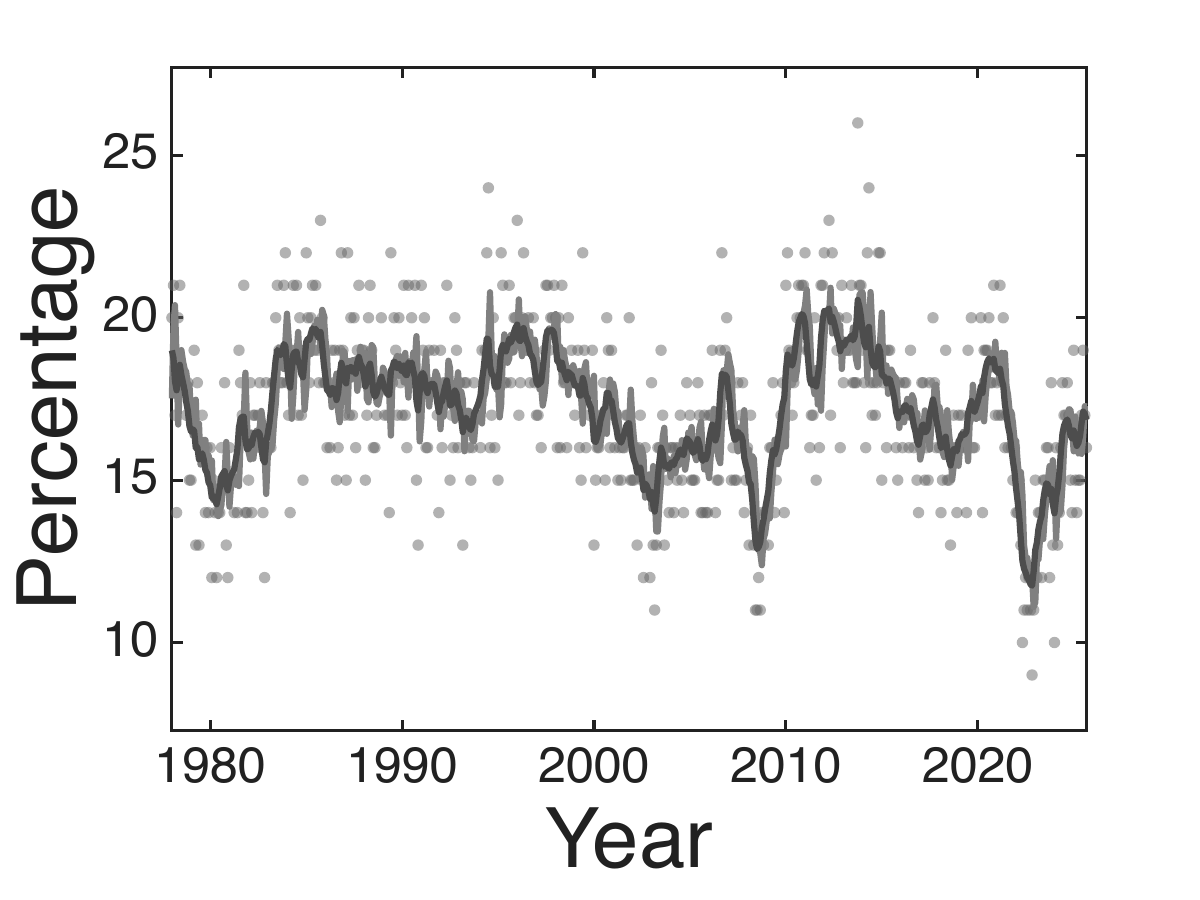} &
\includegraphics[width=0.24\linewidth]{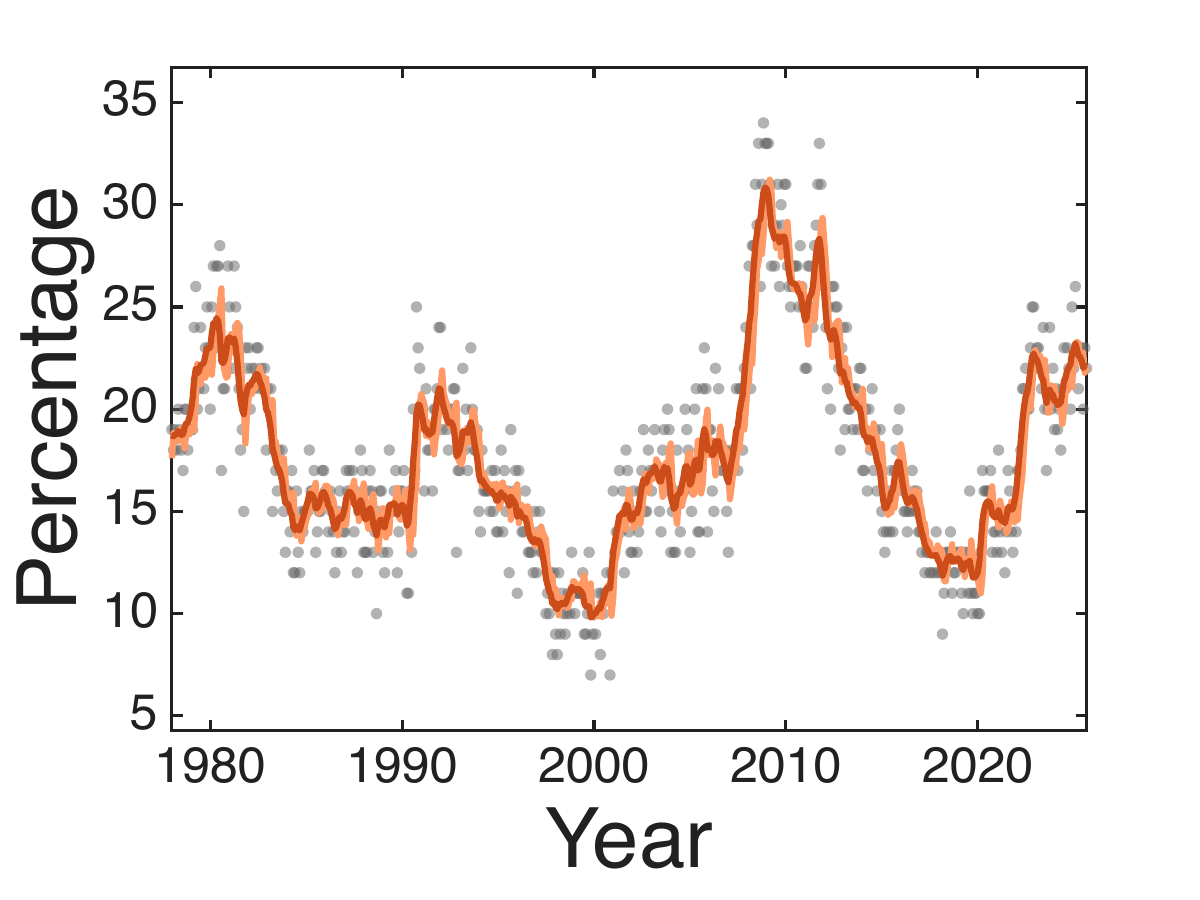} \\[3mm]
\textbf{Cont. Decline} & \textbf{Mixed} & \textbf{DK/NA} \\
\includegraphics[width=0.24\linewidth]{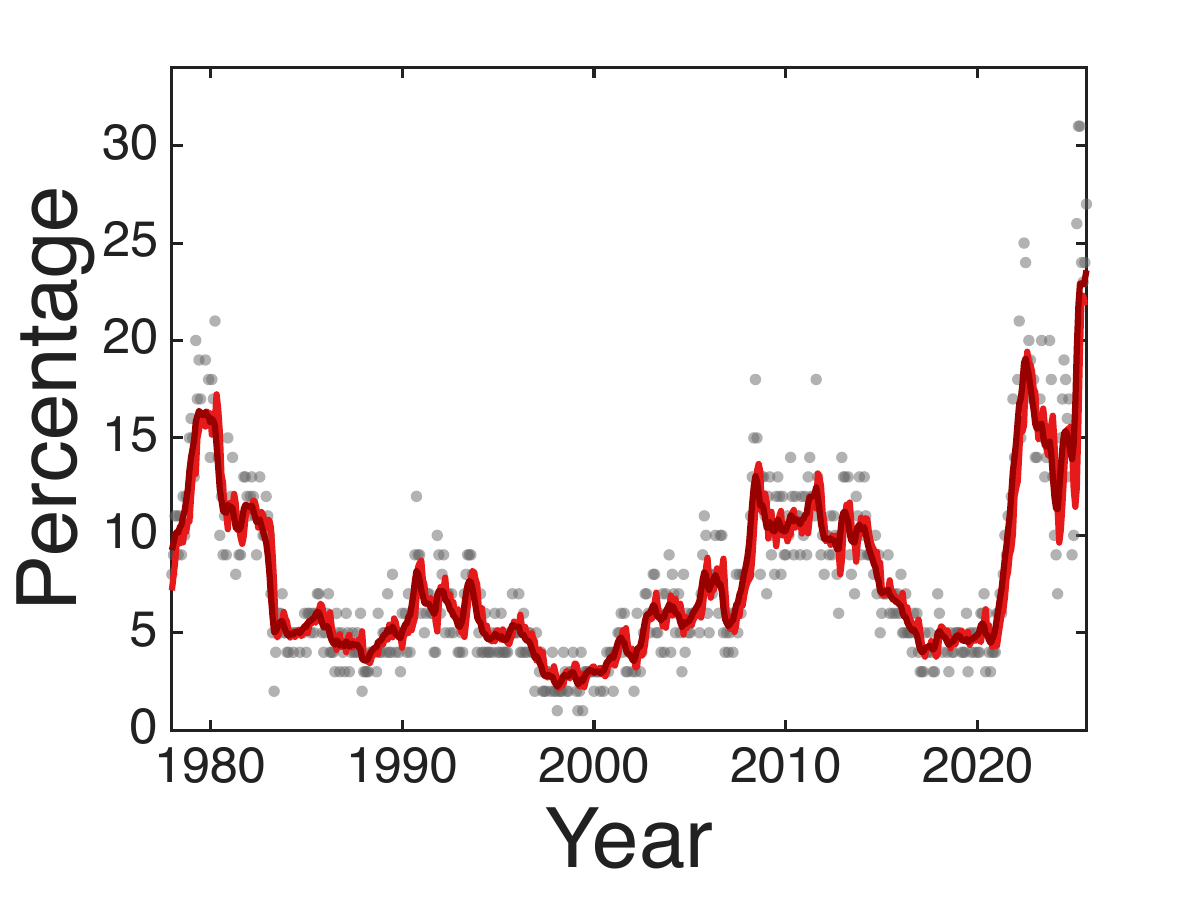} &
\includegraphics[width=0.24\linewidth]{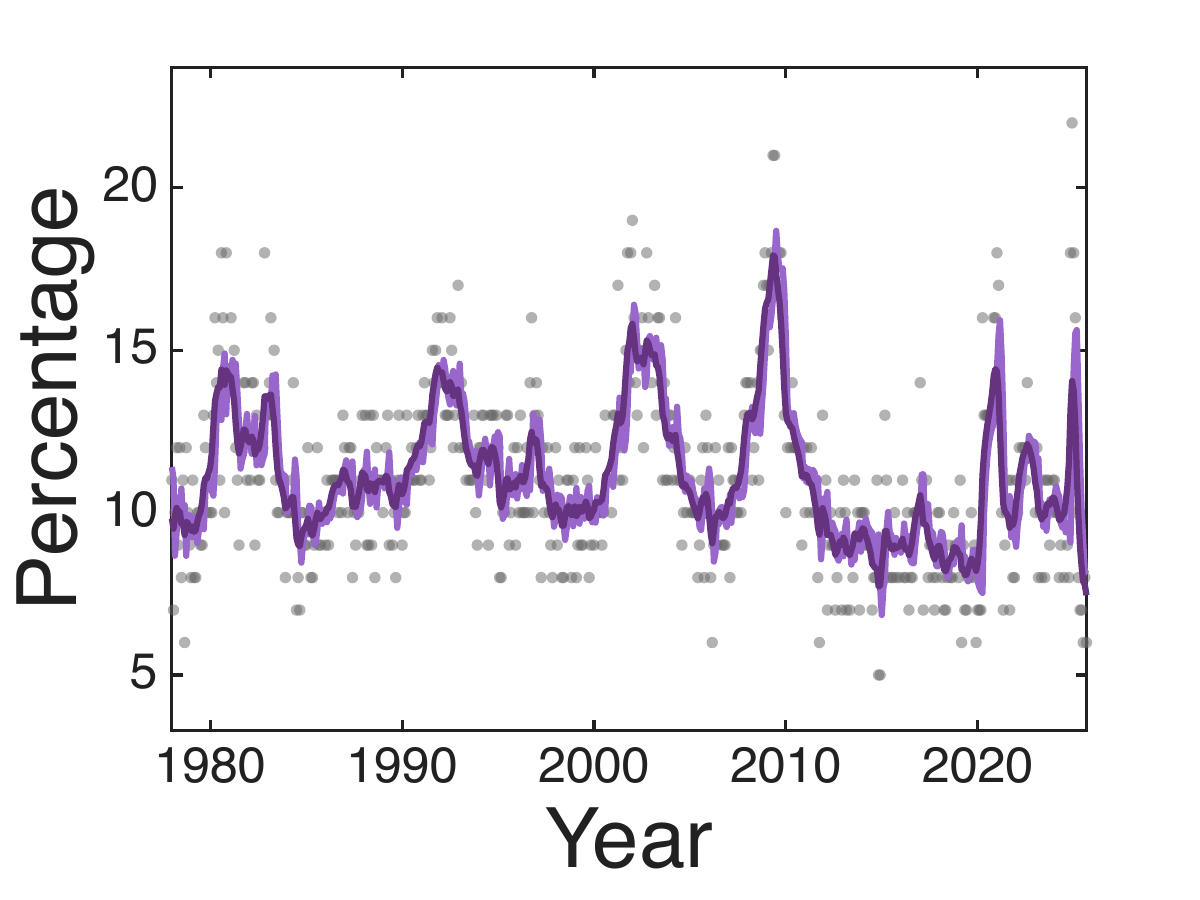} &
\includegraphics[width=0.24\linewidth]{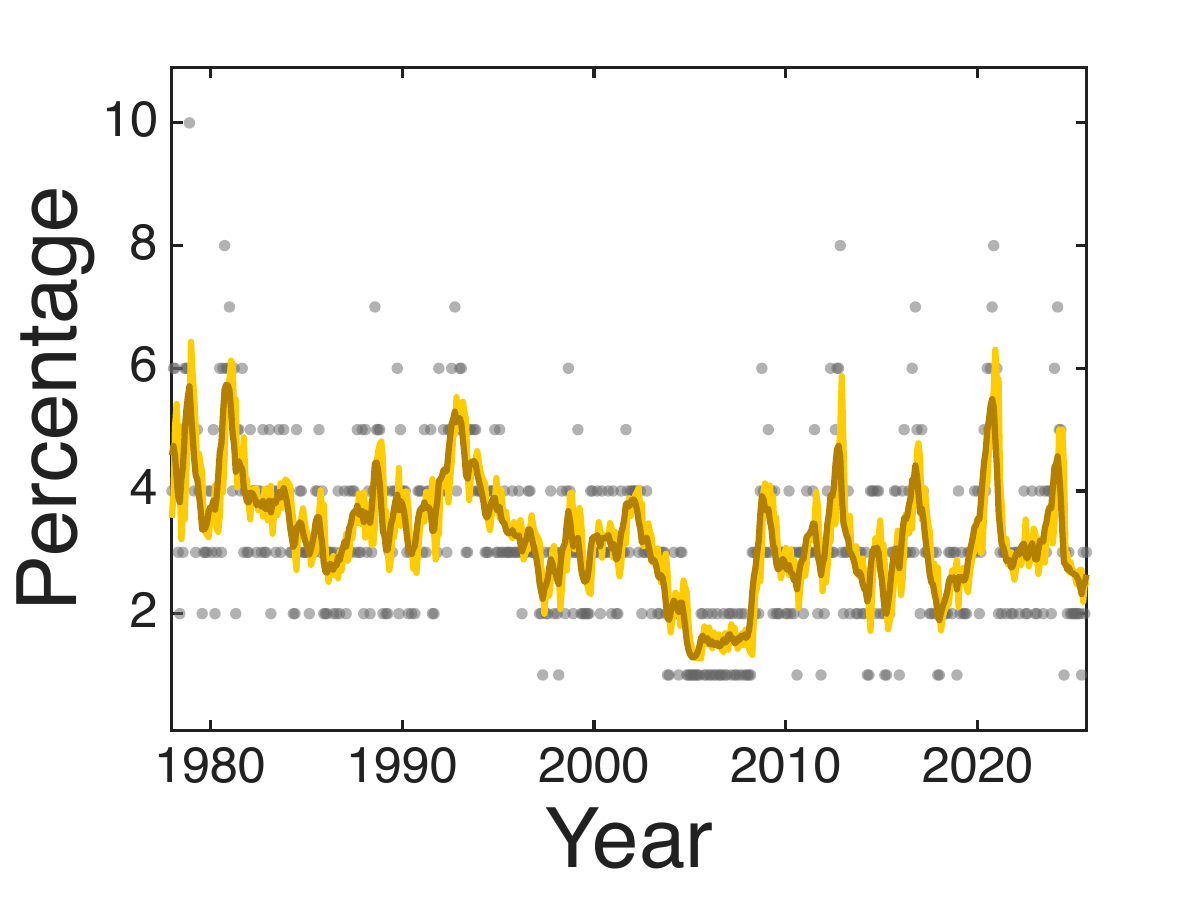}
\end{tabular} 
\vspace{-3mm} 
\caption{Household financial situation:  University of Michigan Survey of Consumers (Jan 1978 to Sept 2025). Categories combine current versus past financial situation, and expectations for the year ahead (Cont. = Continuous, Int. = Intermittent, DK = Don't Know, NA = No Answer): Continuous Increase means better off now and expecting better, Intermittent Increase indicates improvement in one period only. Each panel shows, for one response category, the observed monthly proportions $y_{k,t}$ (gray circles), the predictor $y_{k, t|t-1} = \widetilde{\theta}_{k, t|t-1}/\sum_{j=1}^7 \widetilde{\theta}_{j, t|t-1}$ (light line) and the smoother $y_{k, t|T} = \widetilde{\theta}_{k, t|T}/\sum_{j=1}^7 \widetilde{\theta}_{j, t|T}$ (dark line). Estimates: $(\widehat{\alpha}, \widehat{\lambda}) = (0.95, 0.64).$}
\label{fig:household_financial_results}
\end{figure}

\section{Conclusions}\label{sect:conc}

The exponentially weighted estimands combined with the {\tt CEF} delivers simple exact filters, predictors and smoothers.  They are in the same statistical spirit as generalized linear models (e.g. \cite{McCullaghNelder(89)}) and are relatively easy to fit to data using a quasi-likelihood approach either all at once or using a two step approach.       

The exponentially weighted estimands framework can be used on non-exponential family models and even for general loss functions.  In such cases the optimization has to be carried out numerically or 
through a sequential approximation.



\section*{Acknowledgments}
We are grateful for the comments and encouragement from Dick van Dijk, Siem Jan Koopman, Rutger-Jan Lange and Bram van Os.\ Donker van Heel conducted this research during a visiting research appointment at Harvard University's Department of Statistics, supported by the Fulbright Program. The authors report there are no competing interests to declare.

\bibliographystyle{chicago}
\bibliography{neil.bib, Simon.bib}

@STRING{American="The American Statistican"}

@STRING{Annals="Ann. Statist."}

@preamble{ "\providecommand\SortNoop[1]{}" }

@STRING{Annals="Annals of Statistics"}

@STRING{OUP="Oxford University Press"}

@article{Davis03072021,
author = {Richard A. Davis and Konstantinos Fokianos and Scott H. Holan and Harry Joe and James Livsey and Robert Lund and Vladas Pipiras and Nalini Ravishanker},
title = {Count Time Series: A Methodological Review},
journal = {Journal of the American Statistical Association},
volume = {116},
number = {535},
pages = {1533--1547},
year = {2021}}

@article{lange2024robust,
  title={Implicit score-driven filters for time-varying parameter models},
  author={Lange, Rutger-Jan and van Os, Bram and van Dijk, Dick JC},
  journal={\textit{Preprint}},
  year={2024},
  note={\url{https://arxiv.org/abs/2512.02744}}
}

@article{donkervanheel2025stability,
  title={Gradient-based filtering under misspecification: Stability and error bounds},
  author={Donker van Heel, Simon and Lange, Rutger-Jan and van Os, Bram and van Dijk, Dick J C},
  journal={\textit{Preprint}},
  year={2025},
  note={\url{https://arxiv.org/abs/2502.05021}}
}

@article{luxenberg2024exponentially,
  title={Exponentially Weighted Moving Models},
  author={Luxenberg, Eric and Boyd, Stephen},
  journal={\textit{Preprint}},
  year={2024},
  note={\url{https://arxiv.org/abs/2404.08136}}
}

@article{janson2021central,
  title={A central limit theorem for m-dependent variables},
  author={Janson, Svante},
  journal={\textit{Preprint}},
  year={2021},
  note={\url{https://arxiv.org/abs/2108.12263}}
}

@article{toulis2017asymptotic,
  title={Asymptotic and finite-sample properties of estimators based on stochastic gradients},
  author={Toulis, Panos and Airoldi, Edoardo M},
  journal={Annals of Statistics},
  volume={45},
number={4},
  pages={1694--1727},
  year={2017}
}

@incollection{artemova2022score,
  title={Score-driven models: Methodology and theory},
  author={Artemova, Mariia and Blasques, Francisco and van Brummelen, Janneke and Koopman, Siem Jan},
  booktitle={Oxford Research Encyclopedia of Economics and Finance},
  year={2022},
  publisher={OUP}
}

@article{toulis2015scalable,
  title={Scalable estimation strategies based on stochastic approximations: {C}lassical results and new insights},
  author={Toulis, Panos and Airoldi, Edoardo M},
  journal={Statistics and Computing},
  volume={25},
  pages={781--795},
  year={2015},
  publisher={Springer}
}

@inproceedings{toulis2016towards,
  title={Towards stability and optimality in stochastic gradient descent},
  author={Toulis, Panos and Tran, Dustin and Airoldi, Edo M},
  booktitle={Artificial Intelligence and Statistics},
  pages={1290--1298},
  year={2016},
  organization={PMLR}
}

@article{toulis2021proximal,
  title={The proximal \uppercase{R}obbins--\uppercase{M}onro method},
  author={Toulis, Panos and Horel, Thibaut and Airoldi, Edoardo M},
  journal={Journal of the Royal Statistical Society Series B: Statistical Methodology},
  volume={83},
  number={1},
  pages={188--212},
  year={2021},
  publisher={OUP}
}

@article{blasques2016weighted,
  title={Weighted maximum likelihood for dynamic factor analysis and forecasting with mixed frequency data},
  author={Blasques, Francisco and Koopman, Siem Jan and Mallee, M and Zhang, Zhaoyong},
  journal={Journal of Econometrics},
  volume={193},
  number={2},
  pages={405--417},
  year={2016},
  publisher={Elsevier}
}

@article{dixon1997modelling,
  title={Modelling association football scores and inefficiencies in the football betting market},
  author={Dixon, Mark J and Coles, Stuart G},
  journal={Journal of the Royal Statistical Society: Series C (Applied Statistics)},
  volume={46},
  number={2},
  pages={265--280},
  year={1997},
  publisher={Wiley Online Library}
}

@article{hu2002weighted,
  title={The weighted likelihood},
  author={Hu, Feifang and Zidek, James V},
  journal={Canadian Journal of Statistics},
  volume={30},
  number={3},
  pages={347--371},
  year={2002},
  publisher={Wiley Online Library}
}

\clearpage
\appendix

\begin{center}
\Large \textbf{Online supplement to: \\
[6pt]
``Exponentially weighted estimands and the exponential family: \\
Filtering, prediction and smoothing''}\\[4pt]
\normalsize
S. W. Donker van Heel and N. Shephard
\end{center}

\pagenumbering{arabic}
\renewcommand*{\thepage}{S\arabic{page}}
\setcounter{page}{1}

\setcounter{equation}{0}
\renewcommand{\theequation}{\Alph{section}.\arabic{equation}}
\setcounter{table}{0}
\renewcommand{\thetable}{\Alph{section}.\arabic{table}}
\setcounter{figure}{0}
\renewcommand{\thefigure}{\Alph{section}.\arabic{figure}}

\section{Comparison to score-driven filters}\label{app:scoredriven-full}
 
We compare our predictor to the score-driven filter of \cite{Harvey(13)} and \cite{CrealKoopmanLucas(13)} for the univariate stable frame case with $\psi_t = n_t \psi$ and $n_t = 1$, assuming strictly stationary data (see Example~\ref{thm:MLE}).
 
For the ${\tt CEF}(\theta, h, \psi)$, the score is $\nabla_t(\theta) = h(y_t) - \psi'(\theta)$ and the Fisher information is $\mathcal{I}(\theta) = \psi''(\theta)$. With the most commonly used scaling for the score, inverse Fisher information scaling,\footnote{More generally, the score may be scaled by $[\mathcal{I}(\theta)]^{-\zeta}$ with $\zeta \in \{0, 1/2, 1\}$; see \cite{artemova2022score}. We focus on inverse Fisher scaling ($\zeta = 1$). For the Gaussian, $\psi''(\theta)$ is constant, so all three scalings are equivalent up to a rescaling of $A$.} the score-driven filter uses the prediction-to-prediction recursion
\begin{equation}\label{eq:SD recursion}
\theta^{\text{sd}}_{t|t-1} = \omega + A \, [\psi''(\theta^{\text{sd}}_{t-1|t-2})]^{-1} [h(y_{t-1}) - \psi'(\theta^{\text{sd}}_{t-1|t-2})] + \Phi \, \theta^{\text{sd}}_{t-1|t-2},
\end{equation}
where $\omega$, $A$, $\Phi$ are static hyperparameters (see e.g.~\citealp[Eq.\ 2]{CrealKoopmanLucas(13)}) and the superscript $\text{sd}$ denotes score driven. The predictor $\tilde{\theta}_{t|t-1} = (\psi')^{-1}(\bar{m}_{t|t-1})$ from Example~\ref{thm:MLE} admits a linear recursion in $\widetilde{\mu}_{t|t-1} = \psi'(\widetilde{\theta}_{t|t-1})$; in steady state, 
\begin{equation}\label{eq:our predictor SS}
\widetilde{\mu}_{t|t-1} = \underbrace{\frac{(1-\alpha)(1-\lambda)}{1 - \alpha(1-\lambda)}}_{=:\,c_1} \mathbb{E}[h(Y_1)] + \underbrace{\frac{\alpha\lambda(1-\lambda)}{1 - \alpha(1-\lambda)}}_{=:\,c_2} h(y_{t-1}) + \lambda \, \widetilde{\mu}_{t-1|t-2},
\end{equation}
where $c_1 + c_2 + \lambda = 1$. To compare, we transform~\eqref{eq:SD recursion} to $\mu$-space.
 
\noindent\textit{Equivalence in the Gaussian case.}  The Gaussian distribution with known variance $\sigma^2$ is the only ${\tt CEF}$ member for which $\psi'$ is linear. Here $h(y) = y/\sigma^2$, $\psi'(\theta) = \theta/\sigma^2$, and $\psi''(\theta) = 1/\sigma^2$. As the Fisher information is constant, the scaled score in~\eqref{eq:SD recursion} simplifies to $[\psi''(\theta)]^{-1}[h(y_{t-1}) - \psi'(\theta)] = y_{t-1} - \theta$. Dividing~\eqref{eq:SD recursion} by $\sigma^2$ and writing $\mu^{\text{sd}}_{t|t-1} = \theta^{\text{sd}}_{t|t-1}/\sigma^2$ gives
\begin{equation}\label{eq:SD mu recursion}
\mu^{\text{sd}}_{t|t-1} = \psi'(\omega) + A \, h(y_{t-1}) + (\Phi - A) \, \mu^{\text{sd}}_{t-1|t-2}.
\end{equation}
Comparing~\eqref{eq:SD mu recursion} to~\eqref{eq:our predictor SS}, we match unconditional means. The unconditional mean of~\eqref{eq:SD mu recursion}, obtained by setting $\mu^{\text{sd}}_{t|t-1} = \mu^{\text{sd}}_{t-1|t-2} = \bar{\mu}^{\text{sd}}$ and taking expectations, is $\bar{\mu}^{\text{sd}} = \{\psi'(\omega) + A \, \mathbb{E}[h(Y_1)]\}/(1 - \Phi + A)$. Setting this equal to $\mathbb{E}[h(Y_1)]$ yields $\psi'(\omega) = (1 - \Phi)\mathbb{E}[h(Y_1)]$, under which the intercept in~\eqref{eq:SD mu recursion} becomes $(1 - \Phi)\mathbb{E}[h(Y_1)]$. Matching the remaining coefficients with~\eqref{eq:our predictor SS} requires $\Phi - A = \lambda$, $A = c_2$, and $1 - \Phi = c_1$. From the first two, $\Phi = \lambda + c_2$, and substituting into the third gives $c_1 + c_2 + \lambda = 1$, which holds by construction. Hence the score-driven predictor~\eqref{eq:SD recursion} coincides with ours in steady state under $\Phi = \lambda + c_2$, $A = c_2$, and $\psi'(\omega) = (1 - \Phi)\mathbb{E}[h(Y_1)]$. The same conclusion holds for the implicit score-driven filter of \cite{lange2024robust}, which uses the recursion
\begin{equation}\label{eq:ISD recursion}
\theta^{\textnormal{isd}}_{t|t-1} = \omega + H_t \, [h(y_{t-1}) - \psi'(\theta^{\textnormal{isd}}_{t-1})] + \Phi \, \theta^{\textnormal{isd}}_{t-1|t-2},
\end{equation}
where the score is evaluated at the filtered value $\theta^{\textnormal{isd}}_{t-1}$ rather than the predicted $\theta^{\textnormal{isd}}_{t-1|t-2}$. In the Gaussian case, $\psi'$ is linear and the filtered value can be solved analytically, yielding $\mu^{\textnormal{isd}}_{t-1} = \{\mu^{\textnormal{isd}}_{t-1|t-2} + (H_t/\sigma^2) \, h(y_{t-1})\}/(1 + H_t/\sigma^2)$. Substituting into~\eqref{eq:ISD recursion} and matching coefficients with~\eqref{eq:our predictor SS} gives $\Phi = \lambda + c_2$, $H_t = c_2 \sigma^2 / (1 - c_2)$ for all $t$, and $\psi'(\omega) = (1 - \Phi)\mathbb{E}[h(Y_1)]$. As noted by \cite{lange2024robust}, the explicit and implicit score-driven filters coincide in the Gaussian case for appropriately chosen hyperparameters.

\noindent\textit{Differences out of steady state.} From Example~\ref{thm:MLE}, our predictor is $\widetilde{\theta}_{t|t-1} = (\psi')^{-1}(m_{\lambda,t|t-1}/n_{\lambda,t|t-1})$. Under strict stationarity, $x_{\lambda,t} = \mathbb{E}[h(Y_1)] n_{\lambda,t}$ by Remark~\ref{remark:SS CEF}(a). Using the recursions $n_{\lambda,t} = 1 + \lambda n_{\lambda,t-1}$ and $h_{\lambda,t-1} = h(y_{t-1}) + \lambda h_{\lambda,t-2}$, we can write
\begin{align*}
\widetilde{\mu}_{t|t-1} &= \frac{(1-\alpha)\mathbb{E}[h(Y_1)] + \alpha\lambda \, h(y_{t-1}) + \lambda \, m_{\lambda,t-1|t-2}}{n_{\lambda,t|t-1}} \\
&= \underbrace{\frac{1-\alpha}{n_{\lambda,t|t-1}}}_{=:\,c_{1,t}} \mathbb{E}[h(Y_1)] + \underbrace{\frac{\alpha\lambda}{n_{\lambda,t|t-1}}}_{=:\,c_{2,t}} h(y_{t-1}) + \underbrace{\frac{\lambda \, n_{\lambda,t-1|t-2}}{n_{\lambda,t|t-1}}}_{=:\,c_{3,t}} \widetilde{\mu}_{t-1|t-2},
\end{align*}
where $c_{1,t} + c_{2,t} + c_{3,t} = 1$ for each $t$, since $(1-\alpha) + \alpha\lambda + \lambda \, n_{\lambda,t-1|t-2} = (1-\alpha)n_{\lambda,t} + \alpha\lambda n_{\lambda,t-1} = n_{\lambda,t|t-1}$. Both score-driven filters use constant steady-state coefficients $(c_1, c_2, \lambda)$ for all $t$. The ratios of our time-varying coefficients to these steady-state values are
$$
\frac{c_{1,t}}{c_1} = \frac{c_{2,t}}{c_2} = \frac{n_{\lambda,\infty|\infty}}{n_{\lambda,t|t-1}}, \qquad \frac{c_{3,t}}{\lambda} = \frac{n_{\lambda,t-1|t-2}}{n_{\lambda,t|t-1}},
$$
where $n_{\lambda,\infty|\infty} := (1-\alpha(1-\lambda))/(1-\lambda)$. These satisfy $c_{1,t}/c_1 = c_{2,t}/c_2 > 1$ and $c_{3,t}/\lambda < 1$ out of steady state. That is, relative to the score-driven filter, our predictor places more weight on the anchor $\mathbb{E}[h(Y_1)]$ and on the most recent observation $h(y_{t-1})$, and less weight on the previous predictor $\widetilde{\mu}_{t-1|t-2}$. These differences decay as $t \to \infty$, with slower convergence when $\lambda$ is closer to 1. Table~\ref{tab:coefficients} illustrates this for $\lambda = 0.95$ and $\alpha = 0.5$.
 
\begin{table}[htbp]
\centering
\caption{Ratio of time-varying to steady-state coefficients for $\lambda = 0.95$, $\alpha = 0.5$.}
\label{tab:coefficients}
\begin{tabular}{cccc}
\toprule
$t$ & $n_{\lambda,t|t-1}$ & $c_{1,t}/c_1 = c_{2,t}/c_2$ & $c_{3,t}/\lambda$ \\
\midrule
1 & 0.50 & 39.0 & 0.00 \\
2 & 1.45 & 13.4 & 0.34 \\
3 & 2.35 & 8.3 & 0.62 \\
5 & 3.94 & 4.9 & 0.80 \\
10 & 8.14 & 2.4 & 0.93 \\
$\infty$ & 19.5 & 1.0 & 1.00 \\
\bottomrule
\end{tabular}
\end{table}


\noindent\textit{Non-Gaussian comparison.} For non-Gaussian $\textnormal{CEF}$ members, $\psi'$ is nonlinear and score-driven $\theta$-recursions no longer reduce to a linear recursion on $\mu$. We simulate data from the Poisson ($\psi'(\theta) = \exp(\theta)$), Pareto ($\psi'(\theta) = -1/\theta + \log m$), and Bernoulli ($\psi'(\theta) = \exp(\theta)/(1+\exp(\theta))$) distributions following Assumption~\ref{def:DGP4}, using $\lambda = 0.93$, $T = 2{,}000$, and the centering values ($\mathbb{E}[h(Y_1)]$) from Section~\ref{sec:simulating CEF examples}. The coefficients $(\omega, A, \Phi)$ in the score-driven filter~\eqref{eq:SD recursion} are set via the Gaussian steady-state mapping derived above, with $\theta^{\textnormal{sd}}_{1|0} = \widetilde{\theta}_{1|0} = (\psi')^{-1}(\mathbb{E}[h(Y_1)])$, so that differences are mostly driven by the nonlinearity of $\psi'$. The data and our predictor are identical to those in Figures~\ref{fig:CEFunistateExact}--\ref{fig:CEFunistate_additional}.

Figure~\ref{fig:sd_comparison} shows the results for $\alpha = 0.7$ (top) and $\alpha = 0.95$ (bottom). When $\alpha = 0.7$, the score-driven predictor accurately tracks the conditional mean for all three distributions. When $\alpha = 0.95$, the score-driven predictor remains close to the conditional mean for the Poisson distribution, but overshoots briefly for the Pareto distribution around $t = 1{,}179$, where the score-driven conditional mean reaches $13.17$ while the observation is $2.16$ and our conditional mean is $2.14$. For the Bernoulli distribution, the score-driven filter overshoots the data at $t = 834$ and does not recover the conditional mean path, i.e., it has diverged. The score-driven filter tracks the conditional mean well in most cases, but can become unstable when the scaled score $[\psi''(\theta)]^{-1}[h(y) - \psi'(\theta)]$ in~\eqref{eq:SD recursion} is not Lipschitz continuous and the coefficient multiplying it is not small; see \cite{donkervanheel2025stability} for formal stability conditions for score-driven filters.

\begin{figure}[htbp]
    \centering
    \begin{tabular}{@{}c@{\hspace{-0.25cm}}c@{\hspace{-0.25cm}}c@{\hspace{-0.25cm}}c@{}}
        & \textbf{Poisson} & \textbf{Pareto} & \textbf{Bernoulli} \\
        \raisebox{2.0cm}[0pt][0pt]{{\small $\alpha=0.70$\hspace{0.3cm}}} &
        \includegraphics[width=0.32\linewidth]{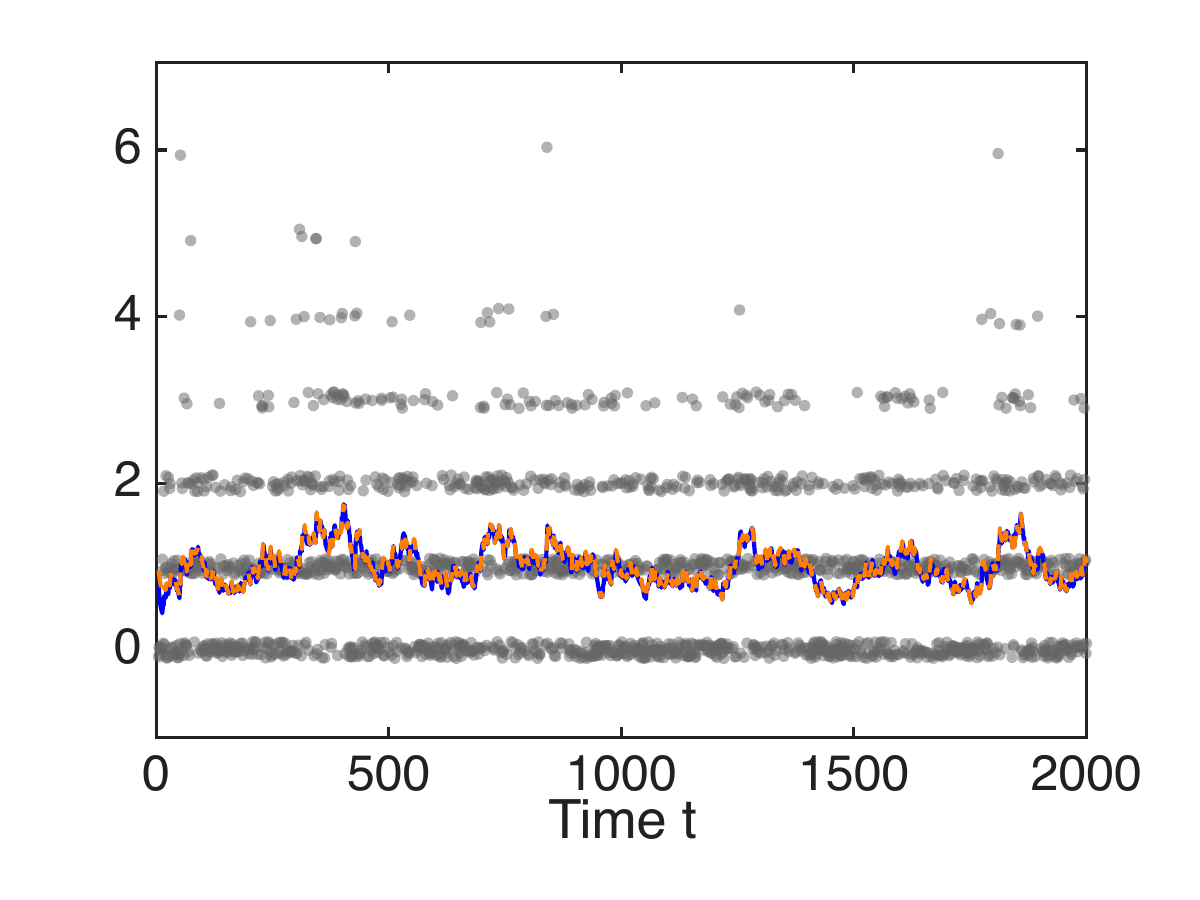} &
        \includegraphics[width=0.32\linewidth]{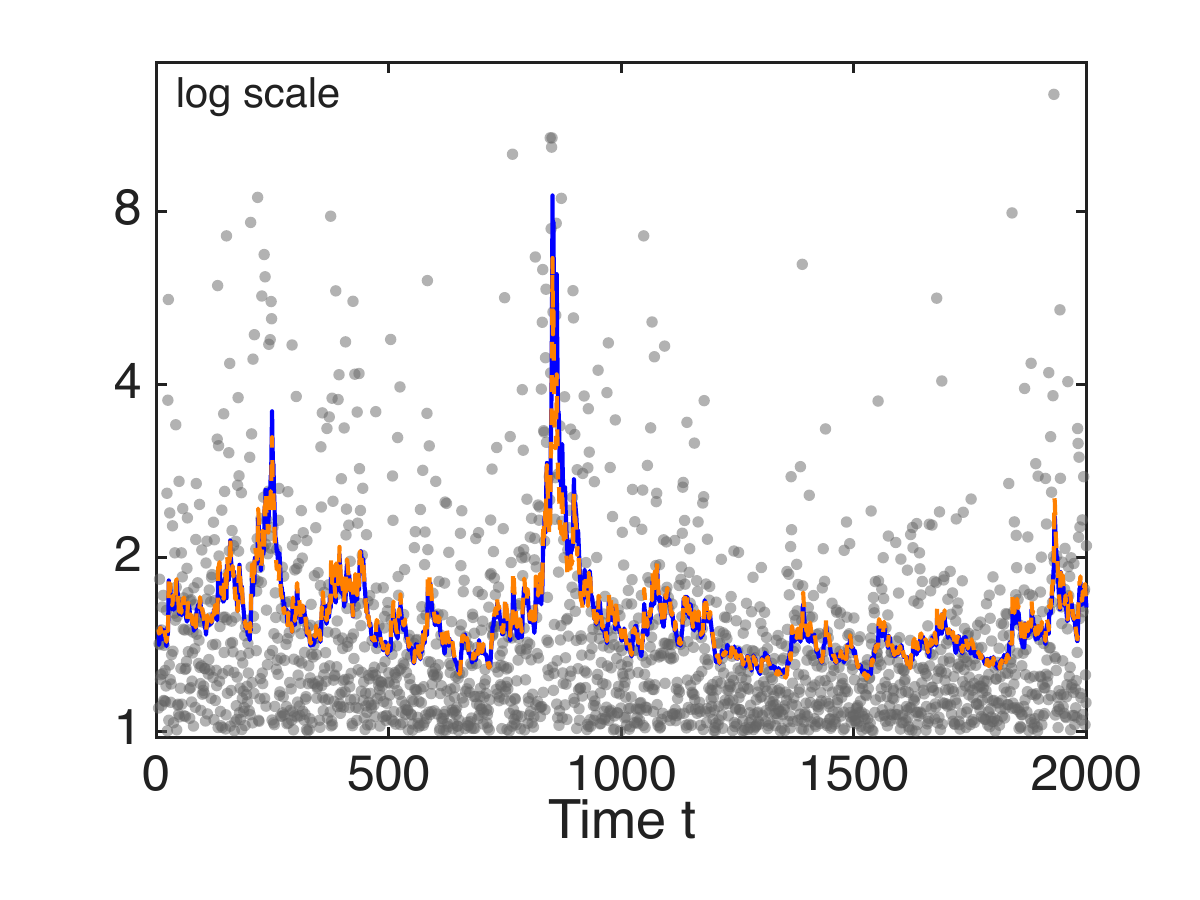} &
        \includegraphics[width=0.32\linewidth]{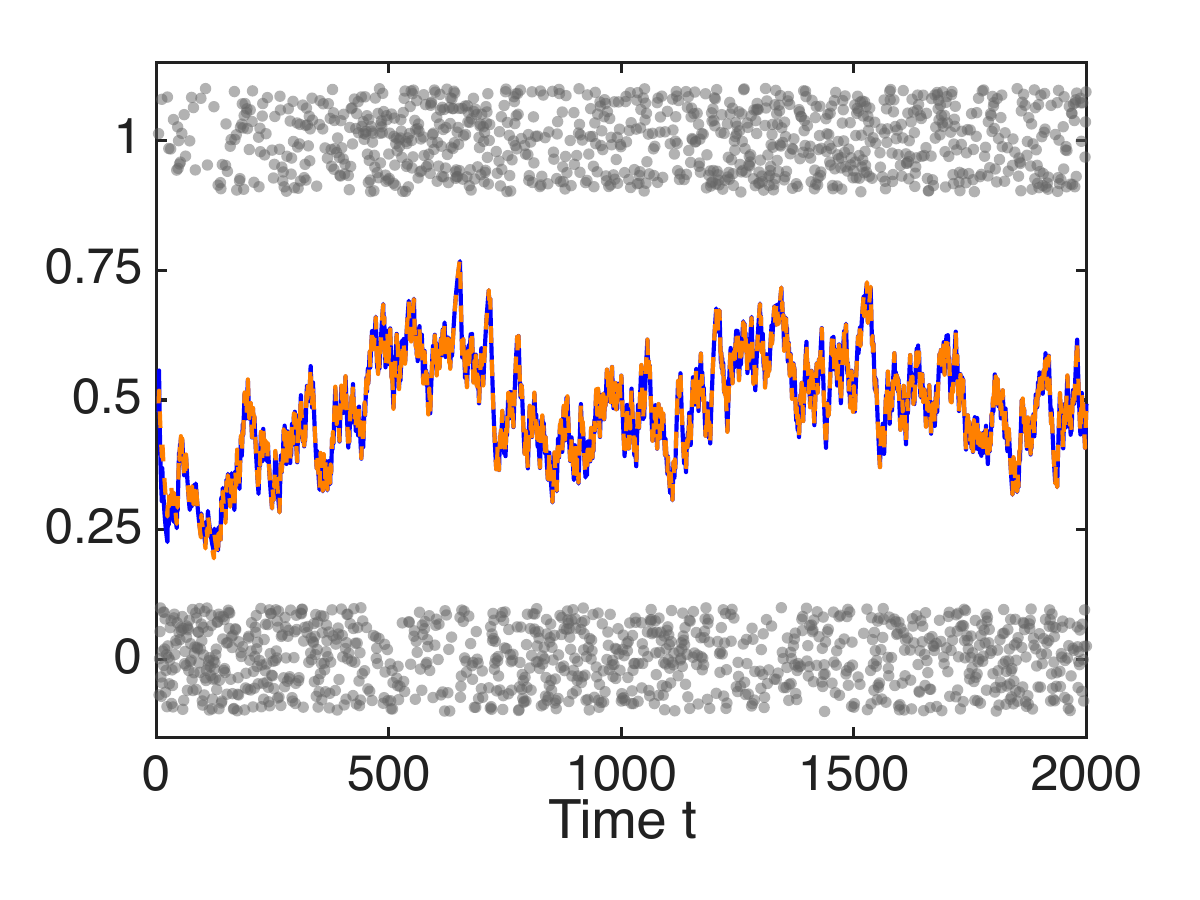} \\[-3mm]
        \raisebox{2.0cm}[0pt][0pt]{{\small $\alpha=0.95$\hspace{0.3cm}}} &
        \includegraphics[width=0.32\linewidth]{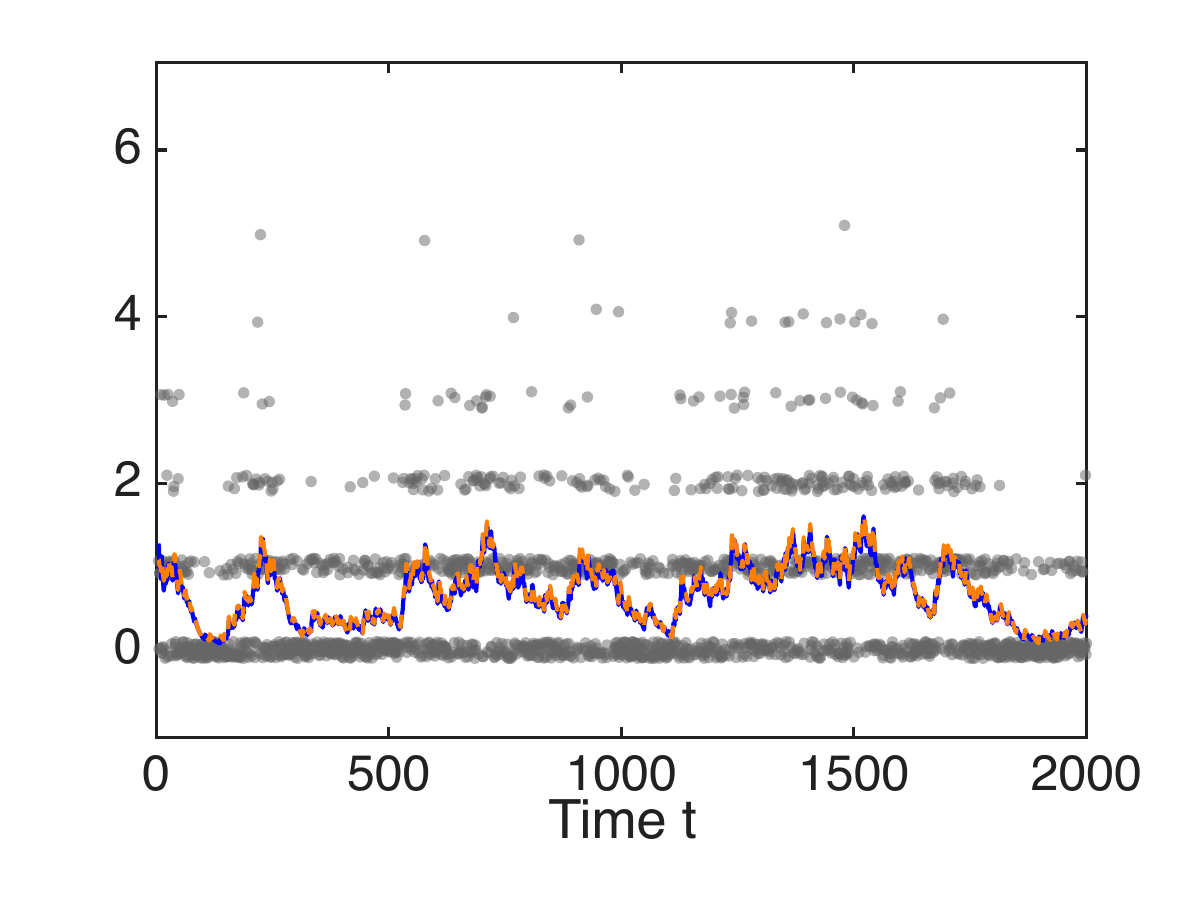} &
        \includegraphics[width=0.32\linewidth]{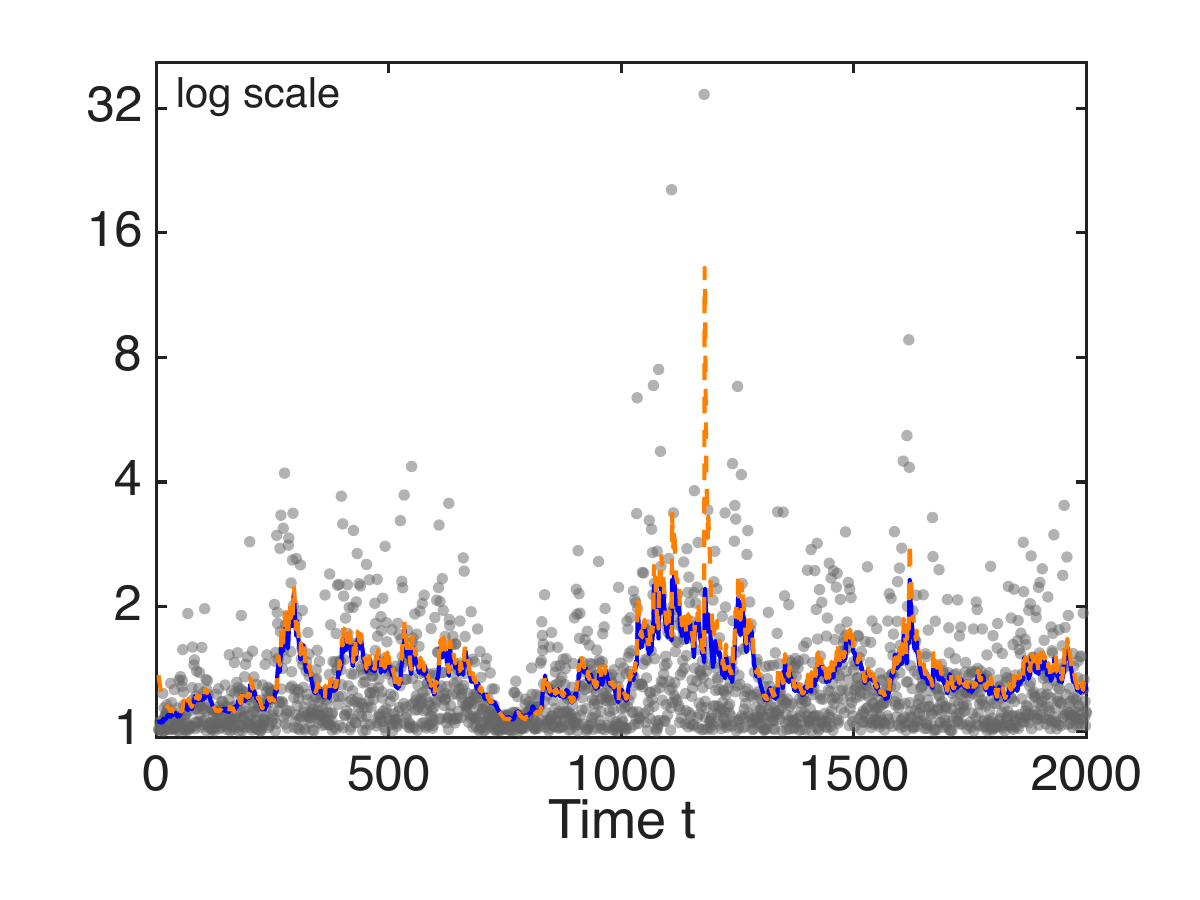} &
        \includegraphics[width=0.32\linewidth]{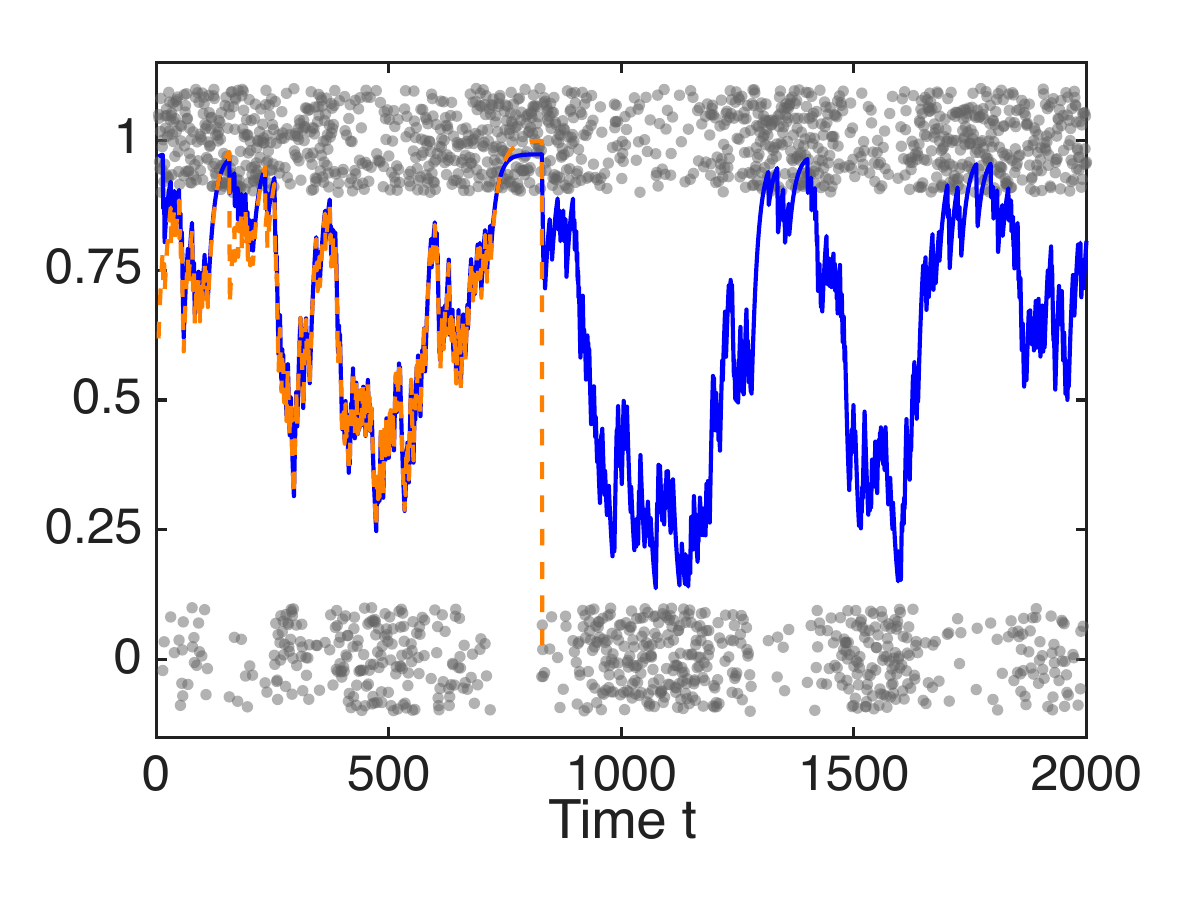}
    \end{tabular}
    \vspace{-5mm}
    \caption{Simulation from $Y_t | Y_{1:t-1} \sim \textnormal{CEF}(\widetilde{\theta}_{t|t-1},h,\psi)$ for time $t=5,...,T=2000$ with discount parameter $\lambda=0.93$. Top: anchoring parameter $\alpha=0.7$; bottom: $\alpha=0.95$. Gray circles show observations $Y_t$. Blue solid lines show the conditional mean $\widetilde{\mathbb{E}[Y_t|Y_{1:t-1}]} = \psi'(\widetilde{\theta}_{t|t-1})$ for Poisson and Bernoulli, and $\widetilde{\mathbb{E}[Y_t|Y_{1:t-1}]} = \widetilde{\theta}_{t|t-1}/(\widetilde{\theta}_{t|t-1}+1)$ for Pareto (when $\widetilde{\theta}_{t|t-1} < -1$; $y$-axis on log scale). Orange dashed lines show the score-driven predictor, computed as $\psi'(\theta^{\textnormal{sd}}_{t|t-1})$ for Poisson and Bernoulli and as $\theta^{\textnormal{sd}}_{t|t-1}/(\theta^{\textnormal{sd}}_{t|t-1}+1)$ for Pareto (when $\theta^{\textnormal{sd}}_{t|t-1} < -1$). The coefficients $(\omega, A, \Phi)$ in~\eqref{eq:SD recursion} are set via the Gaussian steady-state mapping, with $\theta^{\textnormal{sd}}_{1|0} = \widetilde{\theta}_{1|0}=(\psi')^{-1}(\mathbb{E}[h(Y_1)])$.}
    \label{fig:sd_comparison}
\end{figure}

\section{Monte Carlo assessment of estimator precision}\label{app:montecarlo}

To investigate the finite-sample precision of the quasi-likelihood estimator from Definition~\ref{def:2-step estimator}, we conduct a Monte Carlo study using the estimated model from Section~\ref{sec:household} as the data-generating process (DGP). We generate time-varying proportion vectors following Assumption~\ref{def:DGP4}, that is,
$
Y_t \mid Y_{1:t-1} \sim {\tt Dirichlet}(\theta_{t|t-1}),$ where $ \theta_{t|t-1} = (\theta_{1,t|t-1}, \ldots, \theta_{7,t|t-1})^{\top},
$
where $\theta_{t|t-1}$ evolves according to the exponentially weighted predictor (Example~\ref{thm:MLE}), denoting $(\pi^*, \alpha^*, \lambda^*)$ for the hyperparameters in the DGP. We initialize each prediction path with $\widetilde{\theta}_{1|0} = (\psi')^{-1}(\pi^*)$ where $\pi^* = \widehat{\mathbb{E}[h(Y_1)]}$ is the estimated centering parameter from Section~\ref{sec:household}.

To investigate how $\alpha^*$ affects hyperparameter estimation, we consider four scenarios:
\begin{enumerate}[label=(\roman*), itemsep=0pt, parsep=0pt]
\item $(\pi^*, \alpha^*, \lambda^*) = (\widehat{\pi}, 0.01, 0.64)$ (strong anchoring),
\item $(\pi^*, \alpha^*, \lambda^*) = (\widehat{\pi}, 0.30, 0.64)$ (moderate anchoring),
\item $(\pi^*, \alpha^*, \lambda^*) = (\widehat{\pi}, 0.60, 0.64)$ (weak anchoring),
\item $(\pi^*, \alpha^*, \lambda^*) = (\widehat{\pi}, 0.95, 0.64)$ (no anchoring, matching estimated values),
\end{enumerate}
where $\pi^{*} = (-1.76, -1.41, -1.78, -1.77, -2.73, -2.23, -3.53)'$ and $\lambda^* = 0.64$ is fixed across all scenarios to match the estimates from the household data.

Figure~\ref{fig:household_simulation} displays the quasi log-likelihood surfaces and confidence regions for the estimated $(\alpha, \lambda)$. Confidence regions contract systematically as $T$ increases, with maximum quasi-likelihood estimates (red stars) concentrating near the DGP values (blue circles). The pattern of contraction depends strongly on the anchoring hyperparameter $\alpha^*$. When $\alpha^*=0.01$ (first row), the confidence interval for $\lambda$ remains wide even at $T=10{,}000$. This aligns with Remark~\ref{remark 1}(f): when $\alpha$ is close to zero, the frame is nearly stationary, and $\lambda$ has minimal impact on the filter and predictor, making it difficult to estimate from the data. Contrastingly, when $\alpha^* \in \{0.60, 0.95\}$ (third and fourth rows), confidence intervals for both hyperparameters contract rapidly even at moderate sample sizes.

\begin{figure}[htbp]
\centering
\begin{tabular}{@{}c@{\hspace{-0.2cm}}c@{\hspace{0cm}}c@{\hspace{0cm}}c@{\hspace{0cm}}c@{}}
    & \textbf{\boldmath $T=100$} & \textbf{\boldmath $T=1000$} & \textbf{\boldmath $T=5000$} & \textbf{\boldmath $T=10000$} \\
    \raisebox{1.72cm}[0pt][0pt]{{\small $\alpha^*=0.01$\hspace{0.3cm}}} &
    \includegraphics[width=0.23\linewidth]{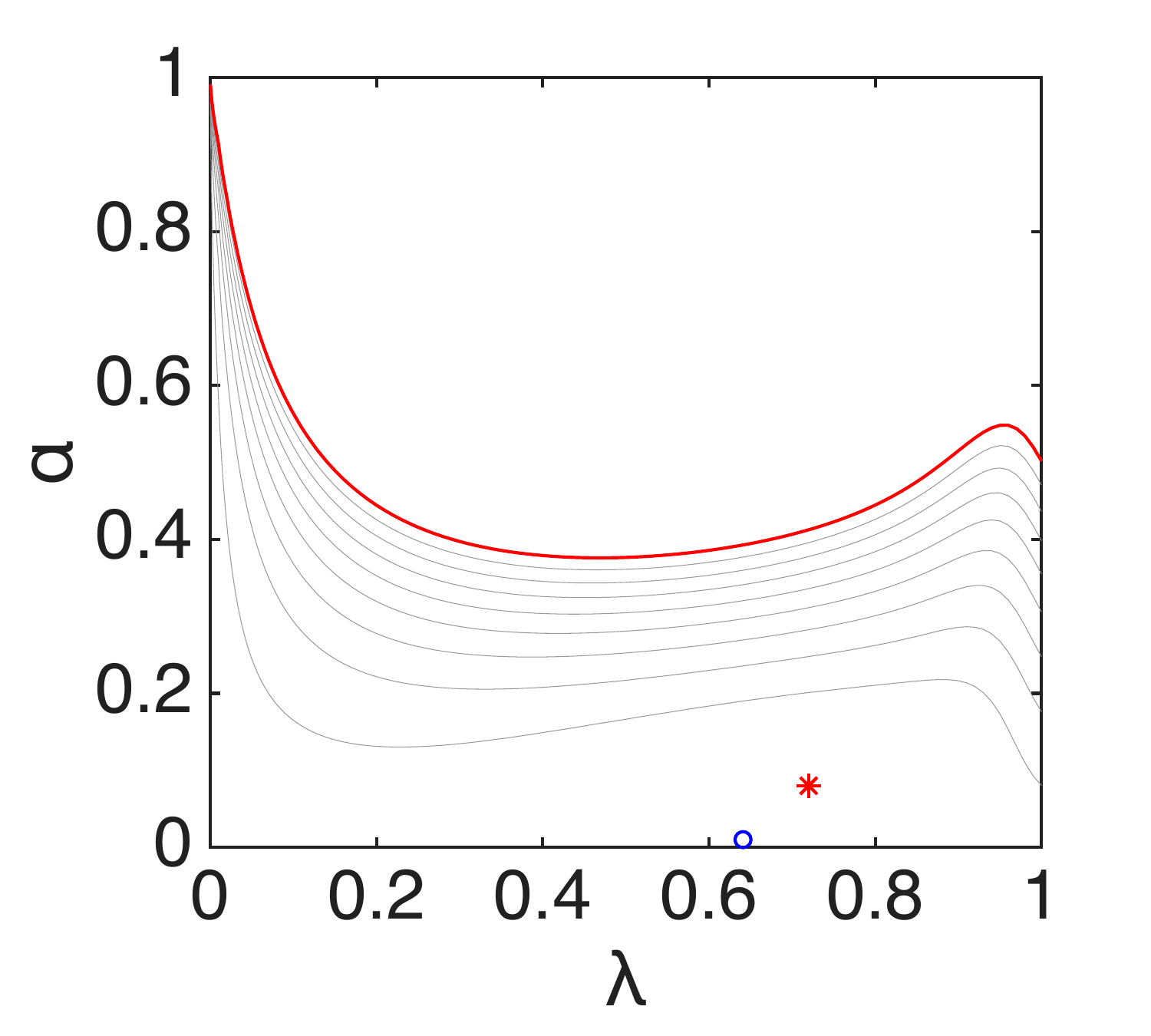} &
    \includegraphics[width=0.23\linewidth]{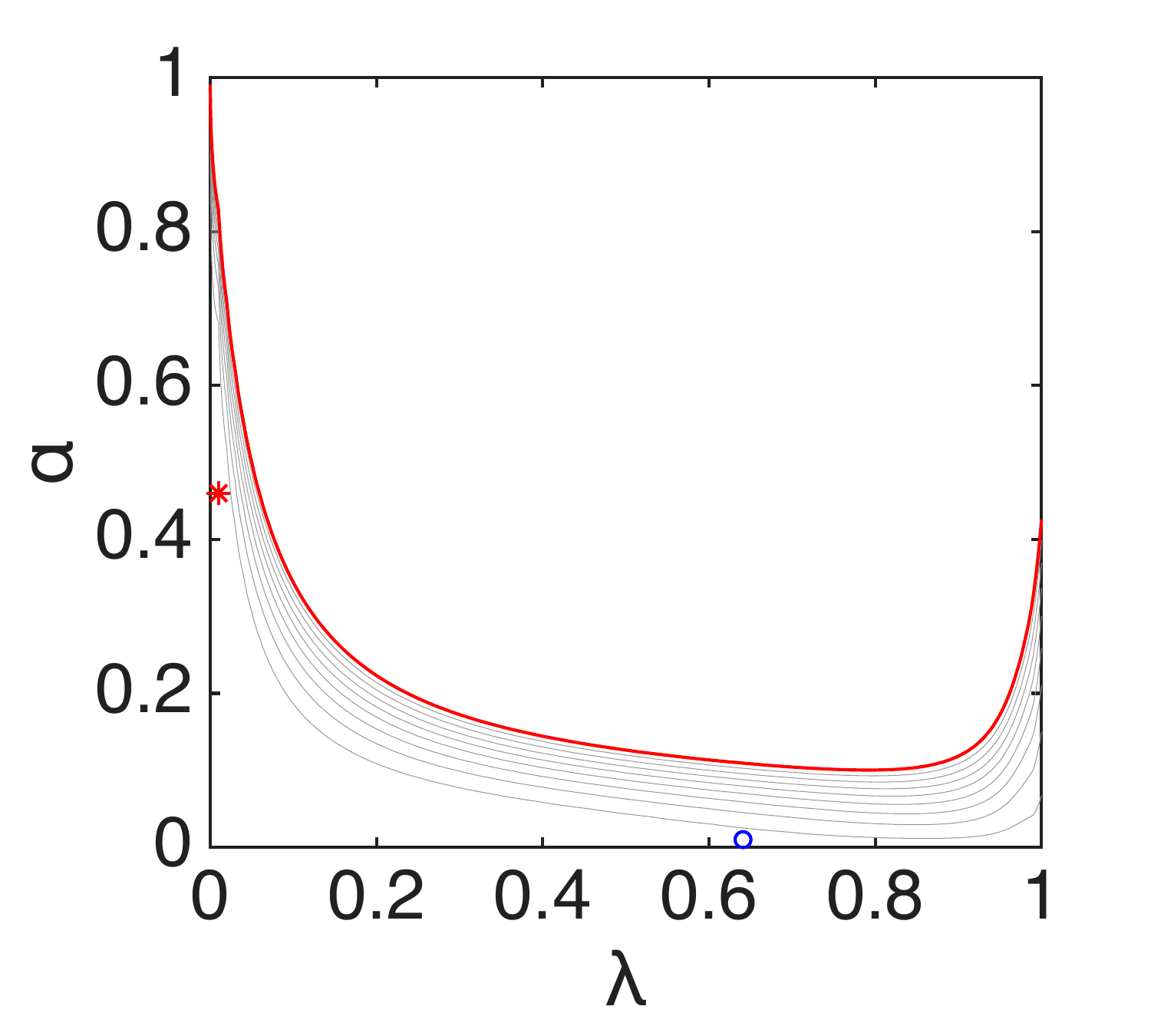} &
    \includegraphics[width=0.23\linewidth]{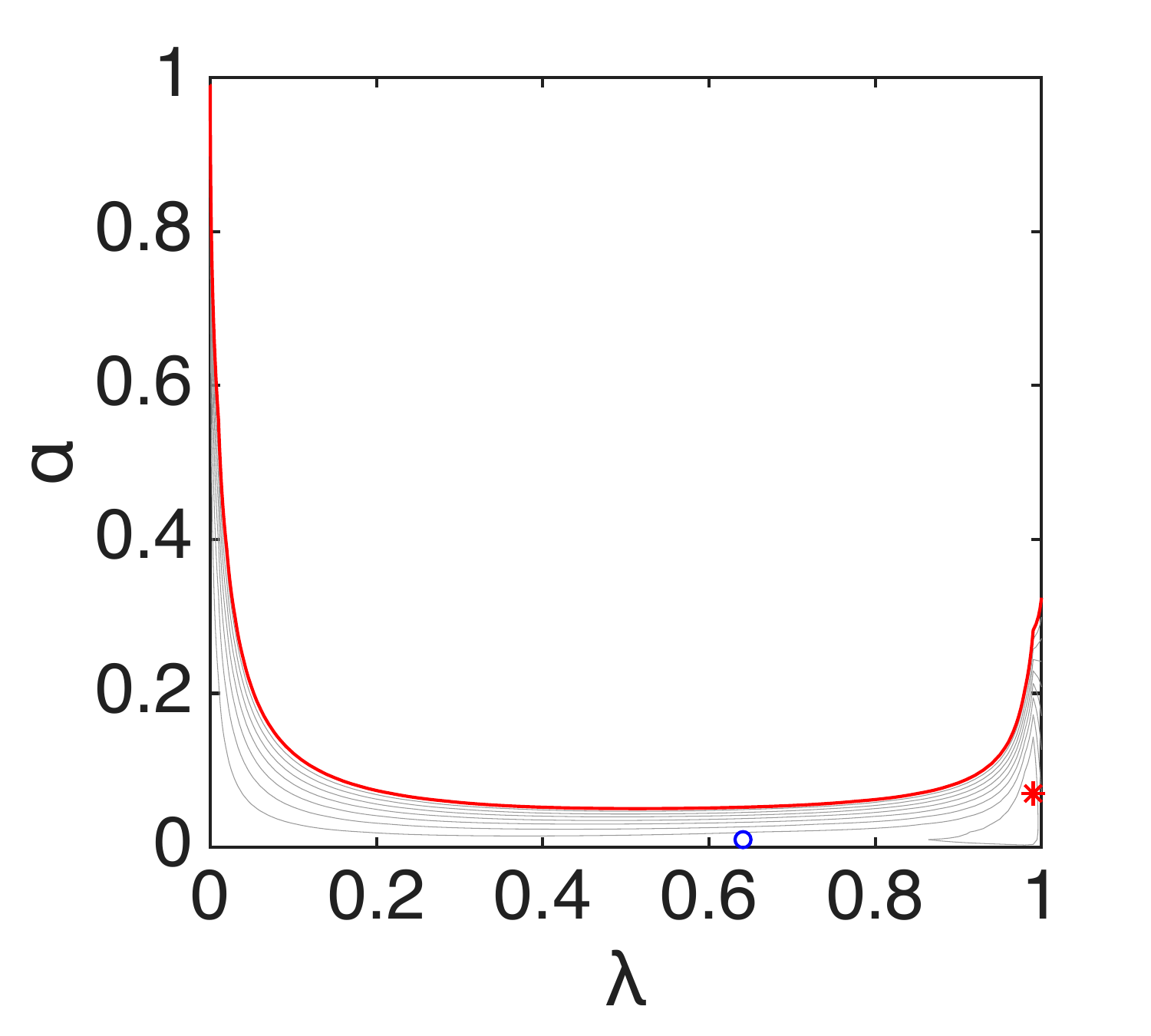} &
    \includegraphics[width=0.23\linewidth]{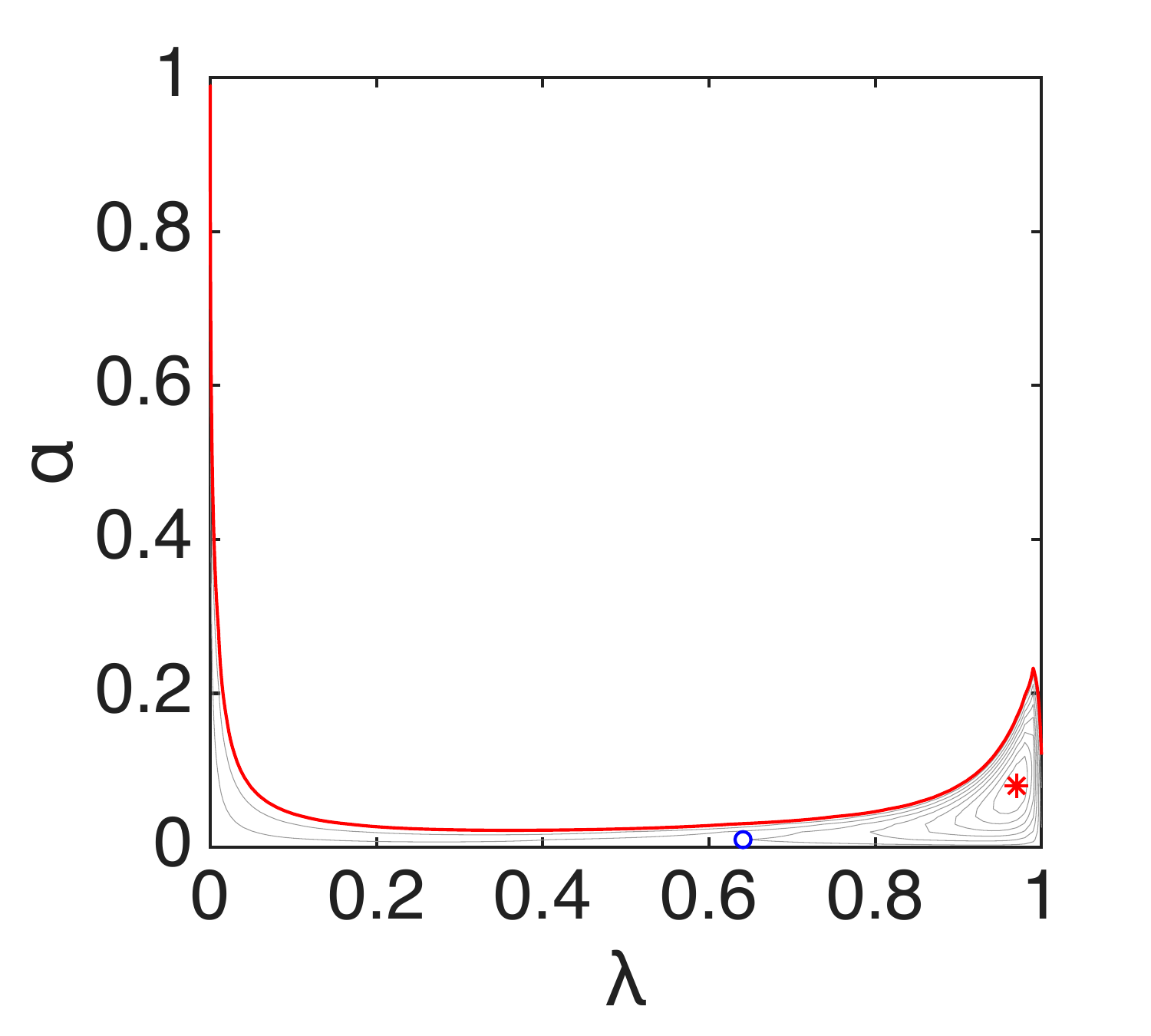}\\[-1mm]
    \raisebox{1.72cm}[0pt][0pt]{{\small $\alpha^*=0.30$\hspace{0.3cm}}} &
    \includegraphics[width=0.23\linewidth]{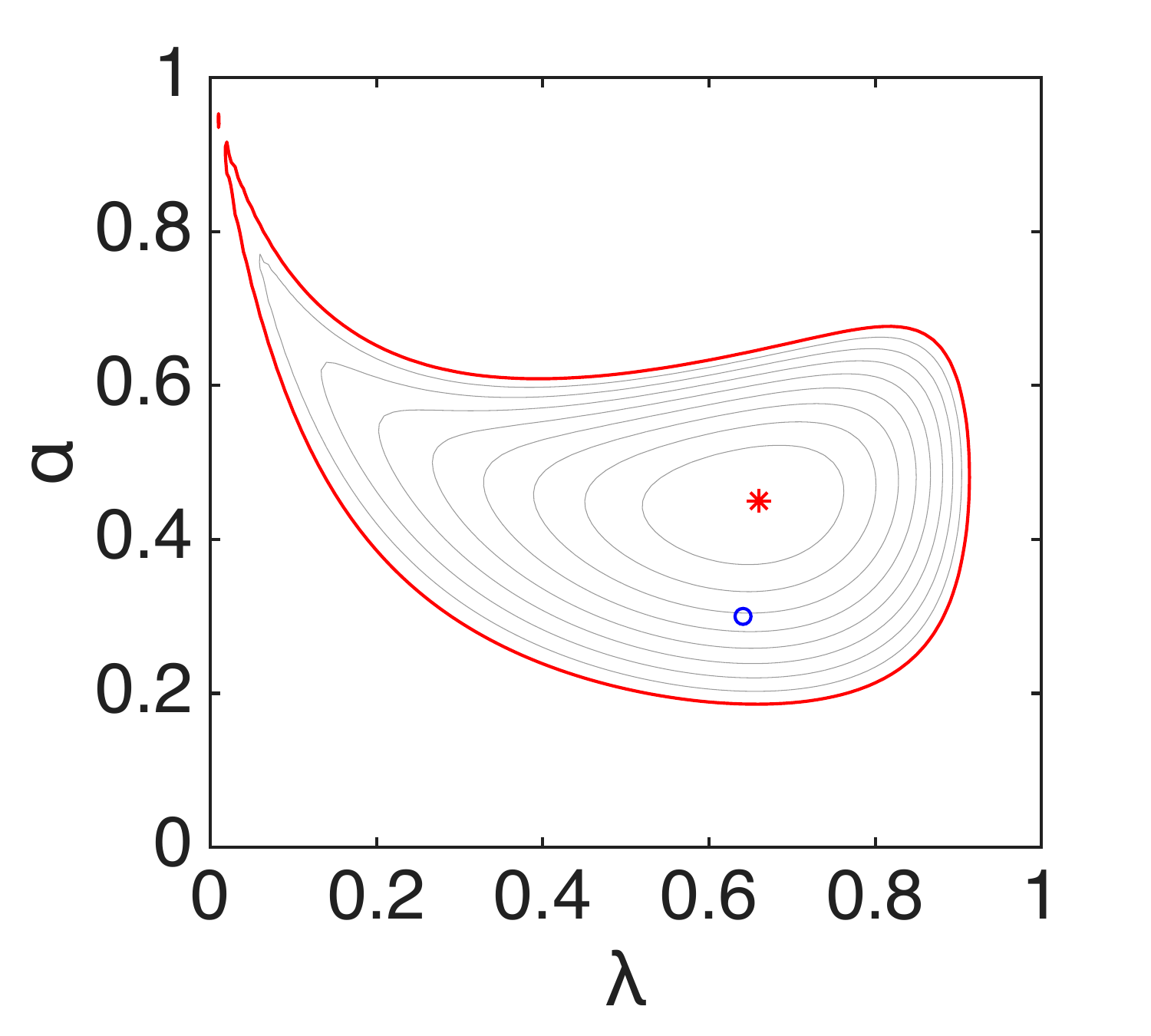} &
    \includegraphics[width=0.23\linewidth]{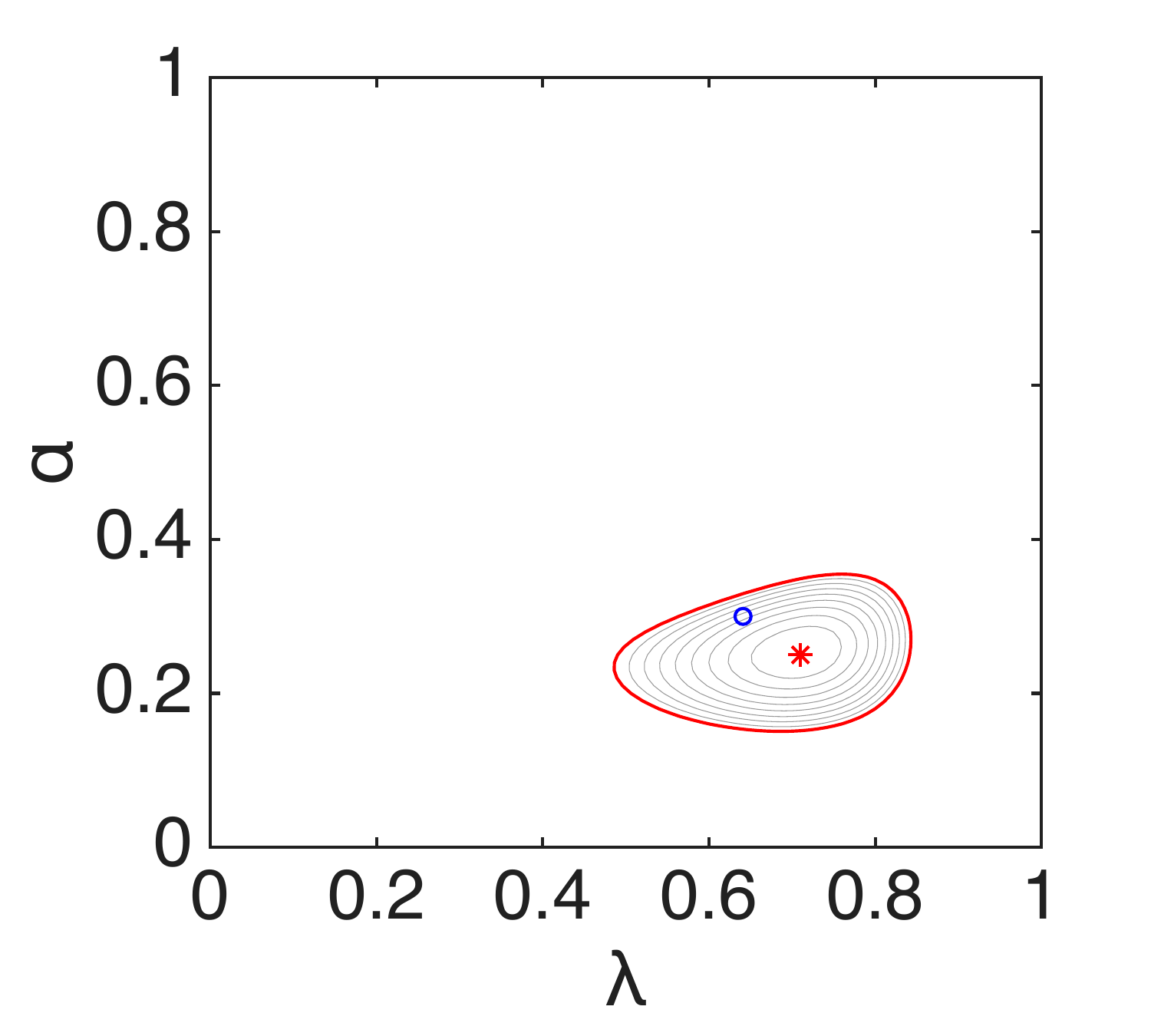} &
    \includegraphics[width=0.23\linewidth]{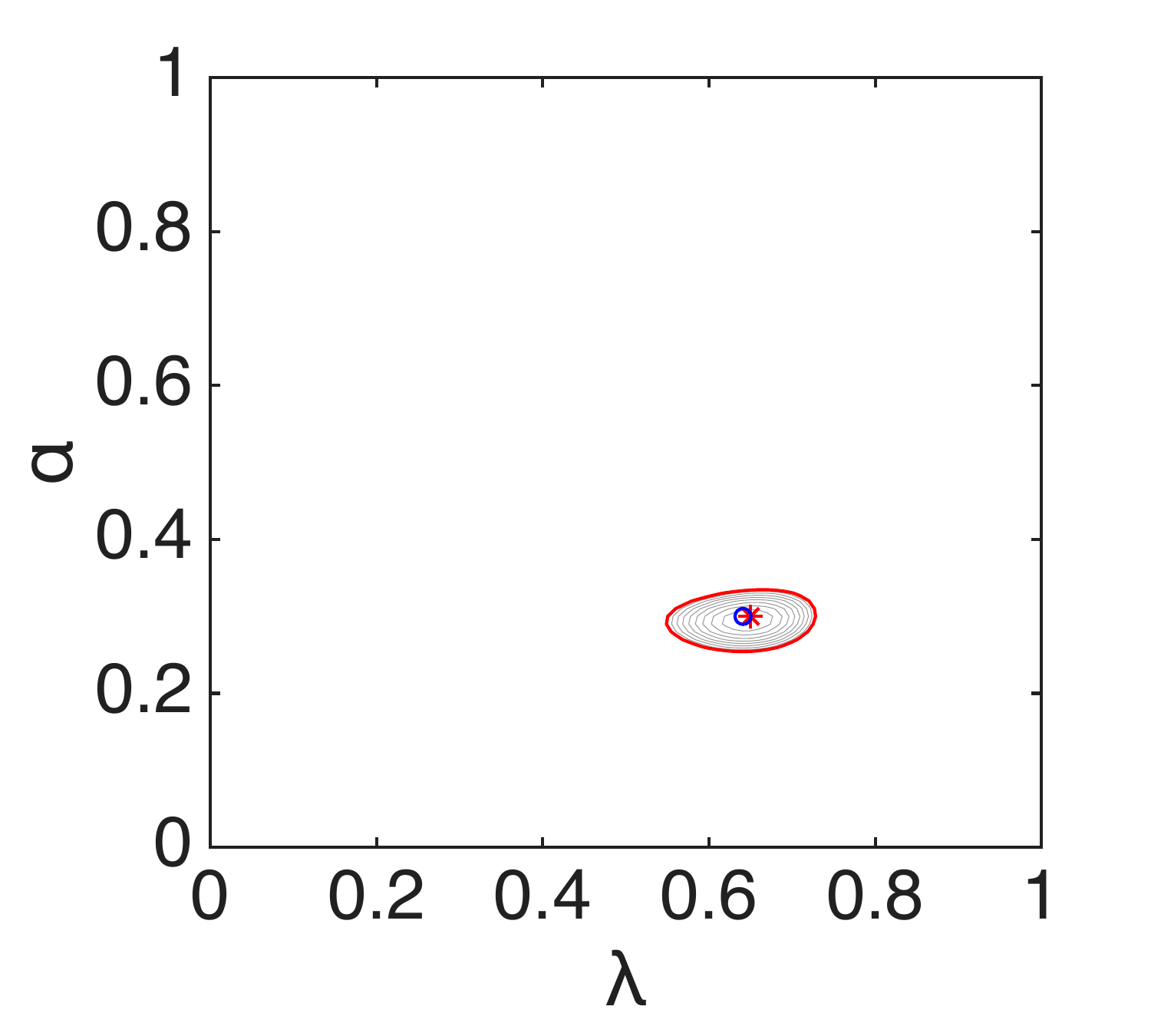} &
    \includegraphics[width=0.23\linewidth]{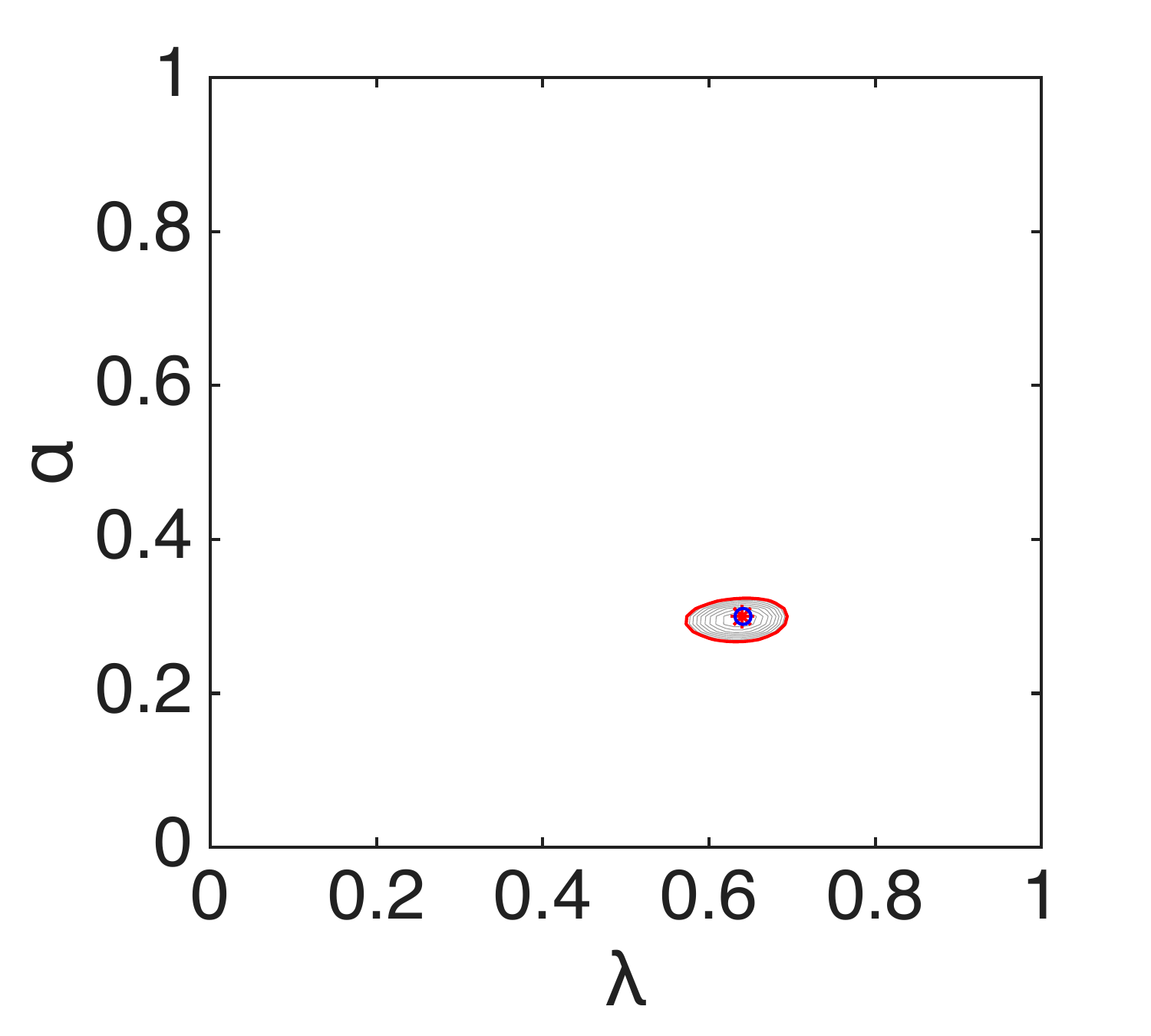}  \\[-1mm]
    \raisebox{1.72cm}[0pt][0pt]{{\small $\alpha^*=0.60$\hspace{0.3cm}}} &
    \includegraphics[width=0.23\linewidth]{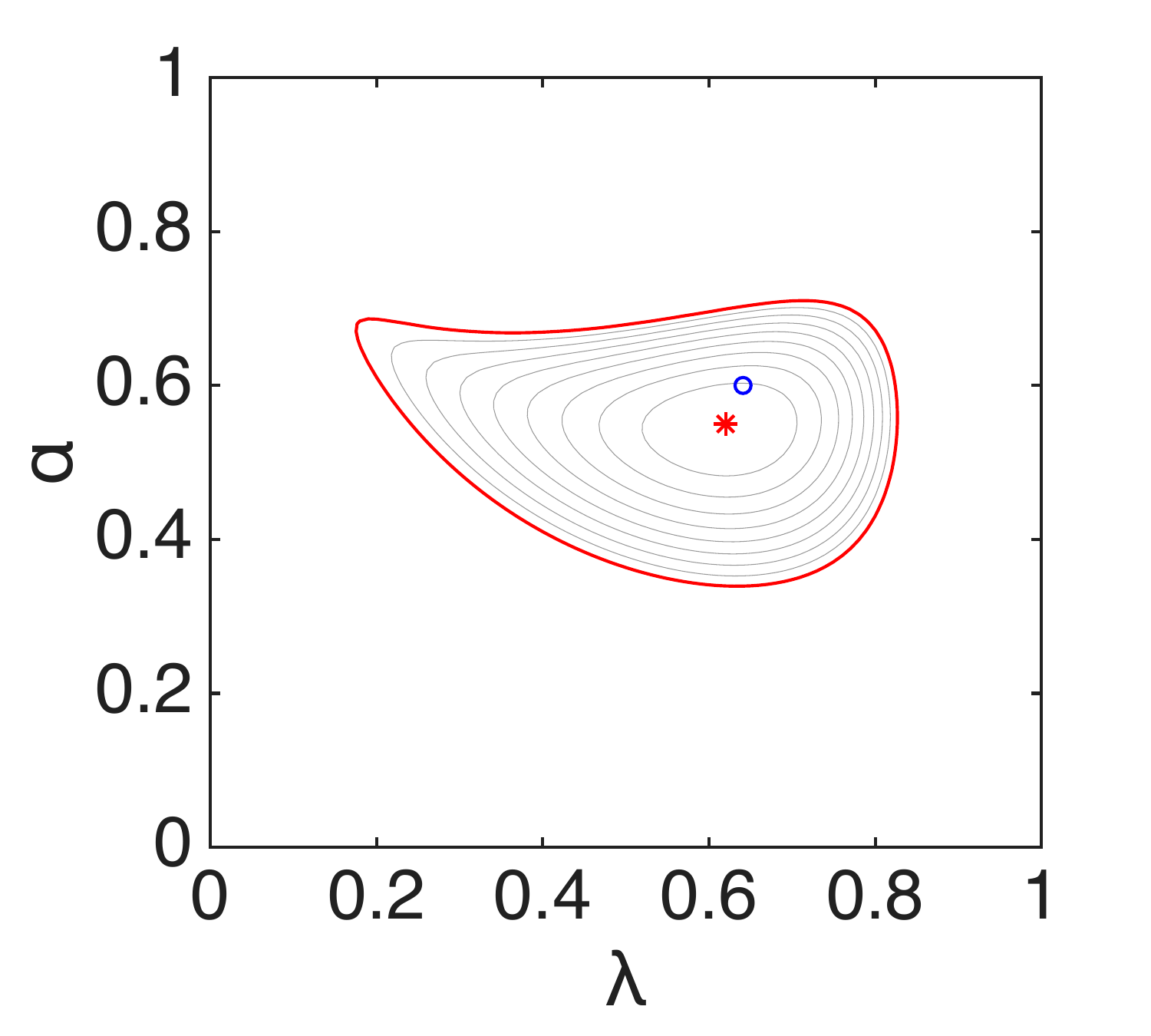} &
    \includegraphics[width=0.23\linewidth]{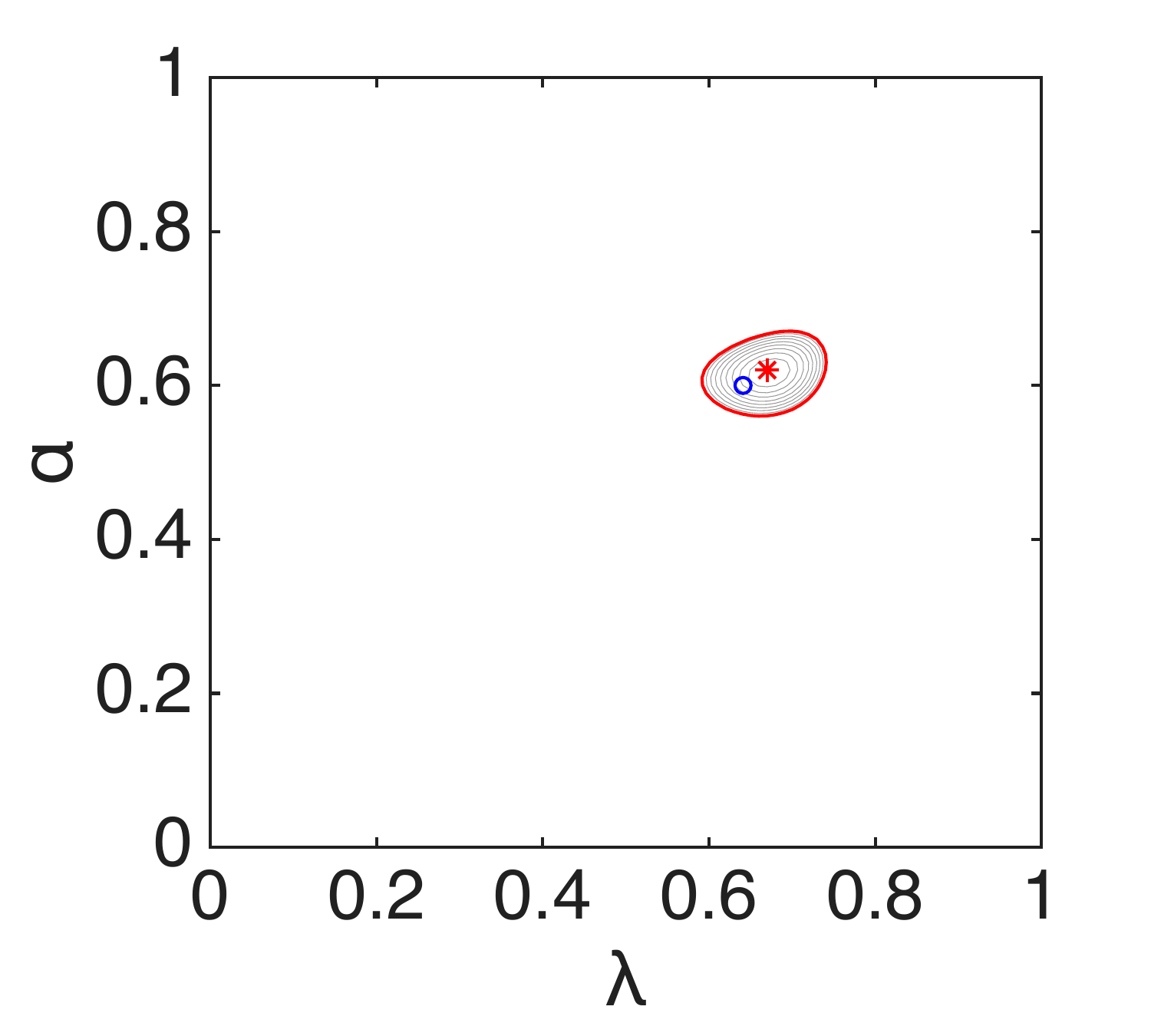} &
    \includegraphics[width=0.23\linewidth]{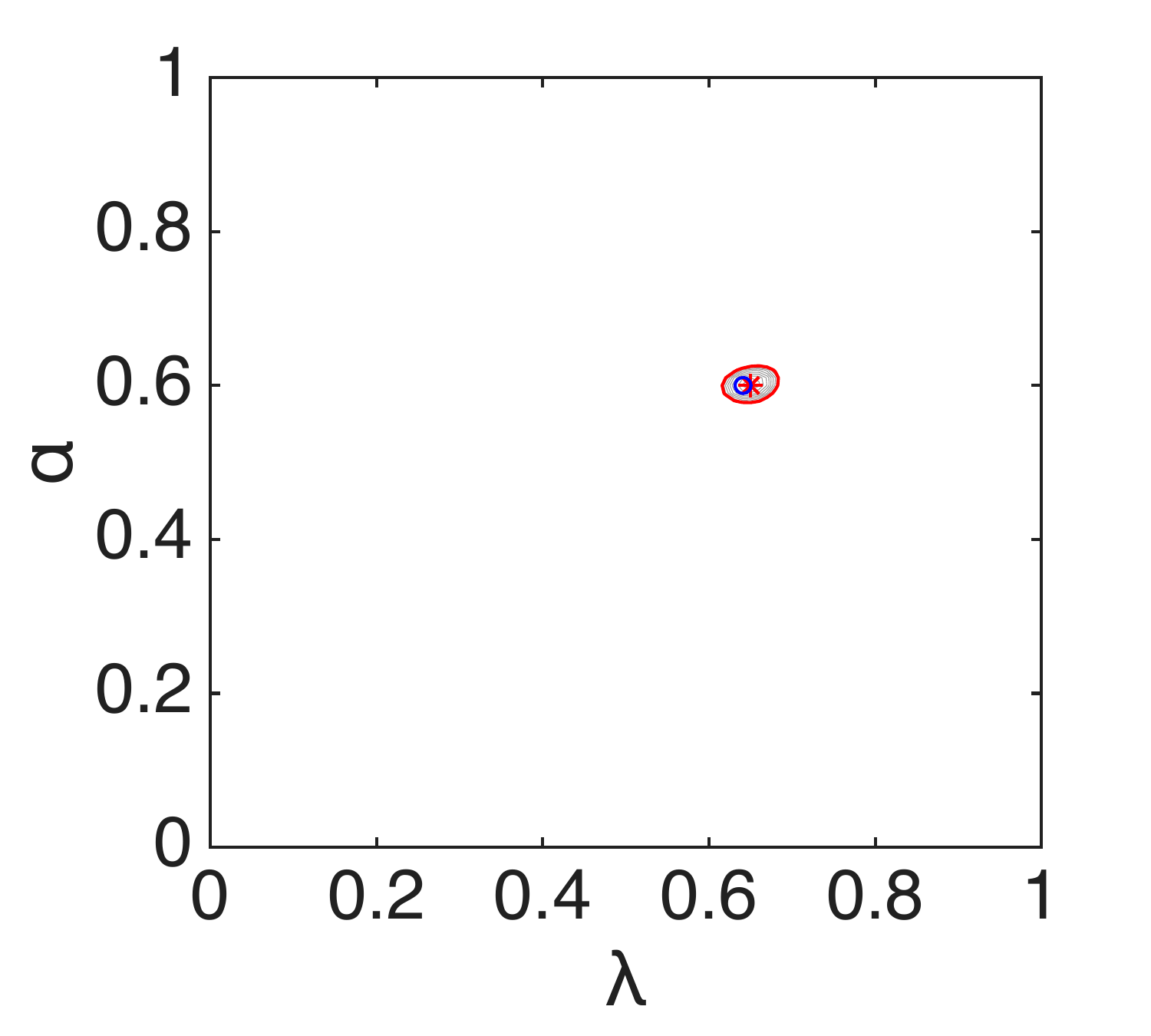} &
    \includegraphics[width=0.23\linewidth]{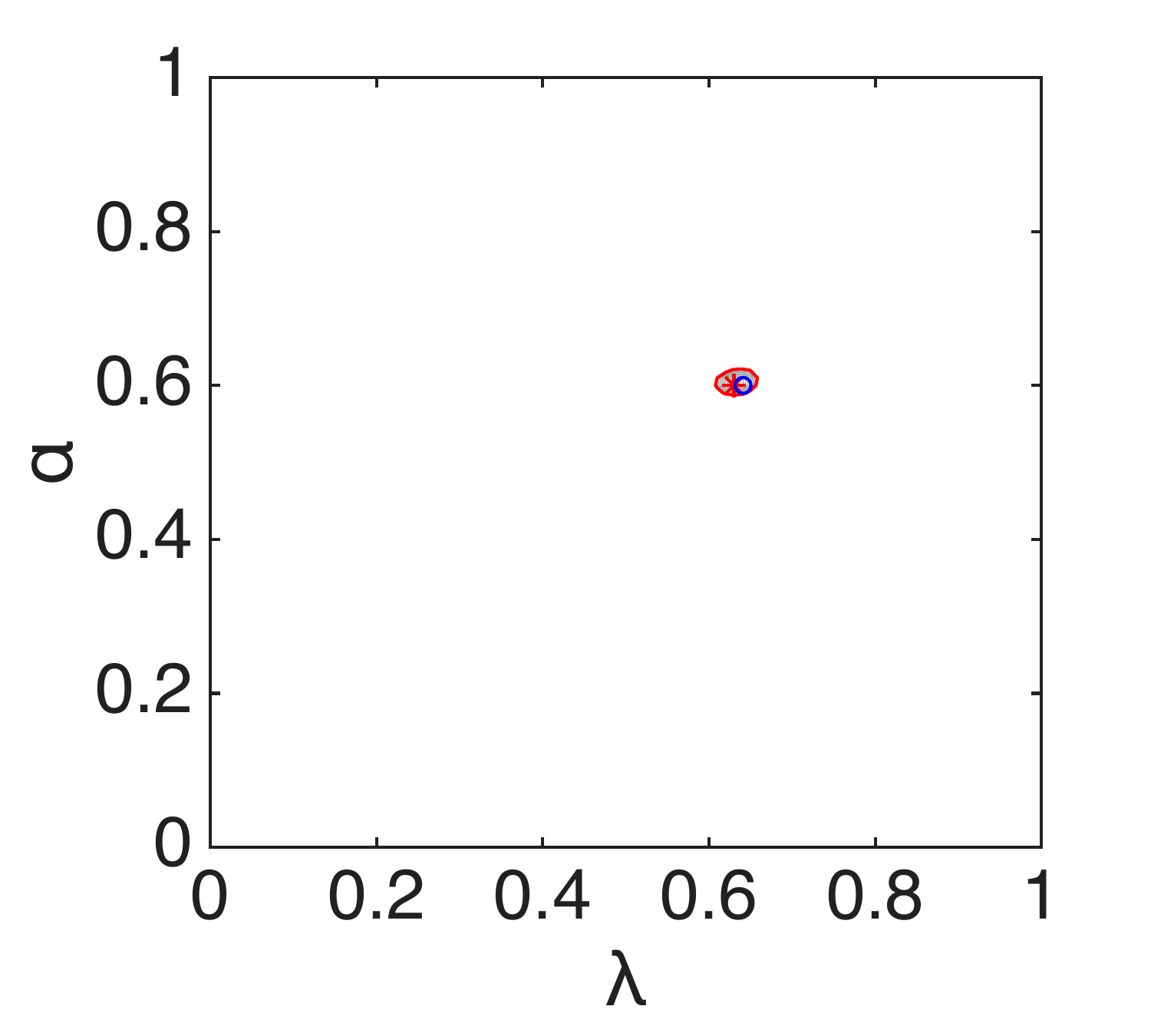}\\[-1mm]
    \raisebox{1.72cm}[0pt][0pt]{{\small $\alpha^*=0.95$\hspace{0.3cm}}} &
    \includegraphics[width=0.23\linewidth]{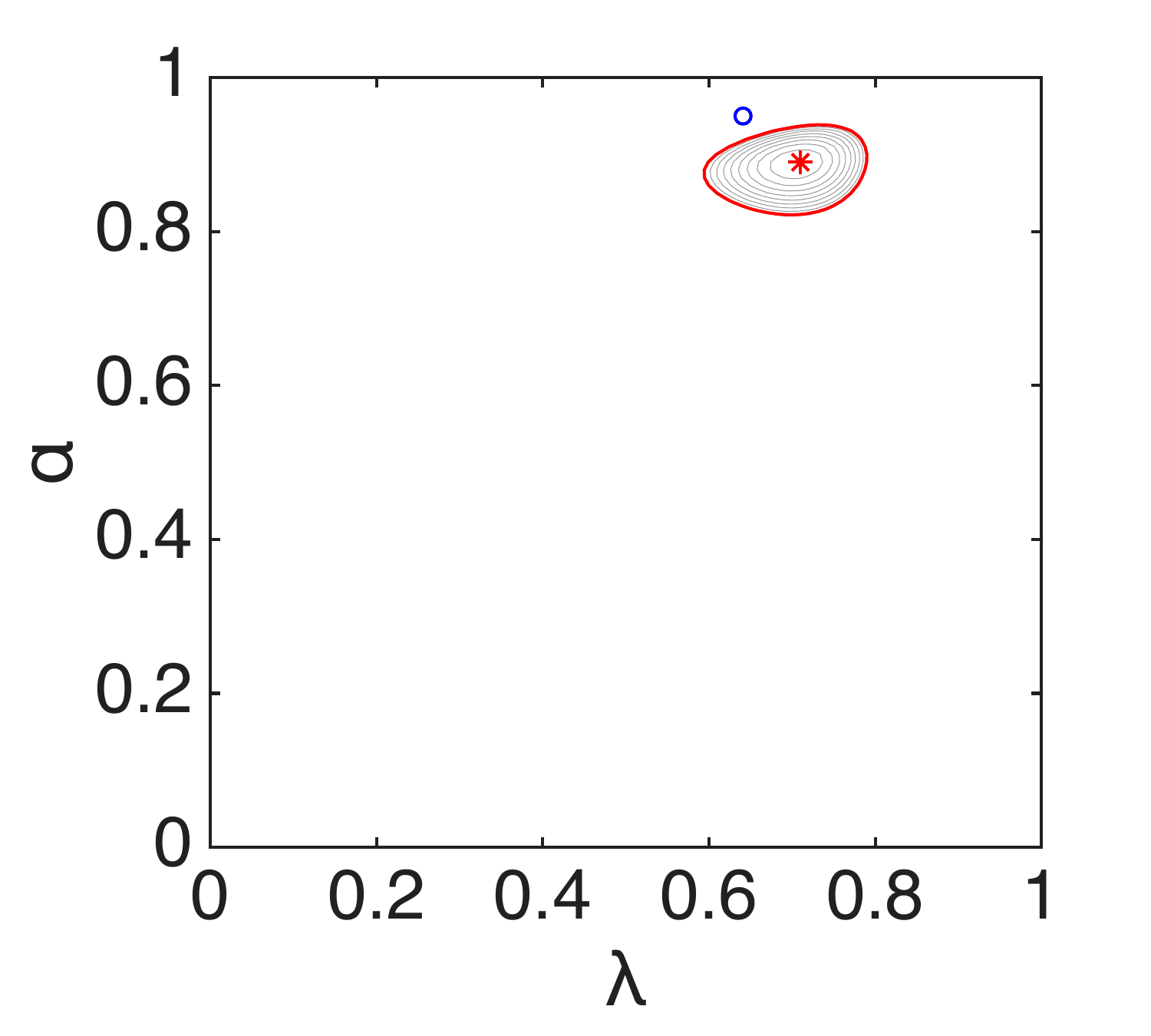} &
    \includegraphics[width=0.23\linewidth]{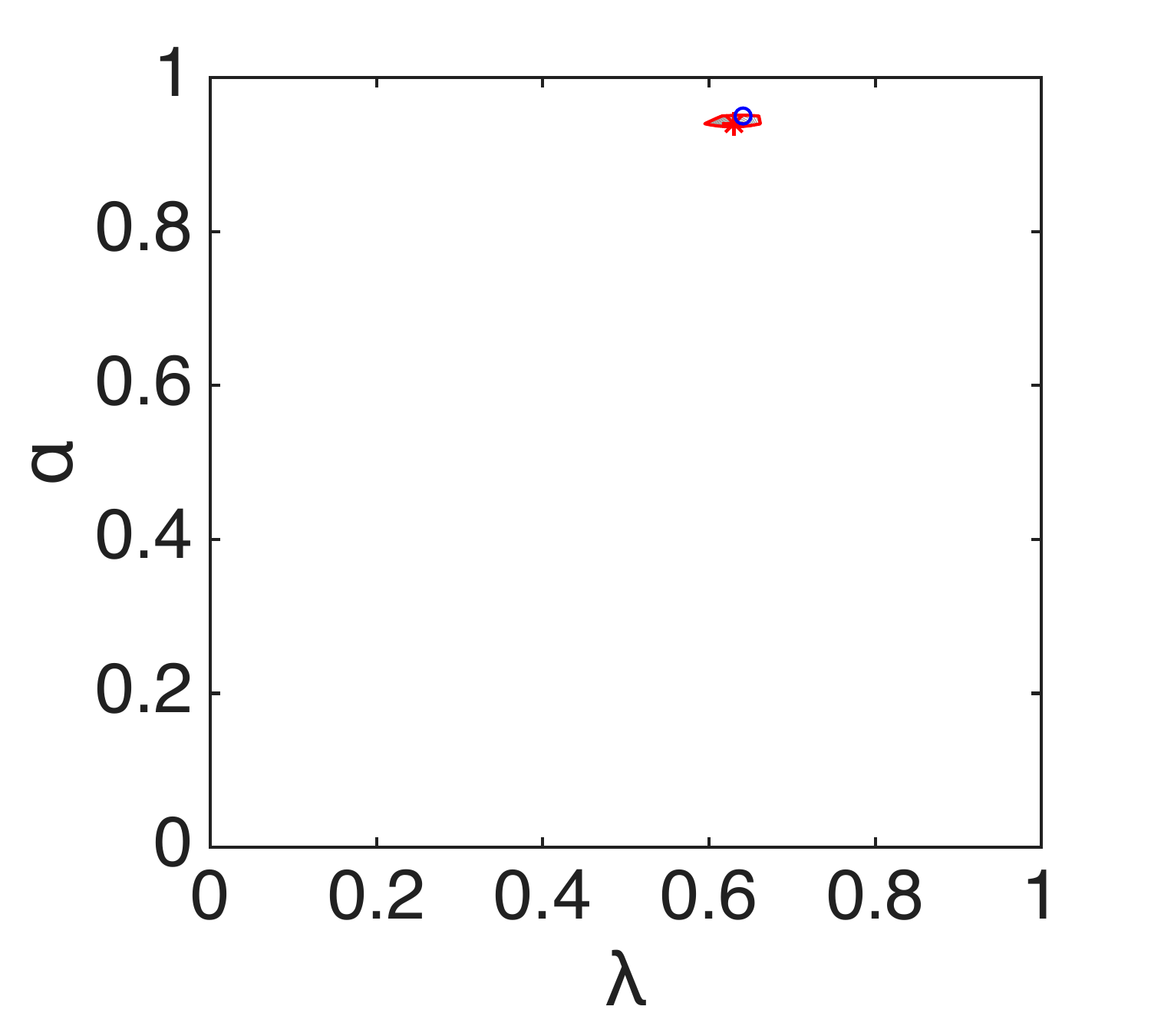} &
    \includegraphics[width=0.23\linewidth]{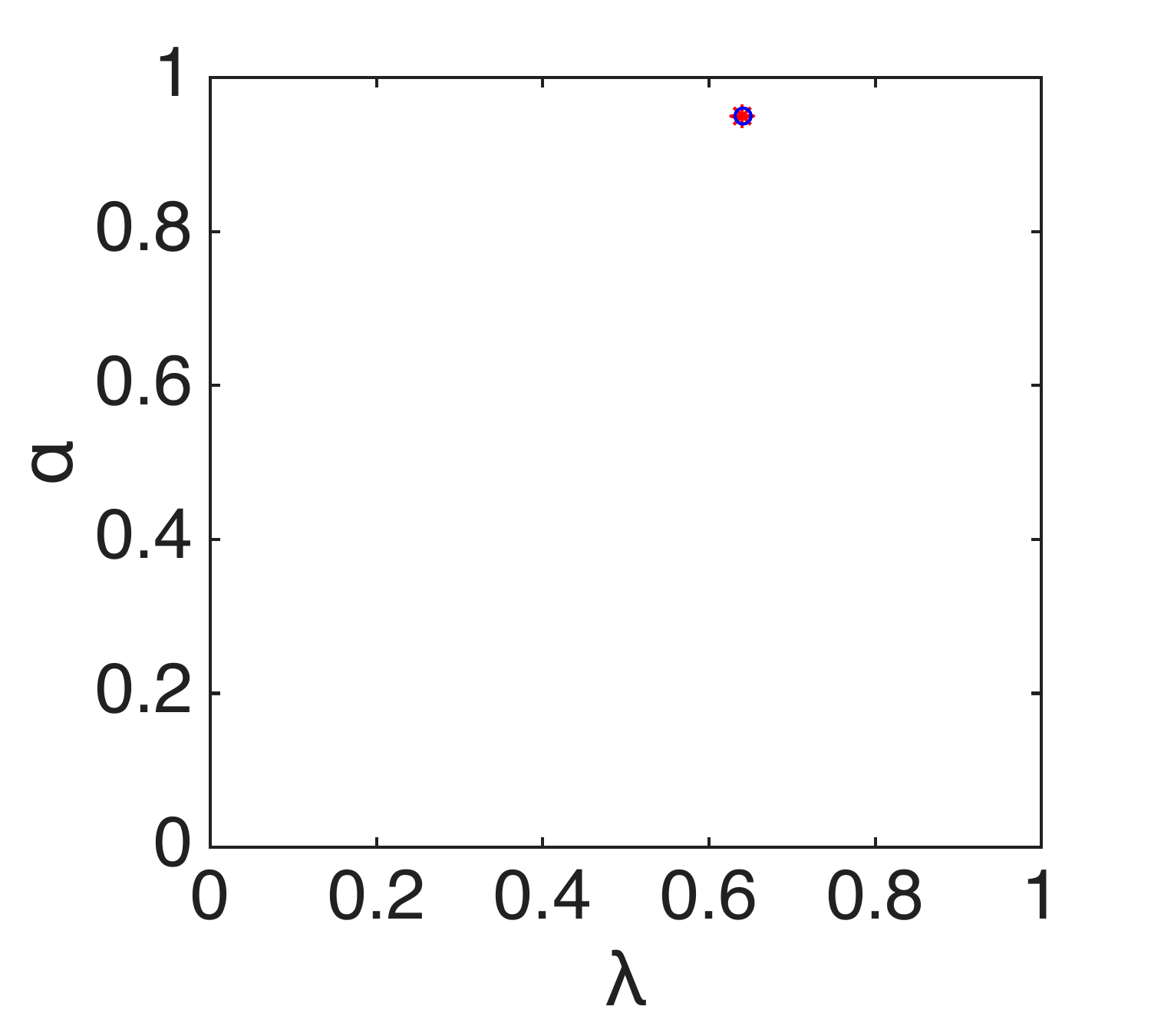} &
    \includegraphics[width=0.23\linewidth]{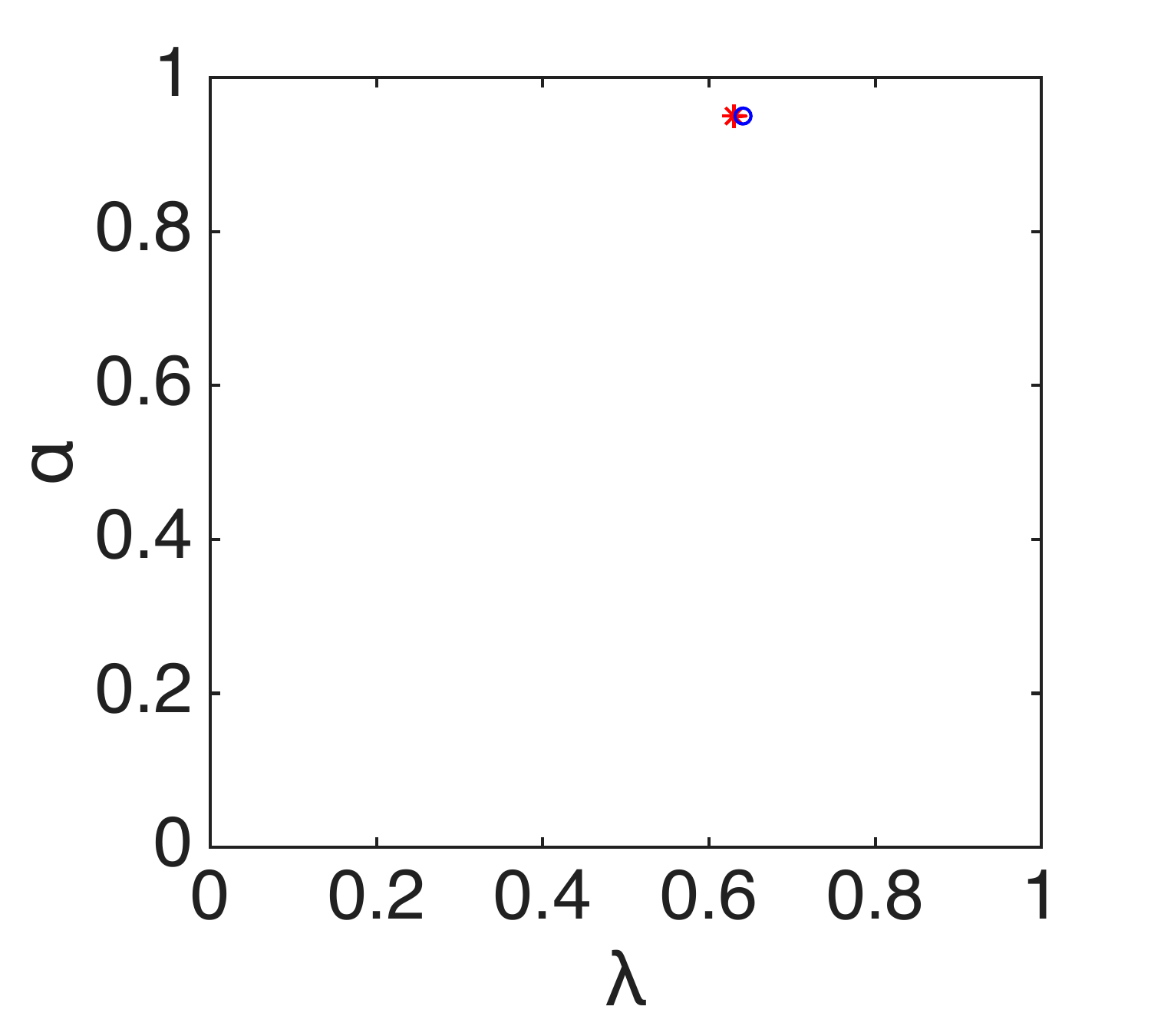}
\end{tabular}
\vspace{-4mm}
\caption{Monte Carlo: quasi-likelihood estimator precision for household financial expectations. Rows: anchoring hyperparameters: (i) $\alpha^* = 0.01$, (ii) $\alpha^* = 0.30$, (iii) $\alpha^* = 0.60$, and (iv) $\alpha^* = 0.95$ (matching estimated household data). All scenarios use the estimated centering parameter $\pi^* = \widehat{\mathbb{E}[h(Y_1)]} = (-1.76, -1.41, -1.78, -1.77, -2.73, -2.23, -3.53)'$ and discount parameter $\lambda^* = 0.64$ from Section~\ref{sec:household}. Display: quasi log-likelihood surface $\ell(\alpha, \lambda)$ with the maximum quasi-likelihood estimate (red star) and DGP value (blue circle). The red contour traces the 99\% confidence region. Columns: $T \in \{100, 1000, 5000, 10000\}$.}
\label{fig:household_simulation}
\end{figure}

\end{document}